\documentclass[12pt]{article}
\usepackage{fullpage, parskip}
\usepackage{algpseudocode}
\usepackage{amsmath, amssymb, amsthm}
\usepackage{float, graphicx, multirow,subcaption, setspace}
\usepackage{comment}
\usepackage{url}
\usepackage{enumitem}
\usepackage{tikz}
\usepackage{bm, bbm}
\usetikzlibrary{shapes.geometric, positioning}
\usetikzlibrary{quotes, angles}
\usepackage{rotating}
\usepackage{xr, xr-hyper}
\usepackage{hyperref}

\usepackage{natbib}
\usepackage{placeins}
\usepackage{algorithm2e}
\RestyleAlgo{ruled}

\definecolor{SkyBlue}{RGB}{14, 118, 188}
\definecolor{BrightRed}{RGB}{223,82, 78}

\DeclareMathOperator{\eig}{eig}
\DeclareMathOperator{\tr}{tr}
\DeclareMathOperator{\sign}{sign}
\DeclareMathOperator{\pen}{pen}

\DeclareMathOperator{\vect}{vec}

\DeclareMathOperator*{\argmax}{arg\,max}
\DeclareMathOperator*{\argmin}{arg\,min}

\newtheorem{myTheorem}{Theorem}

\newtheorem{myLemma}{Lemma}
\newtheorem{myCorollary}{Corollary}
\newtheorem{myRemark}{Remark}
\newtheorem{myDefinition}{Definition}

\newcommand{\bY}{\bm{Y}}
\newcommand{\by}{\bm{y}}

\newcommand{\bx}{\bm{x}}
\def\bw{\bm{w}}

\newcommand{\N}{\mathcal{N}}
\def\P{\mathbb{P}}
\newcommand{\E}{\mathbb{E}}

\newcommand{\bdelta}{\boldsymbol{\delta}}

\title{Estimating sparse direct effects in multivariate regression with the spike-and-slab LASSO}
\author{Yunyi Shen\thanks{Laboratory for Information \& Decision Systems, Massachusetts Institute of Technology. The work was done while the author was at the University of Wisconsin--Madison.}, Claudia Sol\'{i}s-Lemus\thanks{Wisconsin Institute for Discovery \& Dept.~of Plant Pathology, University of Wisconsin--Madison, Correspondence to: solislemus@wisc.edu}, and Sameer K. Deshpande \thanks{Dept.~of Statistics, University of Wisconsin--Madison. Correspondence to: sameer.deshpande@wisc.edu}}

\hypersetup{pdfborder = {0 0 0.5 [3 3]}, colorlinks = true, linkcolor = BrightRed, citecolor = SkyBlue}

\begin{document}
\maketitle
\begin{abstract}
The multivariate regression interpretation of the Gaussian chain graph model simultaneously parametrizes (i) the direct effects of $p$ predictors on $q$ outcomes and (ii) the residual partial covariances between pairs of outcomes.
We introduce a new method for fitting sparse versions of these models with spike-and-slab LASSO (SSL) priors. 
We develop an Expectation Conditional Maximization algorithm to obtain sparse estimates of the $p \times q$ matrix of direct effects and the $q \times q$ residual precision matrix.
Our algorithm iteratively solves a sequence of penalized maximum likelihood problems with self-adaptive penalties that gradually filter out negligible regression coefficients and partial covariances. 
Because it adaptively penalizes individual model parameters, our method is seen to outperform fixed-penalty competitors on simulated data.
We establish the posterior contraction rate for our model, buttressing our method's excellent empirical performance with strong theoretical guarantees. 
Using our method, we estimated the direct effects of diet and residence type on the composition of the gut microbiome of elderly adults.
\end{abstract}

\newpage
\section{Introduction}
\label{sec:introduction}
\subsection{Motivation: Gut microbiome composition}
\label{sec:motivation}
Between 10 and 100 trillion microorganisms live within each person's lower intestines.
These bacteria, fungi, viruses, and other microbes constitute the human gut \textit{microbiome}  \citep{guinane2013role}. 
The composition of human gut microbiome can substantially affect our health and well-being \citep{shreiner2015gut}: in addition to playing an integral role in digestion and metabolic processes, microbes living in the gut can mediate immune response to certain diseases \citep{kamada2014regulation} and may even influence disease pathogenesis and progression \citep{wang2011gut}.

Emerging evidence suggests that the gut microbiome can mediate the effects of diet and medication use on human health \citep{singh2017influence}. 
That is, these factors may first affect the composition of the gut microbiome, which in turn influences health outcomes. 
Further, these factors can impact the composition of the microbiome in both direct and indirect ways.
For instance, many antibiotics target and kill certain microbial species, thereby directly affecting the abundances of the targeted species.
However, by killing the targeted species, the antibiotics may reduce the overall competition for nutrients, which allows non-targeted species to proliferate.
Thus, by directly reducing the abundance of a small number targeted microbes, antibiotics may indirectly increase the abundance of many other non-targeted species. 

In Section~\ref{sec:real_data_experiments}, we re-analyze a dataset \citet{claesson2012gut} containing $n = 178$ covariate-response pairs $(\bx, \by)$ where $\bx$ contains measurements of $p = 11$ factors related to diet, medication use, and residence type and $\by$ contains the logit-transformed relative abundances of $q = 14$ different microbial taxa. 
Our main analytic goal is to identify which lifestyle factors directly affect the relative abundances of which taxa. 


\subsection{Directly modeling direct effects}
\label{sec:direct_effects}

It is tempting to fit $q$ separate linear models, one for each outcome $Y_{k}.$
After all, doing so allows us to estimate the average change in each $Y_{k}$ associated with a change in $X_{j},$ keeping all other covariates fixed.
However, as we alluded in Section~\ref{sec:motivation}, there are two mechanisms through which a change in $X_{j}$ can produce a change in $Y_{k}.$
First, it is possible that $X_{j}$ may directly change $Y_{k}$ itself.
More subtly, it is possible that $X_{j}$ directly induces a change in some other outcome $Y_{k'},$ which in turns induces a change in $Y_{k}$ (i.e. an indirect effect).
Generally speaking, fitting separate linear models to our data estimates combinations of these direct and indirect effects.
Unfortunately, without modeling the residual dependence between outcomes, we cannot disentangle direct effects from the indirect effects.

To that end, a natural starting point for our analysis is the multivariate linear regression model that asserts 
\begin{equation}
\label{eq:marginal_model}
\by \vert B, \Omega, \bx \sim \mathcal{N}(B^{\top}\bx, \Omega^{-1}),
\end{equation}
where $B = (\beta_{j,k})$ is a $p \times q$ matrix of regression coefficients and $\Omega = (\omega_{k,k'})$ is a symmetric, positive definite $q \times q$ residual precision matrix. 
Under this model, $\omega_{k,k'},$ the $(k,k')$ entry of $\Omega,$ quantifies the residual partial covariance between outcomes $Y_{k}$ and $Y_{k'}$ that remains after adjusting for the effects of the covariates. 
Unfortunately, $\beta_{j,k}$ does not quantify the direct effect of $X_{j}$ on outcome $Y_{k}.$
To see this, observe that
\begin{equation}
\label{eq:beta_marginal_formula}
\beta_{j,k} = \E[Y_{k} \vert X_{j} = x + 1, X_{-j}] - \E[Y_{k} \vert X_{j} = x, X_{-j}].
\end{equation}
Notice that the expectations in right-hand side of Equation~\eqref{eq:beta_marginal_formula} do not condition on the values of the other outcomes $Y_{k'}$ for $k' \neq k.$
In other words, $\beta_{j,k}$ represents a certain \textit{marginal} association between $X_{j}$ and $Y_{k},$ keeping all other covariates fixed.
Specifically, $\beta_{j,k}$ represents a weighted average of (i) $X_{j}$'s direct effect on $Y_{k}$ and (ii) $X_{j}$'s indirect effect, which is induced through $X_{j}$'s direct effect on other outcomes $Y_{k'}$'s that themselves may be related to $Y_{k}.$

Although the model in Equation~\eqref{eq:marginal_model} does not directly parametrize the direct effects of interest, we can nevertheless compute them from $B$ and $\Omega.$
Letting $\psi_{j,k}$ be the $(j,k)$ entry of the matrix $\Psi = B\Omega,$ it turns out that 
\begin{equation}
\label{eq:psi_interpretation}
\psi_{j,k}/\omega_{k,k} = \E[Y_{k} \vert X_{j} = x_{j} + 1, Y_{-k}, X_{-j}] - \E[Y_{k} \vert X_{j} = x_{j}, Y_{-k}, X_{-j}].
\end{equation}
Based on this decomposition, a straightforward approach to estimating direct effects involves first fitting the model in Equation~\eqref{eq:marginal_model} and then computing the matrix $\Psi.$

This analytic plan is unfortunately inadequate for our purposes.
For one thing, although neither $p$ nor $q$ is especially large in our application, the total number of parameters ($pq + q(q+1)/2$) exceeds the total sample size $n.$
While we can overcome this challenge by assuming that $B$ and $\Omega$ are sparse (and estimating them accordingly), the resulting matrix of scaled direct effects $\Psi$ tends not to be sparse.
The combination of a sparse $B$, sparse $\Omega,$ and dense $\Psi$ corresponds to a situation in (i) a single covariate directly affects multiple outcomes but (ii) appears to be associated (in the usual linear regression sense) to very few.
In the context of our microbiome applications, a sparse $B$ and dense $\Psi$ would mean that several of the lifestyle factors directly affect the abundance of several taxa but that we would observe very few marginal associations. 
Such a scenario is implausible in the context of microbiome data; in fact the exact opposite tends to be true, with a small number of direct effects producing several marginal associations \citep[see, e.g.,][]{yassour2016natural,blaser2016antibiotic, thorpe2018enhanced, schwartz2020understanding, avis2021targeted, fishbein2023antibiotic}.

We instead propose fitting a re-paramatrized version of the model in Equation~\eqref{eq:marginal_model}:
\begin{equation}
\label{eq:cg_model}
\by \vert \Psi, \Omega, \bx \sim \N(\Omega^{-1}\Psi^{\top}\bx, \Omega^{-1}),
\end{equation}
where $\Psi$ and $\Omega$ are now assumed to be sparse.
Now, we may interpret $\psi_{j,k} \neq 0$ to mean that $X_{j}$ has a direct effect on $Y_{k}.$
Furthermore, whenever $\psi_{j,k} = 0,$ we can conclude that any marginal correlation between $X_{j}$ and $Y_{k}$ is due solely to $X_{j}$'s direct effects on other outcomes $Y_{k'}$ that are themselves conditionally correlated with $Y_{k}.$

\subsection{Our contributions}
\label{sec:contributions}

We introduce the chain graph spike-and-slab LASSO (cgSSL) for fitting the model in Equation~\eqref{eq:cg_model} by placing separate spike-and-slab LASSO priors \citep{RockovaGeorge2018_ssl} on the entries of $\Psi$ and on the off-diagonal entries of $\Omega.$
We derive an efficient Expectation Conditional Maximization algorithm to compute the \textit{maximum a posteriori} (MAP) estimates of $\Psi$ and $\Omega.$
We further quantify the uncertainty around this estimate using \citet{Newton2021}'s weighted Bayesian bootstrap.
Our algorithm involves solving a sequence of penalized maximum likelihood problems with individualized penalties for each parameter $\psi_{j,k}$ and $\omega_{k,k'}.$
In fact, these individualized penalties are \textit{self-adaptive}: the penalties are updated according to the previous iteration's parameter estimates, with smaller (resp.~larger) parameter estimates receiving larger (resp.~smaller) penalties. 
In this way, the algorithm automatically and adaptively \textit{learns} the appropriate amount of shrinkage to apply to each parameter.
On synthetic data, our algorithm displays excellent support recovery and estimation performance.
We further establish the posterior contraction rate for each of $\Psi, \Omega, \Psi\Omega^{-1},$ and $X\Psi\Omega^{-1}.$
Our contraction results provide asymptotic justification for MAP estimation and also upper bound the minimax optimal rates of estimating these quantities in the Frobenius norm. 
To the best of our knowledge, ours are the first posterior contraction results for sparse Gaussian chain graph models with element-wise priors on $\Psi$ and $\Omega.$

In Section~\ref{sec:background}, we review the Gaussian chain graph model and the spike-and-slab LASSO.
We next introduce the cgSSL prior in Section~\ref{sec:cgSSL_prior} and carefully derive our ECM algorithm for finding the MAP in Sections~\ref{sec:map_estimation}.
We present our asymptotic results in Section~\ref{sec:theory} before demonstrating the cgSSL's excellent finite sample performance on several synthetic datasets in Section~\ref{sec:synthetic_experiments}.
We apply the cgSSL to our motivating gut microbiome data in Section~\ref{sec:real_data_experiments}.
Finally, in Section~\ref{sec:discussion}, we outline avenues for future research in spike-and-slab uncertainty quantification and modeling outcomes of mixed type.

\section{Background}
\label{sec:background}
\subsection{Motivating~\eqref{eq:cg_model} and related graphical models}
\label{sec:related_graphical_models}

Under the model in Equation~\eqref{eq:cg_model}, we have for each $k = 1, \ldots, q,$ 
\begin{equation}
\label{eq:y_conditional}
Y_{k} \vert Y_{-k}, X, \Psi, \Omega \sim \mathcal{N}\left(-\omega_{k,k}^{-1}\sum_{k' \neq k}{\omega_{k,k'}y_{k'}} + \omega_{k,k}^{-1}\sum_{j = 1}^{p}{\psi_{j,k}X_{j}}, \omega_{k,k}^{-1}\right).
\end{equation}
In this way, the parameters $\Psi$ and $\Omega$ directly encode important conditional dependence relationship between the covariates and outcomes.
Specifically, if $\psi_{j,k} = 0,$ we can conclude from Equation~\eqref{eq:y_conditional} that $Y_{k}$ is conditionally independent from $X_{j}$ given all other covariates and outcomes.
And if $\omega_{k,k'} = 0,$ then $Y_{k}$ is conditionally independent of $Y_{k'},$ given all other covariates and outcomes.
Consequently, despite its somewhat complicated notation, we can represent the conditional dependencies encoded by the model in Equation~\eqref{eq:cg_model} with a simple graphical model.

Specifically, we construct a graph with $p+q$ vertices, one for each covariate $X_{j}$ and outcome $Y_{k}.$
We then draw a directed edges from $X_{j}$'s vertex to $Y_{k}$'s vertex whenever $\psi_{j,k} \neq 0.$
We additionally draw two directed edges (or, equivalently, a single bi-directed edge) between $Y_{k}$'s and $Y_{k'}$'s vertices whenever $\omega_{k,k'} \neq 0.$
Figure~\ref{fig:labelled_graph} shows a cartoon illustration of such a graph with $p = 3$ covariates and $q = 4$ outcomes.

The graph so constructed is a chain graph and is faithful to the model in Equation~\eqref{eq:cg_model} under \citet{cox1993linear}'s multivariate regression (MVR) interpretation of chain graph; see \citet{sonntag2015chain} for a comparison of different chain graph interpretations. 
Based on this interpretation, and following the examples of \citet{McCarter2014} and \citet{shen2021bayesian}, we will refer to the model in Equation~\eqref{eq:cg_model} as the Gaussian chain graph model.

\begin{figure}[h]
\centering
\includegraphics[width = 0.4\textwidth]{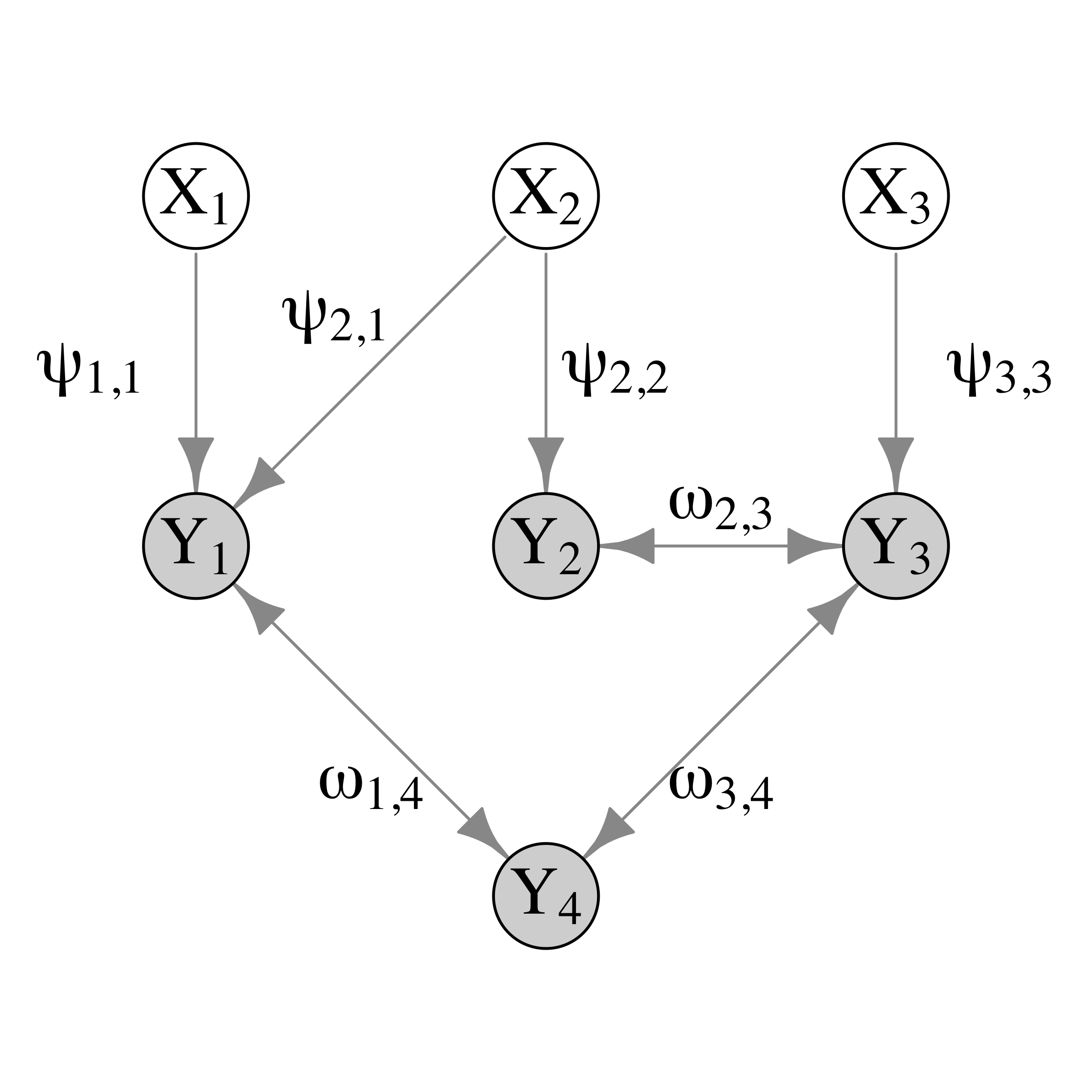}
\caption{Cartoon illustrations of a Gaussian chain graph model with $p = 3$ covariates and $q = 4$ outcomes. Edges indicate conditional dependence. Edge labels correspond to non-zero parameters in~Equation \eqref{eq:cg_model}.}
\label{fig:labelled_graph}
\end{figure}

\cite{McCarter2014} proposing fitting sparse Gaussian chain graph models by maximizing a penalized log-likelihood.
They specifically introduced homogeneous $L_1$ penalties on the entries of $\Psi$ and $\Omega$ and used cross-validation to set the penalty parameters for $\Psi$ and $\Omega.$
\cite{shen2021bayesian} developed a Bayesian version of that chain graphical LASSO and put a gamma prior on the penalty parameters.
In this way, they automatically learned the degree to which each $\psi_{j,k}$ and $\omega_{k,k'}$ should be shrunk to zero.
Although these papers differ in how they determined the appropriate amount of penalization, both deployed a single fixed penalty on all entries in $\Psi$ and a single fixed penalty on all entries in $\Omega.$
With fixed penalties, larger parameter estimates are shrunk towards zero as aggressively as smaller estimates, which can introduce substantial estimation bias. 

Fitting a sparse Gaussian chain graph model is related to the covariate-adjusted Gaussian graphical model problem \citep[see, e.g.,][and references therein]{Consonni2017, Ni2019}.
In that problem, interest lies in uncovering the conditional dependencies between $q$ outcomes that remain after controlling for $p$ covariates.
Typically, this is done by fitting a sparse version of the marginal model in Equation~\eqref{eq:marginal_model} and examining the support of $\Omega$ \citep[see, e.g.,][]{Cai2013, Chen2016, Chen2018}.
Within the context of Figure~\ref{fig:labelled_graph}, those works focus on detecting $Y$--$Y$ edges, while we are primarily interested in detecting $X$--$Y$ edges.
As we alluded to in Section~\ref{sec:direct_effects}, although it is possible to estimate direct effects by first estimating $B$ and $\Omega,$ doing so generally results in a very dense $\Psi$ that is at odds with our expectations, with one important exception: whenever the $j$-th row of $B$ contains all zeros, the $j$-th row of $\Psi$ will as well.

\citet{Consonni2017} introduced an Objective Bayes approach to fitting~\eqref{eq:marginal_model} with a row-sparse $B.$
Although their method can yield sparse $\Psi,$ it makes the restrictive assumption that each covariate directly affects all or none of the outcomes.
In contrast, our proposed procedure places no restrictions on the support of $\Psi,$ allowing covariates to directly affect all, none, or some of the outcomes.
 
 \subsection{Spike-and-slab variable selection \& asymptotics}
The spike-and-slab prior originally comprised a mixture of a point mass at 0 (the ``spike'') and a uniform distribution over a wide interval \citep[the ``slab'';][]{mitchell1988bayesian}.
\citet{george1993variable} introduced a continuous relaxation of the original prior, respectively replacing the point mass spike and uniform slab distributions with zero-mean Gaussians with small and large variances.
Intuitively, the spike generates the ``essentially negligible'' model parameters while the slab generates the ``significant'' parameters values.
In high dimensional problems, spike-and-slab priors often produce extremely multimodal posteriors, which render many Markov chain Monte Carlo strategies computationally prohibitive.

In response, \citet{RockovaGeorge2014_emvs} introduced EMVS, a fast Expectation Maximization \citep[EM;][]{Dempster1977_em} algorithm targeting the \textit{maximum a posteriori} (MAP) estimate of the regression parameters.
They later extended EMVS, which used Gaussian spike and slab distributions, to use Laplacian spike and slab distributions in \citet{RockovaGeorge2018_ssl}.
The resulting spike-and-slab LASSO (SSL) procedure demonstrated excellent empirical performance.  
The SSL algorithm solved a sequence of maximum likelihood problems with adaptive $L_{1}$-penalties that shrink larger parameter estimates to zero less aggressively than smaller parameter estimates.

\citet{RockovaGeorge2014_emvs}'s general EM technique for maximizing spike-and-slab posteriors has been successfully applied to many problems.
For instance, \citet{Tang2017_ssl_glm} deployed the SSL to fit sparse generalized linear models while \citet{Bai2020_groupSSL} introduced a grouped version of the SSL that adaptively shrinks groups of parameter values towards zero.
Beyond single-outcome regression, continuous spike-and-slab priors have been used to estimate sparse Gaussian graphical models \citep{Li2019_ssglasso, Gan2019_unequal, Gan2019_multiple}, sparse factor models \citep{RockovaGeorge2016_factor}, and to biclustering \citep{Moran2021_biclustering}. 
\citet{Deshpande2019} introduce a multivariate SSL for estimating $B$ and $\Omega$ in the marginal regression model in Equation~\eqref{eq:marginal_model}.
In each extension, the adaptive penalization esulted in superior support recovery and parameter estimation compared to fixed penalty methods. 

\citet{RockovaGeorge2018_ssl} proved that, under mild regularity conditions, the posterior induced by the SSL prior in high-dimensional, single-outcome linear regression contracts at a near minimax-optimal rate as $n \rightarrow \infty.$
\citet{Bai2020_groupSSL} extended these results to the group SSL posterior with an unknown variance.
In the context of Gaussian graphical models, \citet{Gan2019_unequal} showed that the MAP estimator corresponding to placing spike-and-slab LASSO priors on the off-diagonal elements of a precision matrix is consistent.
They did not, however, establish the contraction rate of the posterior.
\citet{Ning2020} showed that the joint posterior distribution of $(B,\Omega)$ in the multivariate regression model in Equation~\eqref{eq:marginal_model} concentrates when using a group spike-and-slab prior with Laplace slab and point mass spike on $B$ and a carefully selected prior on the eigendecomposition of $\Omega^{-1}.$
However, the asymptotic properties of the posterior formed by placing SSL priors on the entries of $\Psi$ and $\Omega$ have not yet been established.

\section{Introducing the cgSSL}
\label{sec:proposed_method}
\subsection{The cgSSL prior}
\label{sec:cgSSL_prior}
To quantify the prior belief that many entries in $\Psi$ are essentially negligible, we model each $\psi_{j,k}$ as having been drawn either from a spike distribution, which is sharply concentrated around zero, or a slab distribution, which is much more diffuse.
More specifically, we take the spike distribution to be $\text{Laplace}(\lambda_{0})$ and the slab distribution to be $\text{Laplace}(\lambda_{1}),$ where $0 < \lambda_{1} \ll \lambda_{0}$ are fixed positive constants.
We further let $\theta \in [0,1]$ be the prior probability that each $\psi_{j,k}$ is drawn from the slab and model the $\psi_{j,k}$'s as conditionally independent given $\theta.$
The prior density for $\Psi$, conditional on $\theta,$ is 
\begin{equation}
\label{eq:psi_prior_density}
\pi(\Psi \vert \theta) = \prod_{j = 1}^{p}{\prod_{k = 1}^{q}{\left(\frac{\theta\lambda_{1}}{2}e^{-\lambda_{1}\lvert \psi_{j,k}\rvert} + \frac{(1-\theta)\lambda_{0}}{2}e^{-\lambda_{0}\lvert \psi_{j,k} \rvert}\right)}}.
\end{equation}

Since $\Omega$ is symmetric, it is enough to specify a prior on the entries $\omega_{k,k'}$ where $k \leq k'.$
To this end, we begin by placing an entirely analogous spike-and-slab prior on the off-diagonal entries.
That is, for $k < k',$ we model each $\omega_{k,k'}$ as being drawn from a $\text{Laplace}(\xi_{1})$ with probability $\eta \in [0,1]$ or a $\text{Laplace}(\xi_{0})$ with probability $1 - \eta,$ where $0 < \xi_{1} \ll \xi_{0}.$
We similarly model each $\omega_{k,k'}$ as conditionally independent given $\eta$ and place independent exponential $\text{Exp}(\xi_{1})$ priors on the diagonal entries of $\Omega.$ 
We truncate the resulting distribution of $\Omega \vert \theta$ to the positive definite cone, yielding the prior density
\begin{align}
\begin{split}
\label{eq:omega_prior_density}
\pi(\Omega \vert \eta) &\propto \left(\prod_{1 \leq k < k' \leq q}\left[\frac{\eta\xi_1}{2}e^{-\xi_1|\omega_{k,k'}|}+\frac{(1-\eta)\xi_0}{2}e^{-\xi_0|\omega_{k,k'}|}\right]\right) \\
~&~~ \times \left(\prod_{k=1}^q e^{-\xi_{1}\omega_{k,k}}\right) \times \mathbbm{1}(\Omega\succ 0).
\end{split}
\end{align}
Note that the truncation implicitly introduces dependence between the entries in $\Omega.$
In Section S1.5 of the Supplementary Materials , we demonstrate that the truncated prior still marginally shrinks every $\omega_{k,k'}$ towards zero.
We have also observed empirically that, at least for small $q,$ the marginal prior of $\omega_{k,k'}$ tends to be more concentrated around zero after truncation than before truncation (see Figure S1 in the Supplementary Materials ).

The quantities $1-\theta$ and $1-\eta$ respectively capture the proportion of essentially negligible entries in $\Psi$ and $\Omega.$
We specify independent Beta priors for $\theta$ and $\eta$: $\theta \sim \text{Beta}(a_{\theta}, b_{\theta})$ and $\eta \sim \text{Beta}(a_{\eta}, b_{\eta}),$ where $a_{\theta}, b_{\theta}, a_{\eta}, b_{\eta} > 0$ are fixed positive constants.

\subsection{MAP estimation and uncertainty quantification}
\label{sec:map_estimation}

\subsubsection{MAP Estimation}
We follow \citet{RockovaGeorge2018_ssl} and approximate the \textit{maximum a posteriori} (MAP) estimate of $(\Psi, \theta, \Omega, \eta).$
Throughout, we assume that the columns of $X$ are centered and scaled to have norm $\sqrt{n}.$
For $\Omega \succ 0,$ the log posterior density is, up to an additive constant,
\begin{align*}
\log\pi(\Psi, \theta,\Omega,\eta \vert \by) &= \frac{n}{2} \times \log\lvert \Omega\rvert - \frac{1}{2}\text{tr}\left( \left(\bY - X\Psi\Omega^{-1}\right)^{\top}\Omega\left(\bY - X\Psi\Omega^{-1}\right)\right) \\
&+ \sum_{k = 1}^{q}{\sum_{j = 1}^{p}{\log\left(\theta\lambda_{1}e^{-\lambda_{1}\lvert \psi_{j,k} \rvert} + (1 - \theta)\lambda_{0}e^{-\lambda_{0}\lvert \psi_{j,k}\rvert}\right)}} \\
&+ \sum_{k = 1}^{q}{\left[-\xi_{1}\omega_{k,k} + \sum_{k' = k+1}^{q}{\log\left(\eta\xi_{1}e^{-\xi_{1}\lvert\omega_{k,k'}\rvert} + (1-\eta)\xi_{0}e^{-\xi_{0}\lvert \omega_{k,k'}\rvert}\right)}\right]} \\
&+ (a_{\theta}-1)\log(\theta)+ (b_{\theta}-1)\log(1-\theta) \\
&+ (a_{\eta} - 1)\log(\eta) + (b_{\eta}-1)\log(1-\eta),
\end{align*}
where the first line is the log-likelihood implied by the model in Equation~\eqref{eq:cg_model}.

Optimizing $\log \pi(\Psi, \theta, \Omega,\eta \vert \bY)$ directly is complicated by the non-concavity of $\log \pi(\Omega \vert \eta)$ (i.e. the term in the third line above).
Instead, we iteratively optimize a surrogate objective using an EM-like algorithm.
To motivate this approach, observe that we can obtain the prior density $\pi(\Omega \vert \eta)$ in Equation~\eqref{eq:omega_prior_density} by marginalizing an \textit{augmented} prior 
$$
\pi(\Omega \vert \eta) = \int{\pi(\Omega \vert \bdelta)\pi(\bdelta \vert \eta)d\bdelta}
$$
where $\bdelta = \{\delta_{k,k'}: 1 \leq k < k' \leq q\}$ is a collection of $q(q-1)/2$ i.i.d. $\text{Bernoulli}(\eta)$ variables, 
\begin{align*}
\pi(\Omega \vert \bdelta) &\propto \left(\prod_{1\le k<k'\le q}{\left(\xi_{1}e^{-\xi_{1}\lvert \omega_{k,k'}\rvert}\right)^{\delta_{k,k'}}\left(\xi_{0}e^{-\xi_{0}\lvert\omega_{k,k'}\rvert}\right)^{1-\delta_{k,k'}}}\right) \\
~&~~\times \left(\prod_{k=1}^q e^{-\xi_{1}\omega_{k,k}}\right) \times \mathbbm{1}(\Omega \succ 0),
\end{align*}
and $\delta_{k,k'}$ encodes whether $\omega_{k,k'}$ is drawn from the slab $(\delta_{k,k'} = 1$) or the spike ($\delta_{k,k'} = 0$).

The above marginalization immediately suggests an EM algorithm: rather than optimize $\log \pi(\Psi, \theta, \Omega, \eta \vert \bY)$ directly, we can iteratively optimize a surrogate objective formed by marginalizing the augmented log posterior density. 
That is, starting from some initial guess $(\Psi^{(0)}, \theta^{(0)}, \Omega^{(0)}, \eta^{(0)}),$ for $t > 1,$ the $t^{\text{th}}$ iteration of our algorithm consists of two steps.
In the first step, we compute the surrogate objective
$$
F^{(t)}(\Psi, \theta, \Omega, \eta) = \E_{\bdelta \vert \cdot}[\log \pi(\Psi, \theta, \Omega, \eta, \bdelta \vert \by) \vert \Psi = \Psi^{(t-1)}, \theta = \theta^{(t-1)}, \Omega = \Omega^{(t-1)}, \eta = \eta^{(t-1)}],
$$
where the expectation is taken with respect to the conditional posterior distribution of the indicators $\bdelta$ given the current value of $(\Psi, \theta, \Omega, \eta).$
Then in the second step, we maximize the surrogate objective and set $(\Psi^{(t)}, \theta^{(t)}, \Omega^{(t)}, \eta^{(t)}) = \argmax F^{(t)}(\Psi,\theta, \Omega,  \eta).$
We defer closed form expressions for the log densities of the augmented posterior (i.e., $\log \pi(\Psi,\theta,\Omega,\eta, \bdelta \vert \bY)$) and the surrogate objective function (i.e., $F^{(t)}(\Psi,\theta,\Omega,\eta)$) to Equations S1.5 and S1.6 in the Supplementary Materials . 

Given $\Omega$ and $\eta,$ the indicators $\delta_{k,k'}$ are conditionally independent, making it simple to derive a closed form expression for the surrogate objective $F^{(t)}.$
Unfortunately, maximizing $F^{(t)}$ is still difficult.
Consequently, we carry out two conditional maximizations, first optimizing with respect to $(\Psi,\theta)$ while holding $(\Omega,\eta)$ fixed, and then optimizing with respect to $(\Omega,\eta)$ while holding $(\Psi,\theta)$ fixed.
That is, in the second step of each iteration of our algorithm, we set
\begin{align}
(\Psi^{(t)}, \theta^{(t)}) &= \argmax_{\Psi,\theta}~ F^{(t)}(\Psi, \theta, \Omega^{(t-1)},  \eta^{(t-1)})  \label{eq:psi_theta_update} \\
(\Omega^{(t)}, \eta^{(t)}) &= \argmax_{\Omega,\eta}~ F^{(t)}(\Psi^{(t)}, \theta^{(t)},\Omega,  \eta). \label{eq:omega_eta_update}
\end{align}
In summary, we propose finding the MAP estimate of $(\Psi, \theta, \Omega, \eta)$ using an Expectation Conditional Maximization \citep[ECM;][]{MengRubin1993_ecm} algorithm. 

The objective function $F^{(t)}(\Psi, \theta, \Omega^{(t-1)}, \eta^{(t-1)})$ in Equation~\eqref{eq:psi_theta_update} can be written as the sum of a function of $\Psi$ alone and a function of $\theta$ alone.
So we separately compute $\Psi^{(t)}$ and $\theta^{(t)}$ while fixing $(\Omega, \eta) = (\Omega^{(t-1)}, \eta^{(t-1)})$.
The objective function in Equation~\eqref{eq:omega_eta_update} is similarly separable and we separately compute $\Omega^{(t)}$ and $\eta^{(t)}$ while fixing $(\Psi, \theta) = (\Psi^{(t)}, \theta^{(t)}).$
Computing $\theta^{(t)}$ and $\eta^{(t)}$ is relatively straightforward: we compute $\theta^{(t)}$ using Newton's method and there is a closed form expression for $\eta^{(t)}.$

The main computational challenge lies in computing each of $\Psi^{(t)}$ and $\Omega^{(t)}$ conditionally given all other parameters.
At a high level, we compute $\Psi^{(t)}$ with a cyclical coordinate descent algorithm that blends soft- and hard-thresholding.
Essentially, each entry $\psi_{j,k}$ is thresholded at a level determined by the conditional posterior probability that $\psi_{j,k}$ was drawn from the slab distribution; see Section S1.3 of the Supplementary Materials  for details. 
We compute $\Omega^{(t)}$ by solving an optimization problem that is similar to the graphical LASSO \citep[GLASSO;][]{Friedman2008} objective but includes additional trace terms.
We solve that problem by forming a quadratic approximation of the objective and following a Newton direction for a carefully chosen step size.
See Sections S1.4 and S2 of the Supplementary Materials  for the detailed derivation of the algorithm used to update $\Omega$ and a proof that the algorithm converges to the unique optimum.

\subsubsection{Implementation considerations}
The ability of the proposed ECM algorithm to identify the MAP critically depends on two sets of hyperparameters and the initial estimate $\Psi^{(0}$ and $\Omega^{(0)}.$
The first set of hyperparameters consists of the spike and slab penalties $\lambda_{0}, \lambda_{1}, \xi_{0}$ and $\xi_{1}.$
The second set, containing $a_{\theta}, b_{\theta}, a_{\eta},$ and $b_{\eta},$ encode our initial beliefs about the overall proportion of non-negligible entries in $\Psi$ and $\Omega.$
In this section, we recommend default hyperparameter settings and sketch a particular path-following scheme that provides good initialization for the ECM algorithm.
We have found these recommendations to work very well in practice and present a systematic hyperparameter sensitivity analysis in Section 3.1 of the Supplementary Materials .

It is initially tempting to run our ECM algorithm with very small slab penalties $\lambda_{1}$ and $\xi_{1}$ and very large spike penalties $\lambda_{0}$ and $\xi_{0}$ so that the slabs cover a wide range of non-zero parameter values while the spikes are supported only on narrow ranges of extremely small parameter values.
Unfortunately, our algorithm was quite sensitive to initialization with such choices.
In fact, in early experiments with synthetic data, with very large spike and very small slab penalties, the algorithm tended to estimate $\Psi$ with a zero matrix and $\Omega$ with a diagonal matrix with very small diagonal entries, unless it was initialized close to the true data-generating parameter values.
Such initialization is, of course, impossible in practice.
To overcome this challenge, rather than run cgSSL with a single set of spike and slab penalties, we run our ECM algorithm along sequences of increasing values of the spike-penalties using warm-starts.

Specifically, we fix the slab penalties $\lambda_{1}$ and $\xi_{1}$ and specify grids of $L$ increasing spike penalties $\mathcal{I}_\lambda=\{\lambda_0^{(1)}<\dots <\lambda_0^{(L)}\}$ and $\mathcal{I}_\xi=\{\xi_0^{(1)}<\dots<\xi_0^{(L)}\}.$ 
We then run cgSSL with warm-starts for each combination of spike penalties, yielding a set of posterior modes $\{(\Psi^{(s,t)}, \theta^{(s,t)}, \Omega^{(s,t)}, \eta^{(s,t)})\}$ indexed by the choices $(\lambda_{0}^{(s)}, \xi_{0}^{(t)}).$
To warm-start the estimation of the mode corresponding to $(\lambda_{0}^{(s)}, \xi_{0}^{(t)}),$ we first compute the modes for $(\lambda_{0}^{(s-1)},\xi_{0}^{(t-1)})$, $(\lambda_{0}^{(s)},\xi_{0}^{(t-1)})$ and $(\lambda_{0}^{(s-1)},\xi_{0}^{(t)}).$
Then we initialize the ECM algorithm from the mode with highest density (computed using $(\lambda^{(s)}_{0}, \xi^{(t)}_{0})$). 

Our path-following strategy mirrors one first introduced in \citet{RockovaGeorge2014_emvs} and subsequently deployed by many others \citep[see, e.g.,][]{RockovaGeorge2016_factor, RockovaGeorge2018_ssl, Moran2019_variance, Moran2021_biclustering}. 
In fact, our strategy exactly matches the one used by \citet{Deshpande2019}, who fit sparse versions of the model in Equation~\eqref{eq:marginal_model} by placing spike-and-slab LASSO priors on the $\beta_{j,k}$'s and $\omega_{k,k'}$'s.
Taking a cue from that literature, we refer to our strategy as \textit{dynamic posterior exploration} (DPE).

It is important to stress that the goal of DPE is not to tune or optimize hyperparameters.
Instead, we use DPE to identify a good initialization from which to launch our algorithm with large spike penalties.
It specifically does so by computing $L^{2}$ posterior modes, one for each combination of spike penalties.  
We can additionally think of DPE as passing an initial estimate of $\Psi$ and $\Omega$ through a sequence of increasingly discerning filters.  
As the spike penalties increase, DPE filters out more and more negligible parameter values, producing sparser and sparse estimates of $\Psi$ and $\Omega.$

In practice, we recommend setting $\lambda_{1} = 1$ and $\xi_{1} = 0.01n$ and letting $\mathcal{I}_{\lambda}$ contain ten evenly spaced values ranging from $10$ to $n$ and $\mathcal{I}_{\xi}$ contain ten evenly spaced values from $0.1n$ to $n.$
We further recommend reporting the final mode computed by DPE, corresponding to the largest spike penalties, as a final parameter estimate.
With these choices, the sequence of posterior modes computed by DPE appeared to stabilize in the underlying parameter space.
That is, for large penalty values $\lambda_{0}^{(s)}, \lambda_{0}^{(s')}, \xi_{0}^{(t)},$ and $\xi_{0}^{(t')},$ we have $(\Psi^{(s,t)}, \theta^{(s,t)}, \Omega^{(s,t)}, \eta^{(s,t)}) \approx (\Psi^{(s',t')}, \theta^{(s',t')}, \Omega^{(s',t')}, \eta^{(s',t')});$ see Figure S2 in the Supplementary Materials  for an illustration.

While experimenting with different penalty choices, we observed that when the spike and slab penalties were approximately equal and small, our ECM algorithm would often return very dense estimates of $\Psi$ and diagonal estimates of $\Omega$ with large diagonal entries.
Essentially, when the spike and slab distributions are not too different and when neither encourages strong shrinkage, our ECM algorithm tended to overfit the data using a dense $\Psi$, leaving very little residual variation to be quantified with $\Omega.$
We found that we could detect such pathological behavior by examining the condition number of the matrix $Y\Omega - X\Psi$. 
To avoid propagating dense $\Psi$'s and diagonal $\Omega$'s through DPE, we terminate our ECM algorithm early whenever the condition number of $Y\Omega - X\Psi$ exceeds $10n.$
We then set the corresponding $\Psi^{(s,t)} = 0_{p\times q}$ and $\Omega^{(s,t)} = I_{q\times q}$ and continue the dynamic exploration from that point.
Though it is not foolproof, we have found this heuristic to work well in practice.
We also note that \citet{Moran2019_variance} utilized a similar early termination strategy in the single-outcome high-dimensional linear regression setting with unknown variance. 

To further discourage our ECM algorithm from over-fitting the data with a dense $\Psi,$ we recommend the default priors $\theta \sim \text{Beta}(1,pq)$ and $\eta \sim \text{Beta}(1,q).$
These priors concentrate probability around models with very few direct covariate effects on the outcomes and very few conditional dependencies between outcomes, after adjusting for the covariates. 
These choices mirror similar prior choices made by \citet{RockovaGeorge2018_ssl} and \citet{Deshpande2019}. 

To summarize, we recommend setting $a_{\theta}=1, b_{\theta}=pq, a_{\eta}=1,$ and $b_{\eta}=q,$ and running DPE with $\lambda_{1} = 1$ and $\xi_{1} = 0.01n$ and letting $\mathcal{I}_{\lambda}$ contain ten evenly spaced values ranging from $10$ to $n$ and $\mathcal{I}_{\xi}$ contain ten evenly spaced values from $0.1n$ to $n.$
We implemented these choices as the default options in the \textsf{R} \citep{R_citation} package \textbf{mSSL}, which is available at \url{https://github.com/YunyiShen/mSSL}.
In Section S3.1 of the Supplementary Materials , we assess the sensitivity of our DPE implementation to (i) the range of spike penalty values in the grids $\mathcal{I}_{\lambda}$ and $\mathcal{I}_{\xi}$; (ii) the size of the grids $\mathcal{I}_{\lambda}$ and $\mathcal{I}_{\xi}$; and the choice of prior for $\theta$ and $\eta.$
We found that the final point estimates returned by our default implementation displayed better recovered the supports of $\Psi$ and $\Omega$ than those found with larger or smaller penalty values.
We additionally did not observe much sensitivity to the number of spike penalties in the grids $\mathcal{I}_{\lambda}$ and $\mathcal{I}_{\xi}$ nor to the choice of priors for $\theta$ and $\eta.$

\subsubsection{Uncertainty quantification via the weight Bayesian bootstrap}

The cgSSL posterior distribution is not log-concave, rendering efficient MCMC or importance sampling computationally prohibitive.
\citet{Newton2021}'s weighted Bayesian bootstrap, on the other hand, offers a computationally practical and embarrassingly parallel alternative. 
At a high-level, the procedure works by repeatedly solving a MAP estimation problem that randomly re-weights every observation's contribution to the log-likelihood and the log-prior density. 
We specifically use the following two-step procedure.
In the first step, we run cgSSL with DPE and obtain point estimates $\hat{\Psi}, \hat{\theta}, \hat{\Omega},$ and $\hat{\eta}.$
Then in the second step, we repeatedly solve the single optimization problem
\begin{equation}
\label{eq:master_for_bb}
    \argmin_{\Psi,\Omega,}\left\{ \sum_{i=1}^n w_i\ell_{i}(\Psi,\Omega) +w_0\left[\log \pi^{(L)}_{\psi}(\Psi )+\log \pi^{(L)}_{\omega}(\Omega)\right] \right\},
\end{equation}
where $\bw = (w_{0}, w_{1}, \ldots, w_{n})$ is a vector of independent $\text{Gamma}(1,1)$ weights, $\pi^{(L)}(\Psi)$ is the conditional prior density of $\Psi \vert \theta = \hat{\theta}^{(L,L)}$ with $\lambda_{0} = \lambda_{0}^{(L)},$ and $\pi^{(L)}(\Omega)$ is analogously defined.
Note, we do not re-run the full DPE procedure to generate each bootstrap sample.
Instead, we fix the values of $\theta$ and $\eta$ and also warm-start our optimization from the stabilized $\Psi$ and $\Omega$ estimates.

\section{Asymptotic theory of  cgSSL}
\label{sec:theory}
If the Gaussian chain graph model in Equation~\eqref{eq:cg_model} is well-specified --- that is, if our data are truly generated according to the model --- will the posterior distribution of $\Psi$ and $\Omega$ collapse to a point-mass at the true data generating parameters as $n\to\infty$?
In this section, we answer the question affirmatively: under some mild assumptions and with some slight modifications, the cgSSL posterior concentrates around the truth. 
We further establish the rate of concentration, which quantifies the speed at which the posterior distribution shrinks to the true data generating parameters.
We begin by briefly reviewing our general proof strategy before precisely stating our assumptions and results.
Proofs of our main results are available in Section S5 of the Supplementary Materials .

\subsection{Proof strategy}
Following \citet{Ning2020} and \citet{Bai2020_groupSSL}, we first showed that the posterior of $(\Psi,\Omega)$ concentrates in log-affinity.
Posterior concentration of the individual parameters followed as a consequence.
To show that the posterior concentrates in log-affinity, we verified the three conditions of Theorem 8.23 of \citet{ghosal2017fundamentals}.
First, we confirmed that the cgSSL prior places enough prior probability mass in small neighborhoods around every possible choice of $(\Psi,\Omega).$
This was done by verifying that for each $(\Psi,\Omega),$ the prior probability contained in a small Kullback-Leibler ball around $(\Psi,\Omega)$ can be lower bounded by a function of the ball's radius.
Then we studied a sequence of likelihood ratio tests defined on sieves of the parameter space that can correctly distinguish between parameter values that are sufficiently far away from each other in log-affinity.
In particular, we bounded the error rate of such tests and then bounded the covering number of the sieves.

\citet{Ning2020} studied the sparse marginal regression model in Equation~\eqref{eq:marginal_model} instead of the sparse chain graph. 
Although these are somewhat different models, our overall proof strategy is quite similar to theirs.
However, there are important technical differences.
First, they placed a prior on $\Omega$'s eigendecomposition while we placed an arguably simpler and more natural element-wise prior on $\Omega.$
The second and more substantive difference is in how we bound the covering number of sieves of the underlying parameter space.
Because they specified exactly sparse priors on the elements of $B = \Psi\Omega^{-1},$ it was enough for them to carefully bound the covering number of exactly low-dimensional sets of the form $\mathcal{A} \times \{0\}^{r}$ where $\mathcal{A}$ is some subset of a multi-dimensional Euclidean space and $r > 0$ is a positive integer.
In contrast, because we specified absolutely continuous priors on the elements of $\Psi,$ we had to cover ``effectively low-dimensional'' sets of the form $\mathcal{A} \times [-\delta, \delta]^{r}$ for small $\delta > 0.$
Our key lemma, Lemma S4 of the Supplementary Materials , provides sufficient conditions on $\delta$ for bounding the $\epsilon$-packing number of effectively low-dimensional sets using the $\epsilon'$-packing number of $\mathcal{A}$ for a carefully chosen $\epsilon' > 0.$

\subsection{Contraction of cgSSL}

In order to establish our posterior concentration results, we first assume that the data were generated according to a sparse Gaussian chain graph model with true parameter $\Psi_{0}$ and $\Omega_{0}.$ 
We additionally modify our prior on $\Omega$ by truncating it to the set $\{\Omega \succ \tau\}$ for some small $\tau$ that does not depend on $n$.
Finally, we make the following assumptions about the spectra of $\Psi_{0}$ and $\Omega_{0}$ and on the dimensions $n,p,$ and $q.$
Note that for sequences $\{a_{n}\}$ and $\{b_{n}\},$ we write $a_{n}\lesssim b_{n}$ to mean that there is some constant $C$ independent of $n$ such that $a_{n} \leq Cb_{n}.$
\begin{itemize}
\item[\textbf{A1}]{Bounded operator norms: $\Psi_{0} \in \mathcal{T}_{0}=\{\Psi:|||\Psi|||_2<a_1\}$ and $\Omega_{0} \in \mathcal{H}_{0} = \{\Omega:\text{eig}(\Omega) \subseteq [1/b_2,1/b_1]\}$ where $||| \cdot |||$ is the operator norm, $\text{eig}$ is the set of eigenvalues, and $a_{1}, b_{1}, b_{2} > 0$ are fixed positive constants not depending on $n.$}
\item[\textbf{A2}]{Dimensionality: We assume that $\log(n)\lesssim \log(q);$  $\log(n)\lesssim \log(p);$  and 
$$
\max\{p,q,s_0^\Omega,s_0^\Psi\}\log(\max\{p,q\})/n\to 0,
$$
where $s_{0}^{\Omega}$ and $s_{0}^{\Psi}$ are the number of non-zero free parameters in $\Omega$ and $\Psi.$}
\item[\textbf{A3}]{Tuning the $\Psi$ prior: We assume that $(1-\theta)/\theta\sim (pq)^{2+a'}$;
$\lambda_0 \sim \max\{n,pq\}^{2+b'}$; and $\lambda_1\asymp  1/n$ where $a' > 0$ and $b' > 1/2$ are fixed constants not depending on $n$}
\item[\textbf{A4}]{Tuning the $\Omega$ prior: We assume that $(1-\eta)/\eta\sim \max\{Q,pq\}^{2+a};$ $\xi_0\sim \max\{Q,pq,n\}^{4+b};$ and $\xi_{1} \asymp 1/\max\{Q,n\},$ where $Q = q(q-1)/2$ and $a,b > 0$ are fixed constants not depending on $n.$}
\end{itemize}

Before proceeding, we highlight two key differences between the above assumptions and model introduced in Section~\ref{sec:cgSSL_prior}.
Although the prior in Section~\ref{sec:cgSSL_prior} restricts $\Omega$ to the positive-definite cone, our modified prior and Assumption A1 bound the smallest eigenvalue of $\Omega$ away from zero.
The stronger assumption ensures that the entries of $\Psi\Omega^{-1}$ do not diverge in our theoretical analysis and parallels those made by \citet{Gan2019_unequal}, \citet{Ning2020}, and \citet{Sagar2021precision}.
Based on these works and ours, we conjecture that bounded eigenvalue assumptions may in fact be necessary.
Additionally, like \citet{RockovaGeorge2018_ssl} and \citet{Gan2019_unequal}, we restricted our theoretical analysis to the setting where the proportion of non-negligible parameters, $\theta$ and $\eta,$ are fixed and known (Assumptions A3 and A4).

\begin{myTheorem}[Posterior contraction]
\label{thm:posterior_contraction}
Under Assumptions A1--A4, there is a constant $M_{1} > 0$ not depending on $n$ such that 
\begin{align}
\sup_{\Psi\in\mathcal{T}_0,\Omega\in\mathcal{H}_0}\E_0 \Pi\left(\Psi:||X(\Psi\Omega^{-1}-\Psi_0\Omega^{-1}_0)||_F^2\ge M_1n\epsilon_n^2|Y_1,\dots,Y_n\right) & \longrightarrow 0 \label{eqn:contract_regressionfun_cg}\\
\sup_{\Psi\in\mathcal{T}_0,\Omega\in\mathcal{H}_0}\E_0 \Pi\left(\Omega:||\Omega-\Omega_0||_F^2\ge M_1\epsilon_n^2|Y_1,\dots,Y_n\right) &\longrightarrow 0 \label{eqn:contract_omega_cg}
\end{align}
 where $\epsilon_n=\sqrt{\max\{p,q,s_0^\Omega,s_0^\Psi\}\log(\max\{p,q\})/n}.$ Note that $\epsilon_{n} \to 0$ as $n \rightarrow \infty.$ 
 \end{myTheorem} 
 
A key step in proving Theorem~\ref{thm:posterior_contraction} is Lemma~\ref{lemma:dimension_recovery}.
In order to state this lemma, we denote the effective dimensions of $\Psi$ and $\Omega$ by $\lvert \nu_{\psi}(\Psi) \rvert $ and $\lvert \nu_{\omega}(\Omega) \rvert.$
The effective dimension of $\Psi$ (resp. $\Omega$) counts the number of entries (resp. off-diagonal entries in the lower-triangle) whose absolute value exceeds the intersection point of the spike and slab prior densities. 

\begin{myLemma}[Dimension recovery]
    \label{lemma:dimension_recovery}
    For a sufficiently large constant $C_3'>0$, we have:
    \begin{align}
        \sup_{\Psi\in\mathcal{T}_0,\Omega\in\mathcal{H}_0}\E_0\Pi\left(\Psi:|\nu_{\psi}(\Psi)|>C_3's^\star|Y_1,\dots,Y_n\right)&\to 0\\
        \sup_{B\in\mathcal{T}_0,\Omega\in\mathcal{H}_0}\E_0\Pi\left(\Omega:|\nu_{\omega}(\Omega)|>C_3's^\star|Y_1,\dots,Y_n\right)&\to 0
    \end{align}
    where $s^\star=\max\{p,q,s_0^\Omega,s_0^\Psi\}.$ 
\end{myLemma}
Lemma~\ref{lemma:dimension_recovery} guarantees that the cgSSL posterior does not grossly overestimate the number of non-zero entries in $\Psi$ and $\Omega.$

Note that the result in Equation~\eqref{eqn:contract_regressionfun_cg} shows that the vector $X\Psi\Omega^{-1}$ converges to the vector of evaluations of the true regression function $\Omega_{0}^{-1}\Psi_{0}^{\top}\bx.$
Importantly, apart from Assumption A2 about the dimensions of $X,$ Theorem~\ref{thm:posterior_contraction} does not require any assumptions about the design matrix $X.$
The contraction rates for $\Psi$ and $\Psi\Omega^{-1},$ however, depend critically on a restricted eigenvalue of $X,$ which we define as
$$
\phi^{2}(s) = \inf_{A\in \mathbb{R}^{p\times q}:0\le |\nu(A)|\le s} \left\{ \frac{\lVert XA \rVert^{2}_{F}}{n\lVert A \rVert^{2}_{F}} \right\}.
$$

\begin{myCorollary}[Regression coefficients recovery]
    \label{coro:contraction_B_cg}
    Under Assumptions A1--A4, there is a constant $M' > 0$ not depending on $n$ such that
    \begin{align}
        \sup_{\Psi\in \mathcal{T}_0,\Omega\in\mathcal{H}_0}\E_0 \Pi\left(||\Psi\Omega^{-1}-\Psi_0\Omega^{-1}_0||_F^2\ge \frac{M'\epsilon_n^2}{\phi^2(s_0^\Psi+C_3's^\star)}\right)&\to 0 \label{eqn:recover_marginal_cg}\\
        \sup_{\Psi\in \mathcal{T}_0,\Omega\in\mathcal{H}_0}\E_0 \Pi\left(||\Psi-\Psi_0||_F^2\ge \frac{M'\epsilon_n^2}{\min\{\phi^2(s_0^\Psi+C_3's^\star),1\}}\right)&\to 0.\label{eqn:recover_cond_cg}
    \end{align}
\end{myCorollary}
Corollary~\ref{coro:contraction_B_cg} shows that the posterior distribution of $\Psi\Omega^{-1}$ can contract at a faster or slower rate than the posterior distributions of $X\Psi\Omega^{-1}$ and $\Omega,$ depending on the design matrix.
In particular, when $X$ is poorly conditioned, we might expect the rate to be slower. 
In contrast, the term $\min\{\phi^2(s_0^\Psi+C_3's^\star),1\}$ appearing in the denominator of the rate in Equation~\eqref{eqn:recover_cond_cg} implies that the posterior distribution of $\Psi$ cannot concentrate at a faster rate than the posterior distributions of $\Psi\Omega^{-1}$ and $\Omega,$ regardless of the design matrix.
To develop some intuition about this phenomenon, notice that
$$
\Psi - \Psi_{0} = (\Psi\Omega^{-1}-\Psi_0\Omega^{-1}_0)\Omega+(\Psi_0\Omega_0^{-1}(\Omega-\Omega_0)\Omega^{-1})\Omega.
$$
Roughly speaking, the decomposition suggests that in order to estimate $\Psi$ well, we must estimate both $\Omega$ and $\Psi\Omega^{-1}$ well.
That is, estimating $\Psi$ is at least as hard, statistically, as estimating $\Omega$ and $\Psi\Omega^{-1}.$
Taken together, Corollary~\ref{coro:contraction_B_cg} suggests that while a carefully constructed design matrix can improve estimation of the matrix of \textit{marginal} effects, $B = \Psi\Omega^{-1},$ it cannot generally improve estimation of the matrix of \textit{direct} effects $\Psi.$

\section{Synthetic experiments}
\label{sec:synthetic_experiments}
We performed a simulation study to assess how well cgSSL with DPE (\texttt{cgSSL-DPE}) (i) recovers the supports of $\Psi$ and $\Omega$ and (ii) estimates each matrix.
We simulated several synthetic datasets of various dimensions and with different sparsity patterns in $\Omega$ (Figure~\ref{fig:simulation_design}) from the model in Equation~\eqref{eq:cg_model}.
We compared cgSSL to several competitors: a fixed-penalty method (\texttt{cgLASSO}), which uses 10-fold cross-validation to select a single penalty $\lambda$ for the entries in $\Psi$ and a single fixed penalty $\xi$ for the entries in $\Omega$; \citet{shen2021bayesian}'s CAR-LASSO procedure (\texttt{CAR}), which puts a common Laplace prior on the entries in $\Psi$ and a common Laplace prior entries in $\Omega$; \citet{shen2021bayesian}'s adaptive CAR-LASSO (\texttt{CAR-A}), which puts individualized Laplace priors on the entries in $\Psi$ and $\Omega;$ \citet{Deshpande2019}'s \texttt{mSSL} procedure that places spike-and-slab LASSO priors on the entries of $B$ and $\Omega$ in Equation~\eqref{eq:marginal_model}; and \citet{Consonni2017}'s Objective Bayes procedure (\texttt{OBFB}).

Before proceeding, we note that \texttt{mSSL} is inherently misspecified as it assumes $B = \Psi\Omega^{-1}$ is sparse instead of $\Psi.$
Nevertheless, we included \texttt{mSSL} in our experiments to investigate how misspecifying the mean structure affects our ability to recover $\Omega.$
\texttt{OBFB} is similarly misspecified, albeit to a somewhat lesser extent, as it assumes that $B$ (and consequently $\Psi$) is row-sparse.
Unfortunately, the implementation of \texttt{OBFB} only returns posterior samples of indicators encoding which rows of $B$ and which entries in $\Omega$ are non-zero.
We computed a point-estimate of $\Omega'$s support using the posterior mode of the relevant indicators.
We were unable, however, to reliably reconstruct $\Psi$'s support from \texttt{OBFB}'s output.
This is because a zero-entry in $\Psi$ can occur when a non-zero row of $B$ is orthogonal to a non-zero column in $\Omega.$
For these reasons, we only report \texttt{OBFB}'s performance in recovering the support of $\Omega.$ 

Across all choices of dimension and $\Omega,$ we found that \texttt{cgSSL-DPE} achieved somewhat lower sensitivity but much higher precision in estimating the supports of both $\Psi$ and $\Omega$ than the competing methods.
This means that while \texttt{cgSSL-DPE} tended to return fewer non-zero parameter estimates than the other methods, we can be much more certain that those parameters are truly non-zero.
Put another way, although the other methods can recover more of the truly non-zero signal, they do so at the expense of making many more false positive identifications in the supports of $\Psi$ and $\Omega$ than \texttt{cgSSL-DPE}. 

\subsection{Simulation design}

We simulated data with three different dimensions $(n,p,q) = (100, 10, 10),$ $(100, 20, 30),$ and $(400, 100, 30).$
For each choice of $(n,p,q),$ we considered seven different $\Omega$'s: (i) an AR(1) model for $\Omega^{-1}$ so that $\Omega$ is tri-diagonal; (ii) an AR(2) model for $\Omega^{-1}$ so that $\omega_{k,k'} = 0$ whenever $\lvert k - k'\vert > 2$; (iii) a block model in which $\Omega$ is block-diagonal with two dense $q/2 \times q/2$ diagonal blocks; (iv) a star graph where the off-diagonal entry $\omega_{k,k'} = 0$ unless $k$ or $k'$ is equal to 1; (v) a small-world network; (vi) a tree network; and (viii) dense model with all off-diagonal elements $\omega_{k,k'} = 2.$

In the AR(1) model we set $(\Omega^{-1})_{k,k'} = 0.7^{\lvert k - k' \rvert}$ so that $\omega_{k,k'} = 0$ whenever $\lvert k - k' \vert > 1.$
In the AR(2) model, we set $\omega_{k,k} = 1, \omega_{k-1,k} = \omega_{k,k-1} = 0.5,$ and $\omega_{k-2,k} = \omega_{k,k-2} = 0.25.$
For the block model, we partitioned $\Sigma = \Omega^{-1}$ into 4 $q/2 \times q/2$ blocks and set all entries in the off-diagonal blocks of $\Sigma$ to zero. 
We then set $\sigma_{k,k} = 1$ and $\sigma_{k,k'} = 0.5$ for $1 \leq k \neq k' \leq q/2$ and for $q/2 + 1 \leq k \neq k' \leq q.$
For the star graph, we set $\omega_{k,k} = 1$, $\omega_{1,k} = \omega_{k,1} = 0.1$ for each $k > 1,$ and set the remaining off-diagonal elements of $\Omega$ equal to zero.
For the small-world and tree networks, we first generated an appropriate random graph and then drew $\Omega$ from a G-Wishart distribution \citep{Roverato2002, Lenkoski2013} with three degrees of freedom and an identity scale matrix.
We generated the small-world graph using the Watts-Strogatz \citep{WattsStrogatz1998} model with a single community and rewiring probability of 0.1.
We generated the tree graph by running a loop-erased random walk on a complete graph. 

These seven specifications of $\Omega$ (top row of Figure~\ref{fig:simulation_design}) correspond to rather different underlying graphical structure among the outcomes (bottom row of Figure~\ref{fig:simulation_design}). 
The AR(1) model, for instance, represents an extremely sparse but regular structure while the AR(2) model is somewhat less sparse.
While the star model and AR(1) model contain the same number of edges, the underlying graphs have markedly different degree distributions.
Compared to the AR(1), AR(2), and star models, the block model is considerably denser.
We included a dense $\Omega$ to assess how well all of the methods perform in a misspecified regime. 

\begin{figure}[h]
\centering
\begin{subfigure}[b]{0.13\textwidth}
\centering
\includegraphics[width = \textwidth]{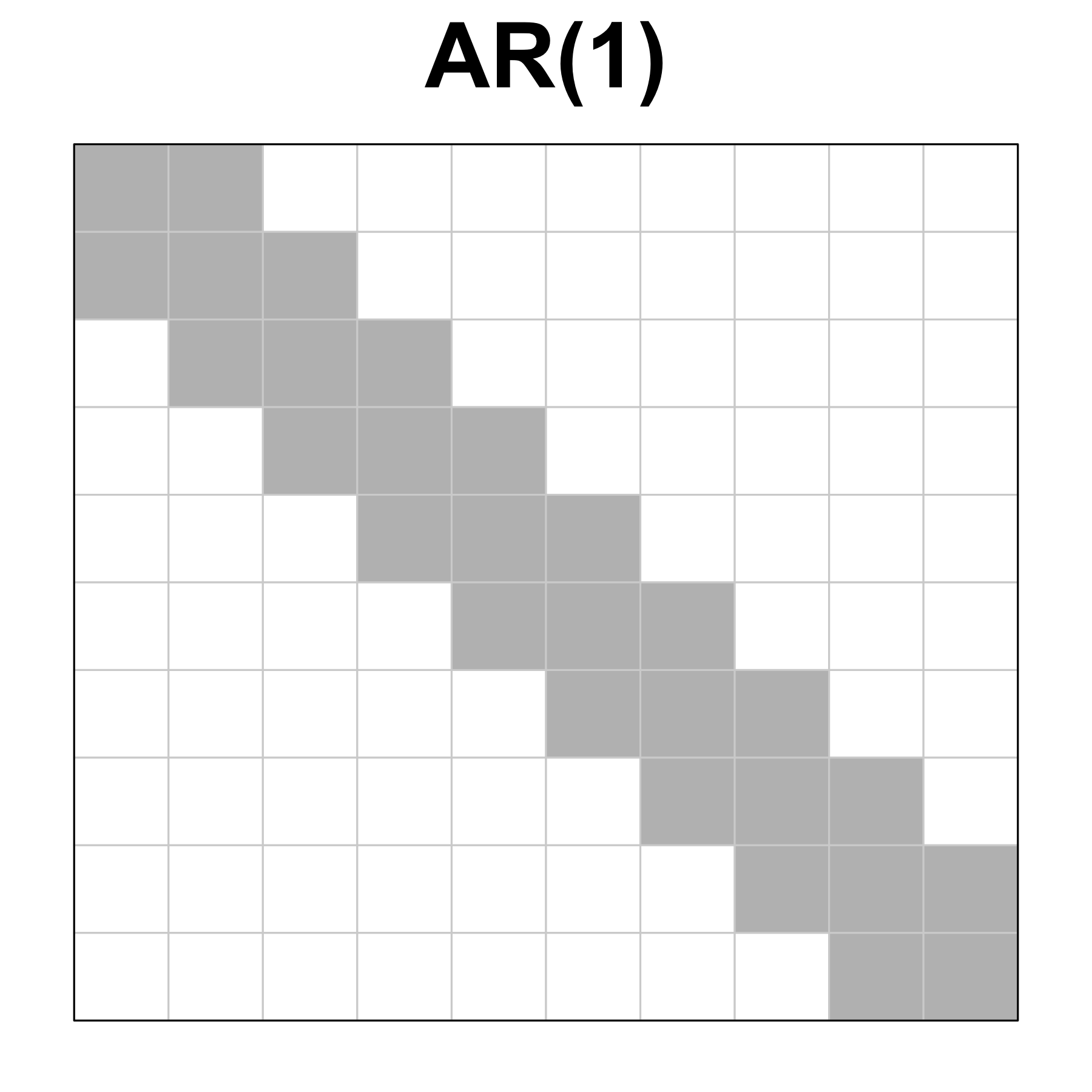}
\end{subfigure}
\begin{subfigure}[b]{0.13\textwidth}
\centering
\includegraphics[width = \textwidth]{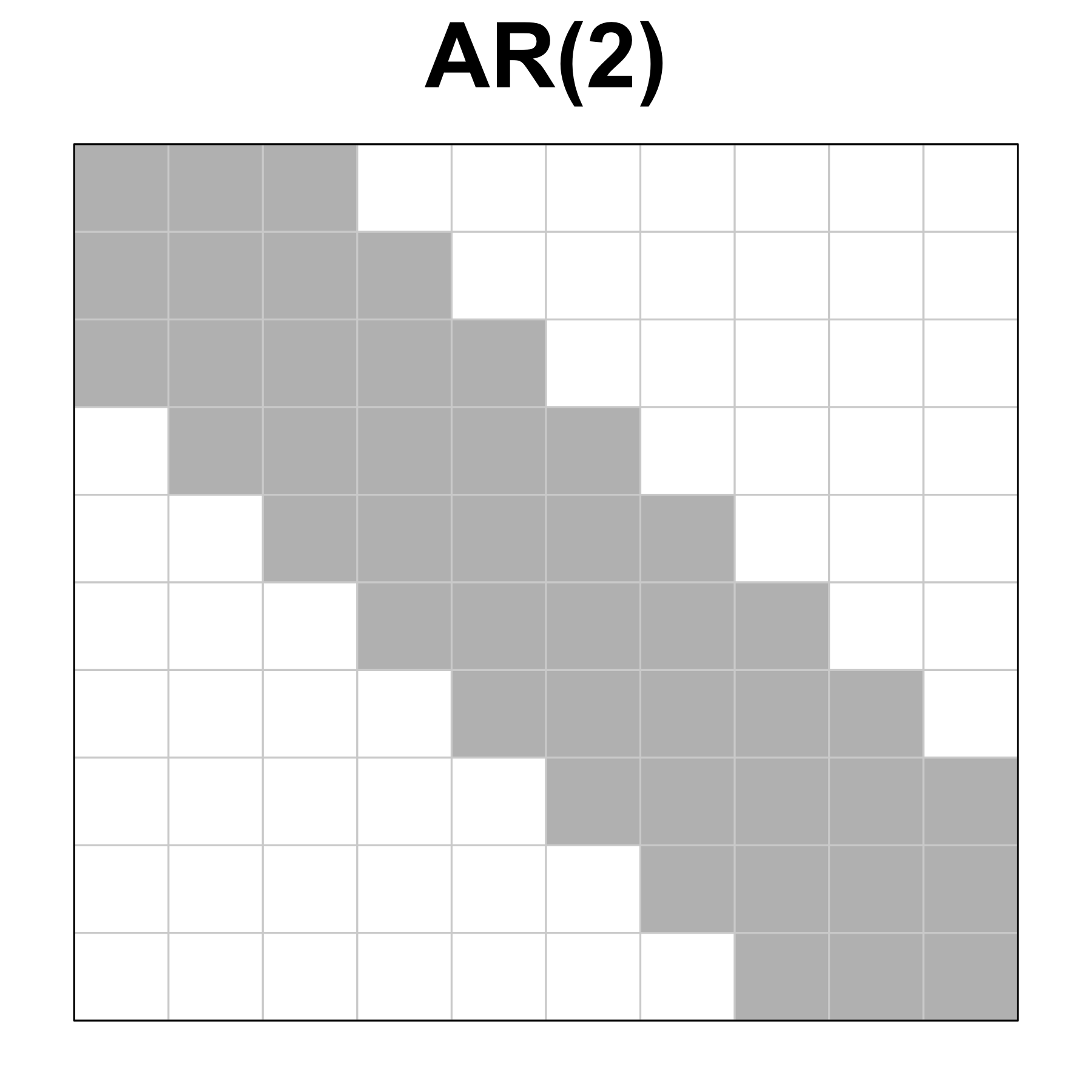}
\end{subfigure}
\begin{subfigure}[b]{0.13\textwidth}
\centering
\includegraphics[width = \textwidth]{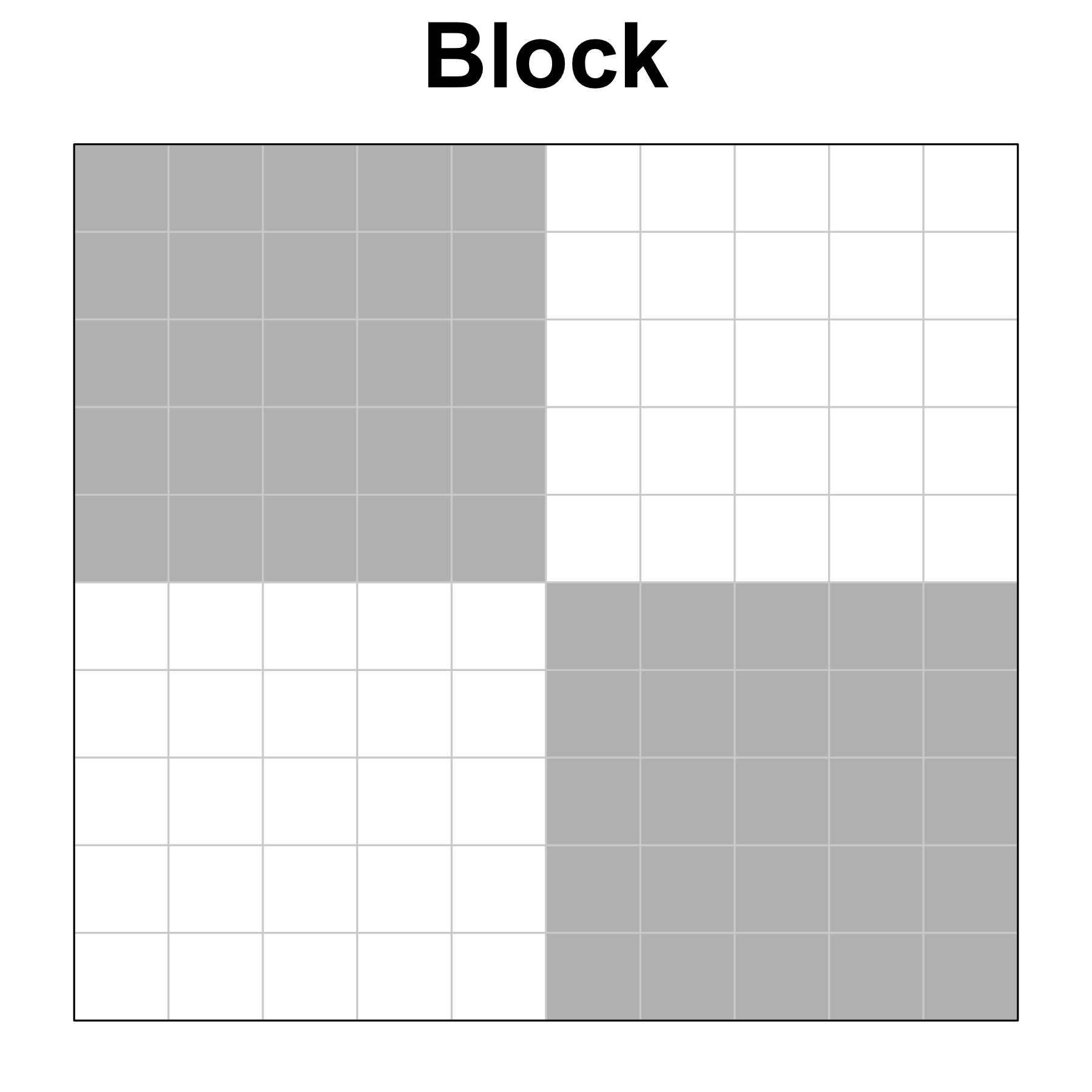}
\end{subfigure}
\begin{subfigure}[b]{0.13\textwidth}
\centering
\includegraphics[width = \textwidth]{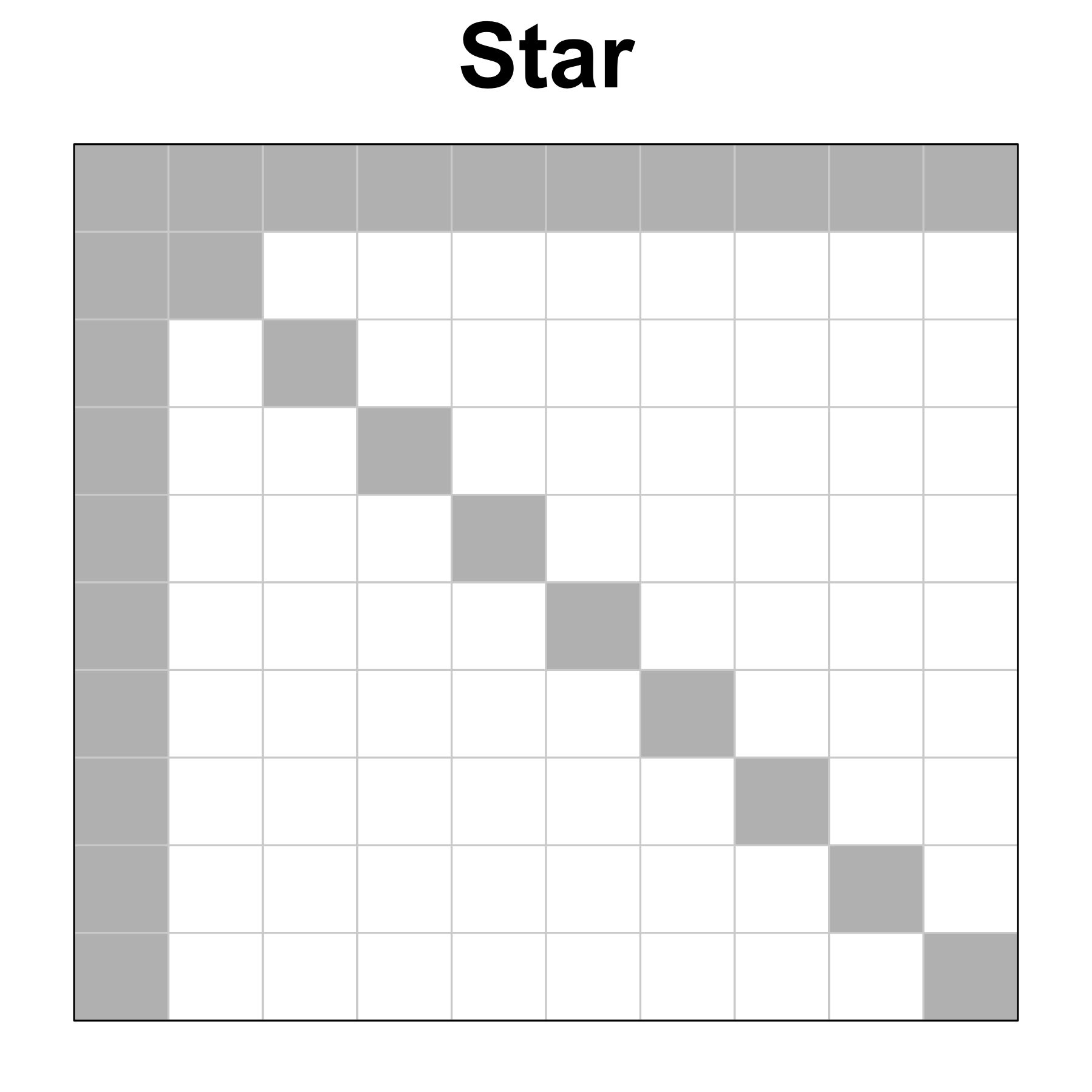}
\end{subfigure}
\begin{subfigure}[b]{0.13\textwidth}
\centering
\includegraphics[width = \textwidth]{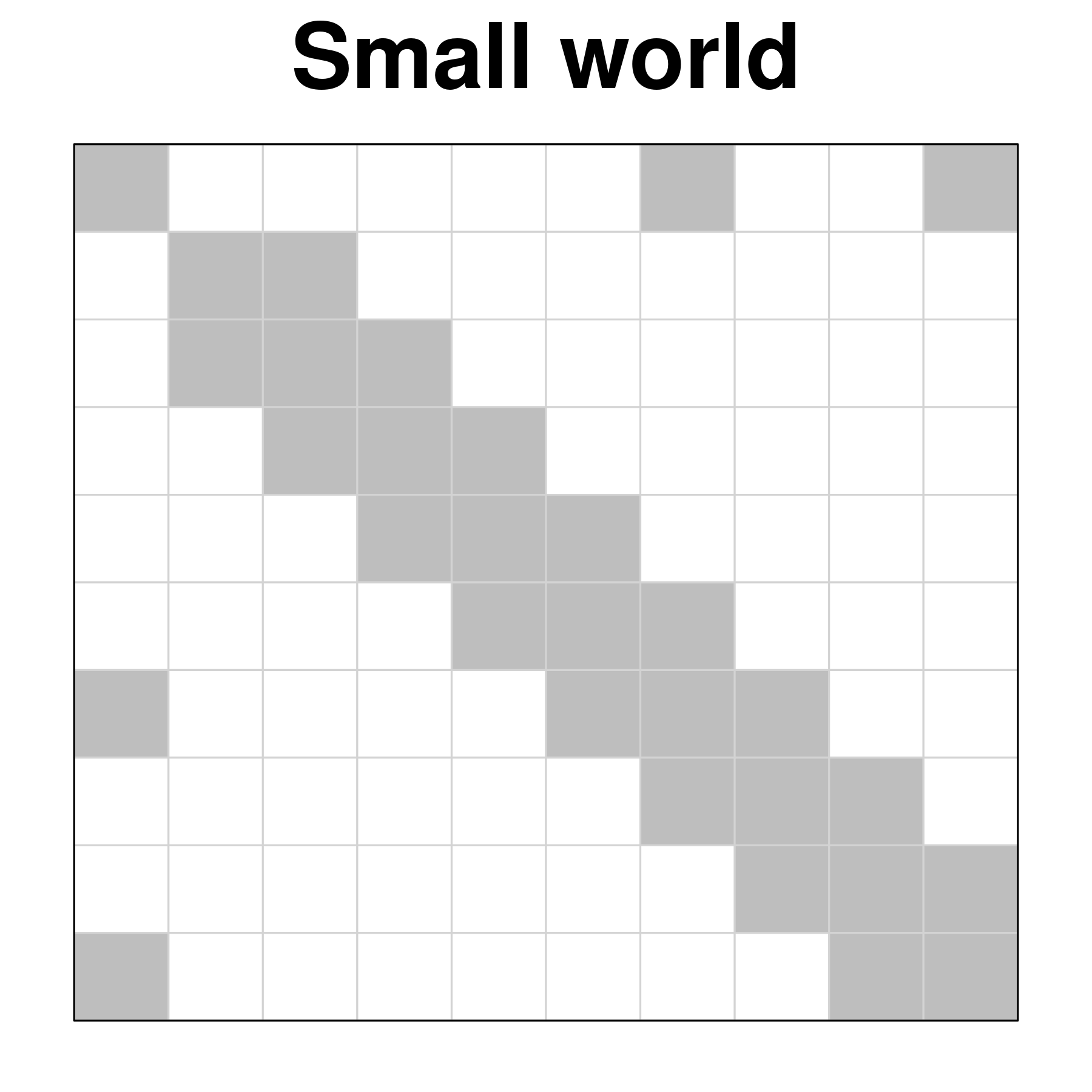}
\end{subfigure}	
\begin{subfigure}[b]{0.13\textwidth}
\centering
\includegraphics[width = \textwidth]{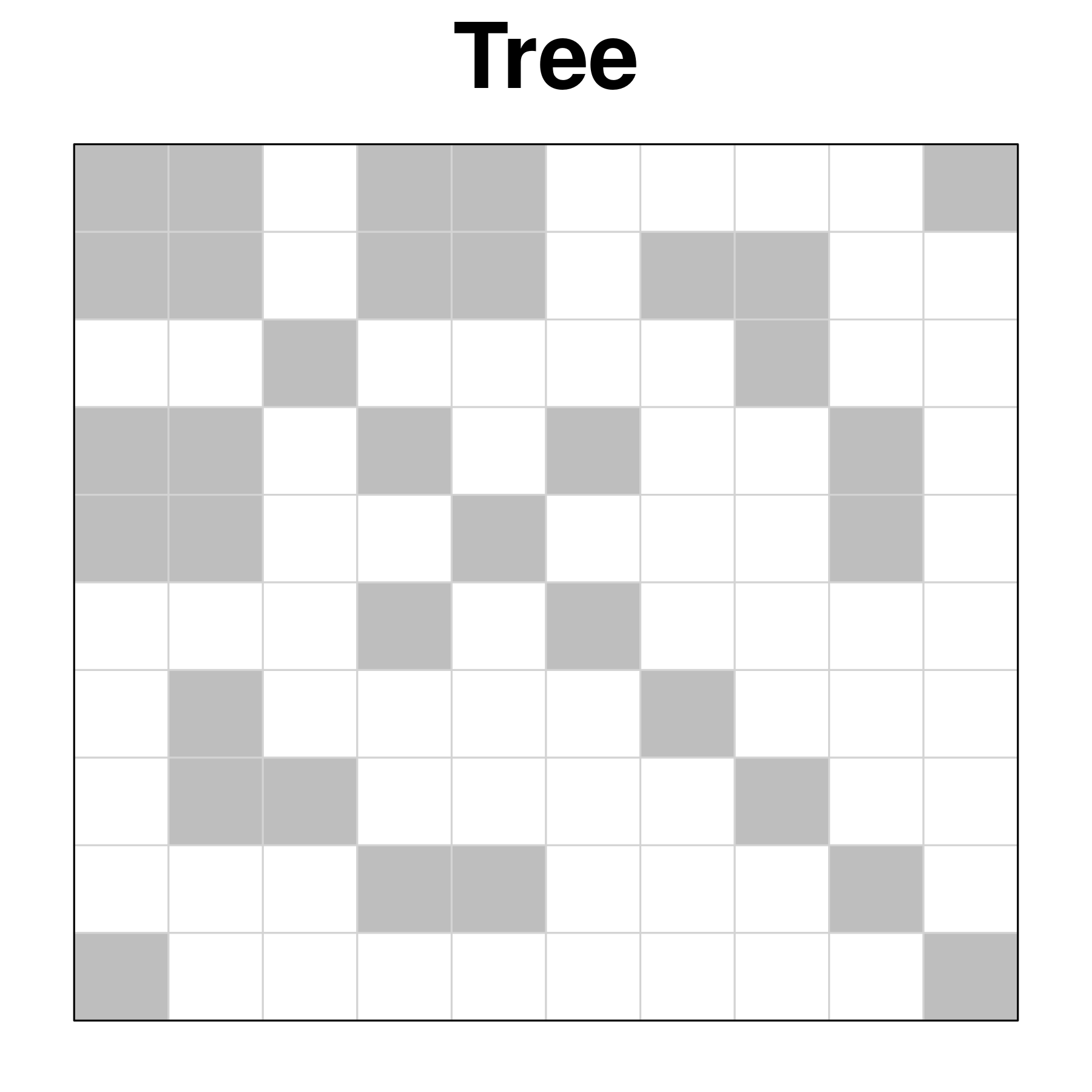}
\end{subfigure}	
\begin{subfigure}[b]{0.13\textwidth}
\centering
\includegraphics[width = \textwidth]{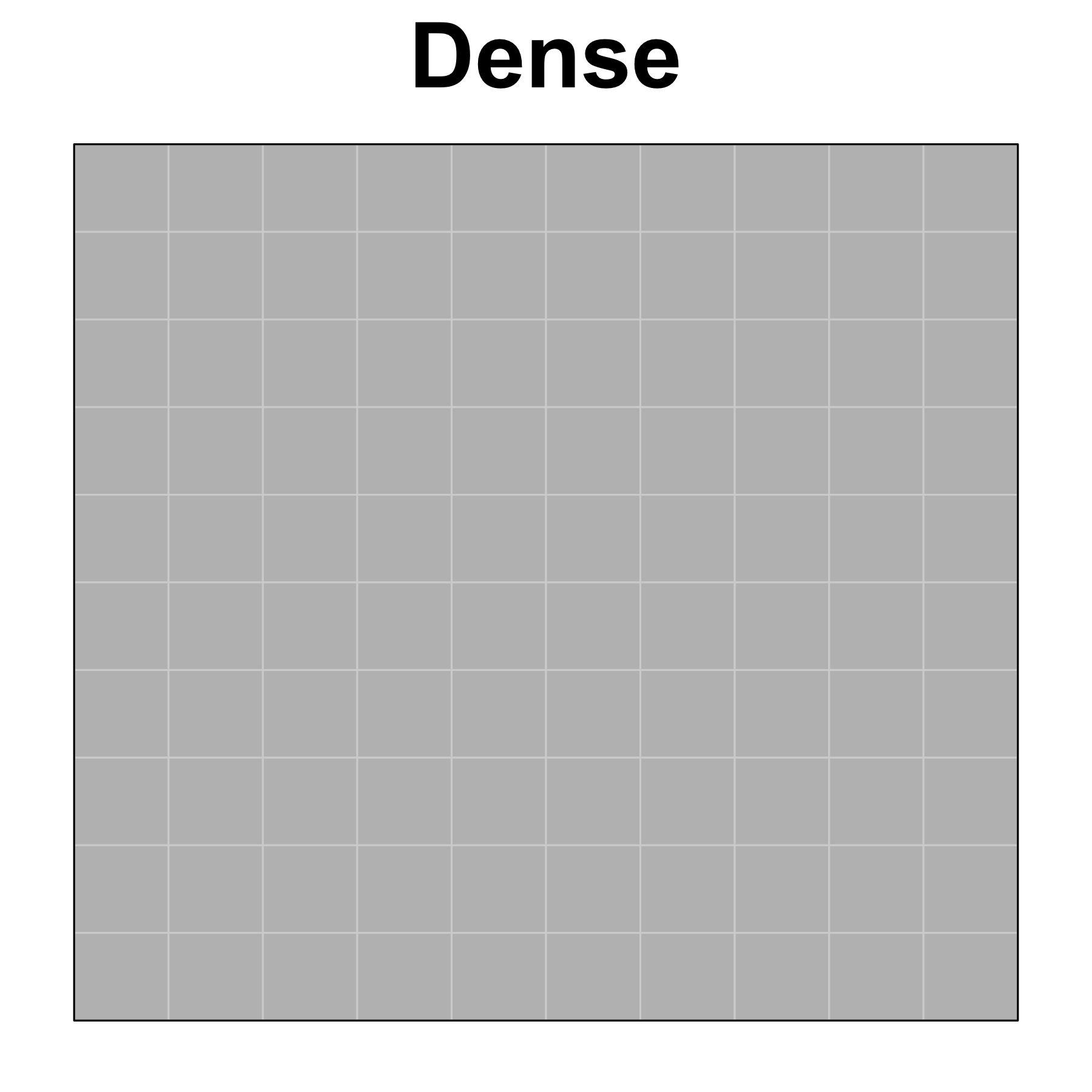}
\end{subfigure}	

\begin{subfigure}[b]{0.13\textwidth}
\centering
\includegraphics[width = \textwidth]{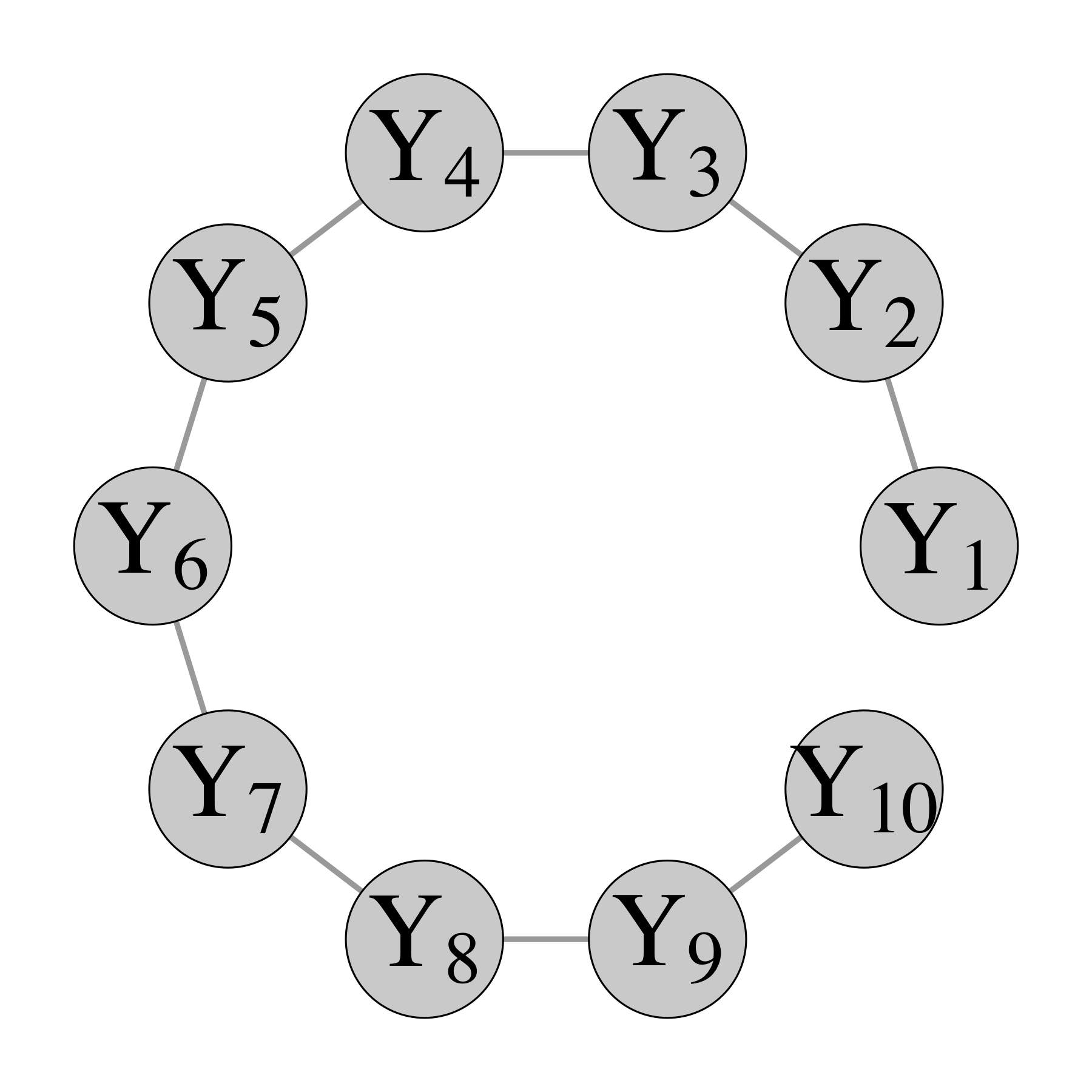}
\end{subfigure}
\begin{subfigure}[b]{0.13\textwidth}
\centering
\includegraphics[width = \textwidth]{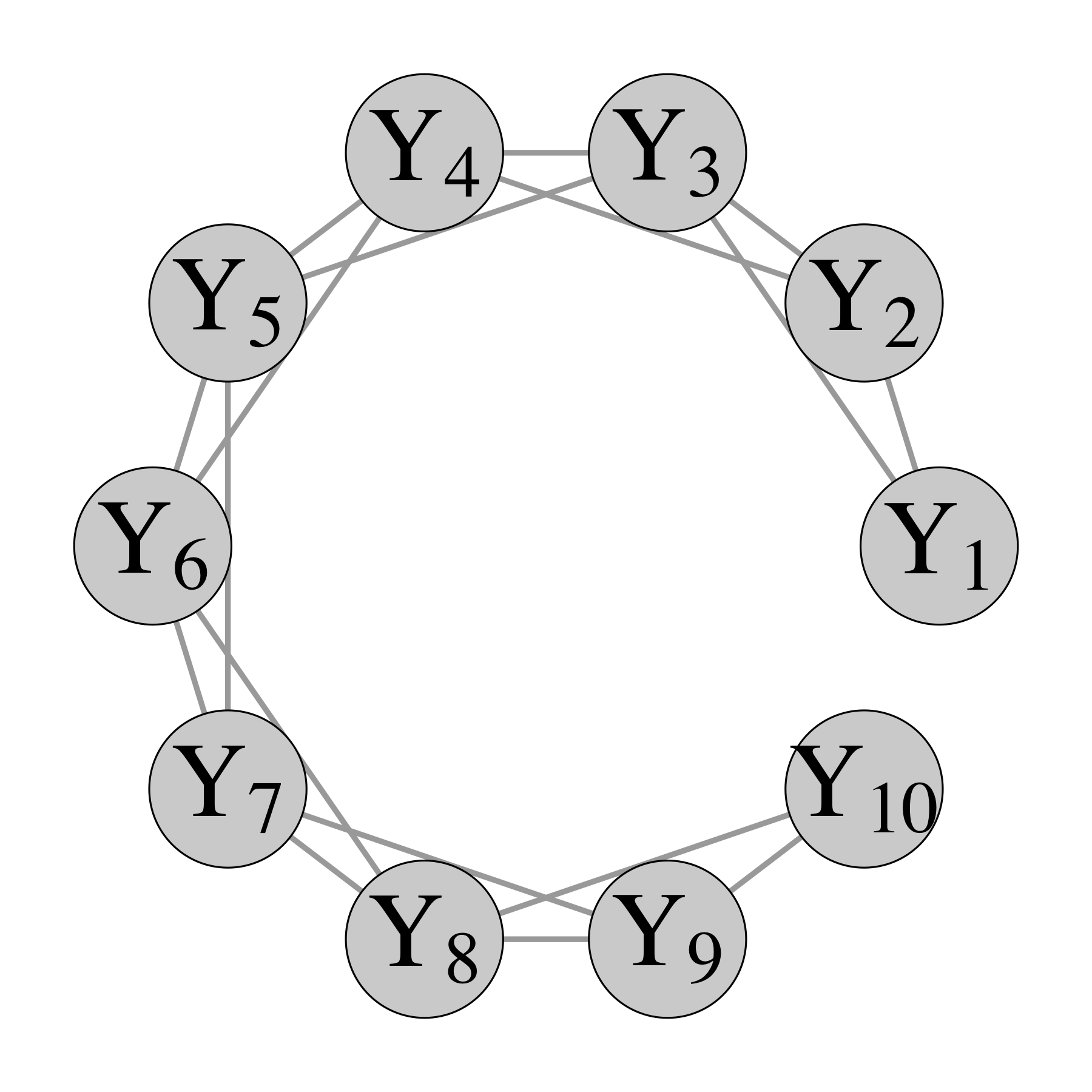}
\end{subfigure}
\begin{subfigure}[b]{0.13\textwidth}
\centering
\includegraphics[width = \textwidth]{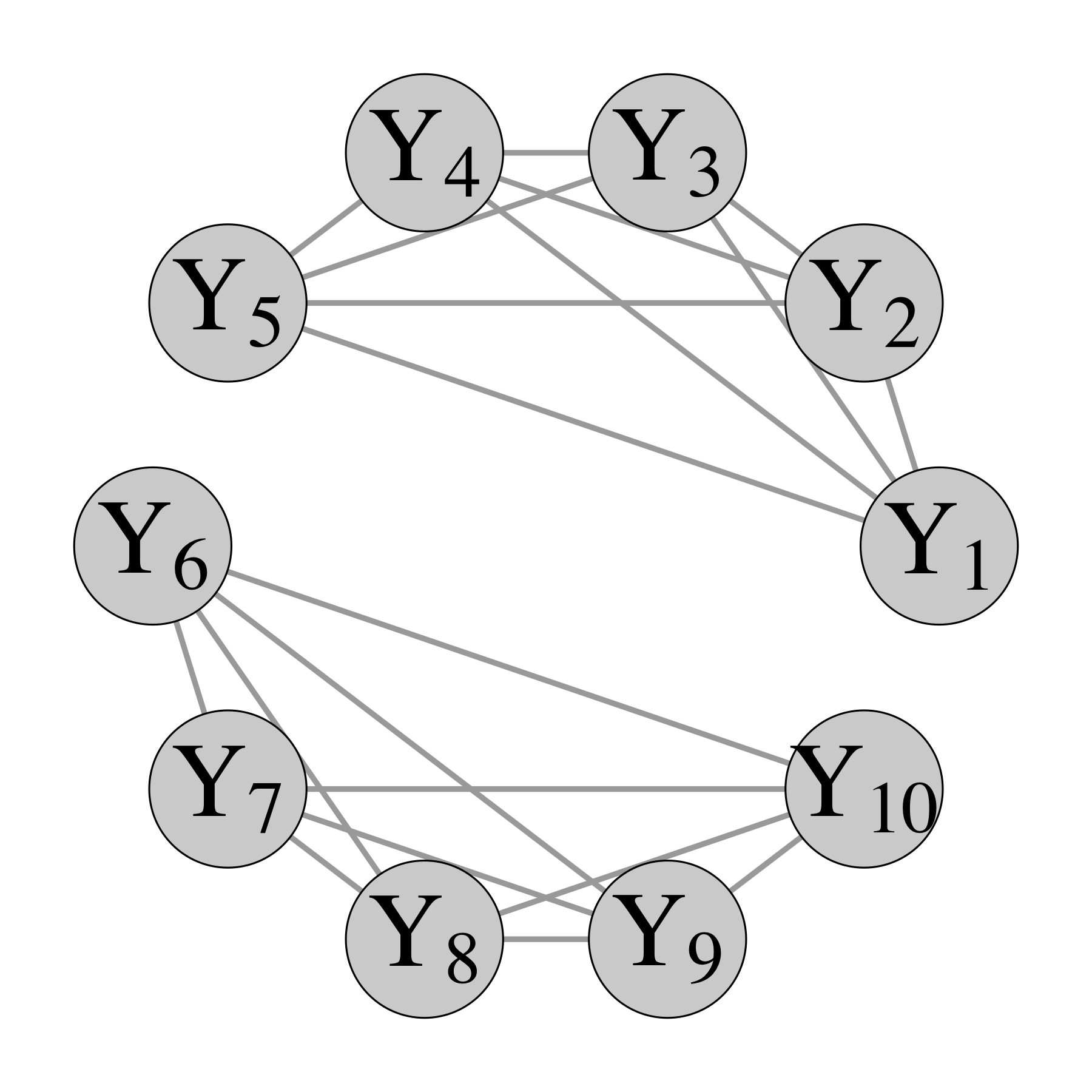}
\end{subfigure}
\begin{subfigure}[b]{0.13\textwidth}
\centering
\includegraphics[width =\textwidth]{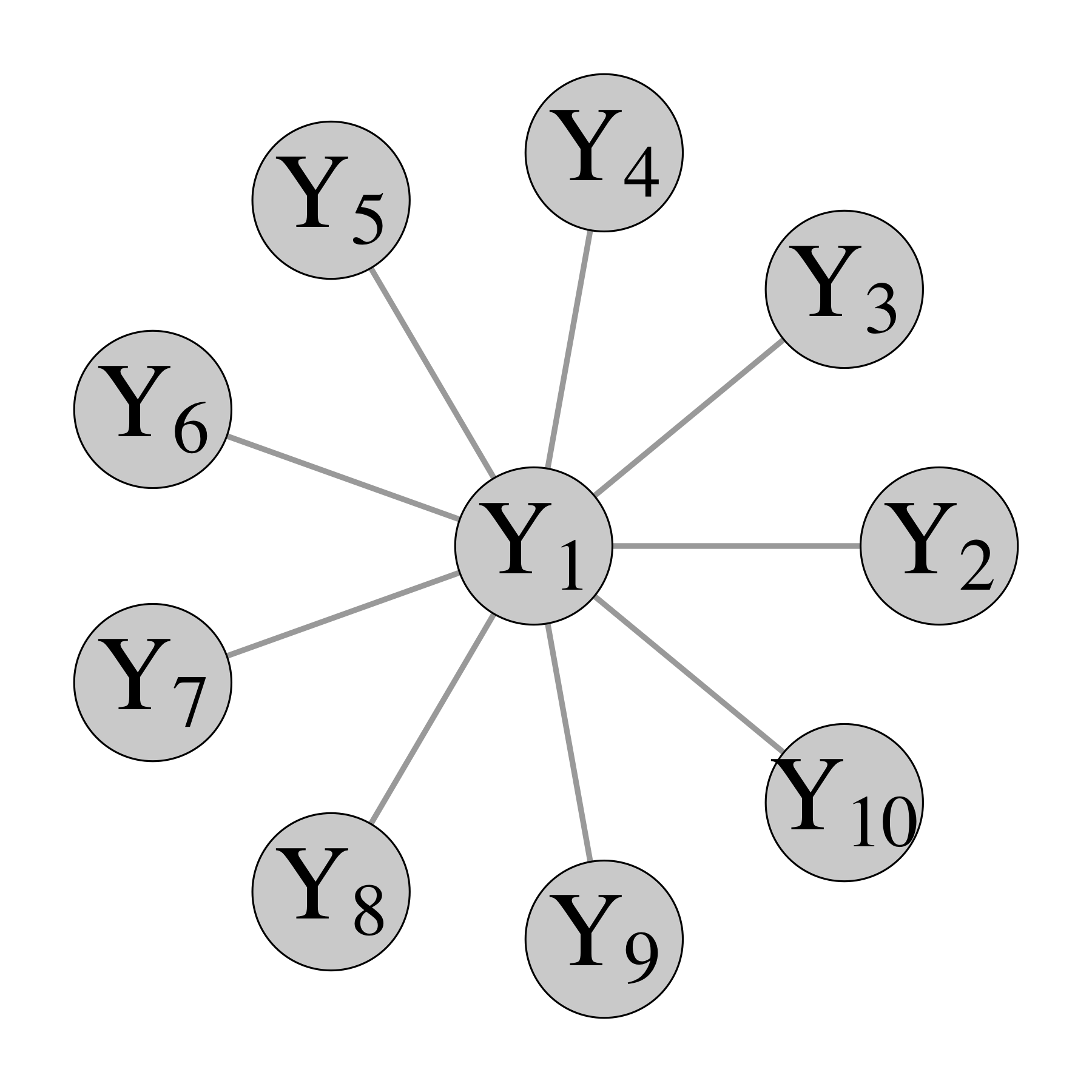}
\end{subfigure}
\begin{subfigure}[b]{0.13\textwidth}
\centering
\includegraphics[width = \textwidth]{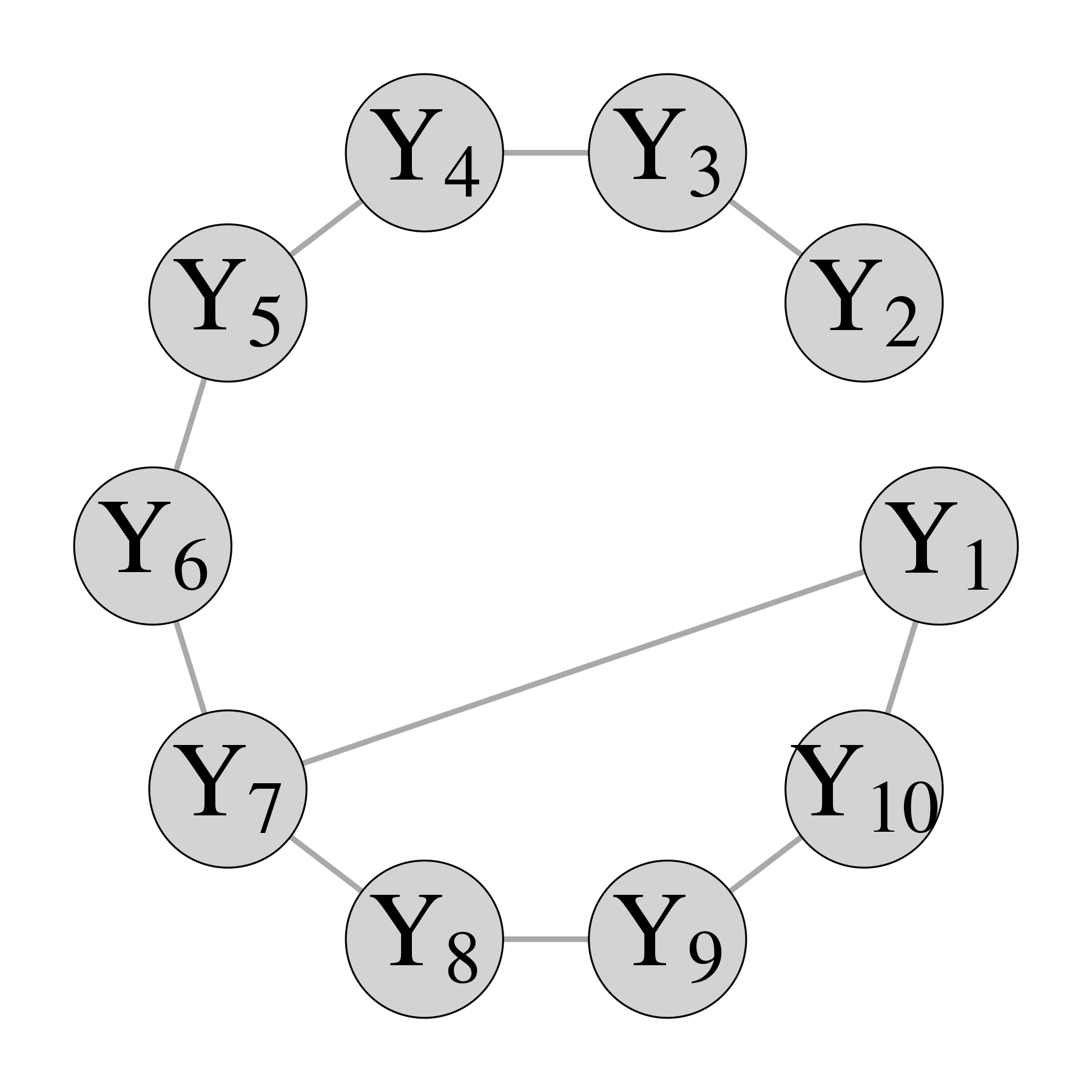}
\end{subfigure}	
\begin{subfigure}[b]{0.13\textwidth}
\centering
\includegraphics[width = \textwidth]{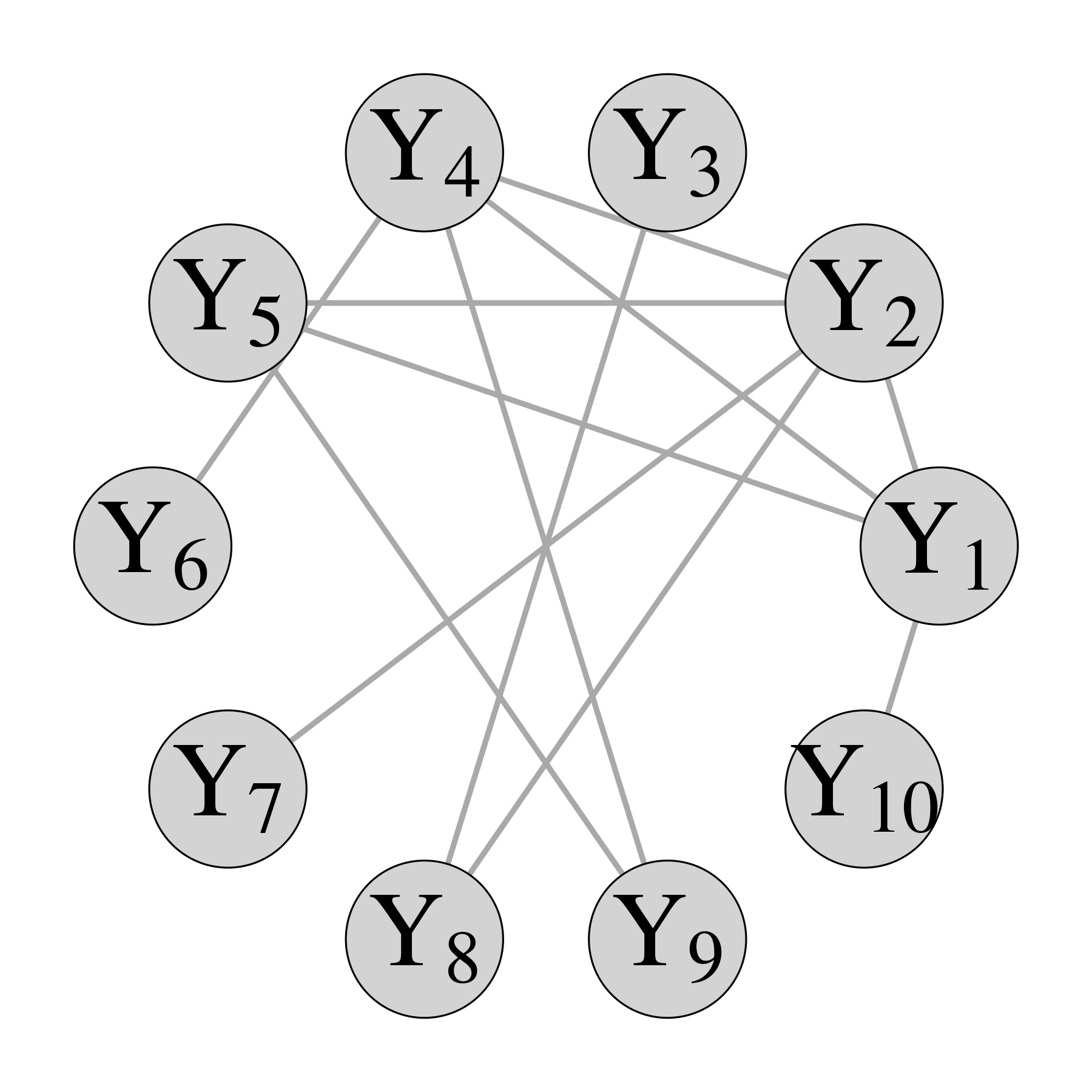}
\end{subfigure}	
\begin{subfigure}[b]{0.13\textwidth}
\centering
\includegraphics[width = \textwidth]{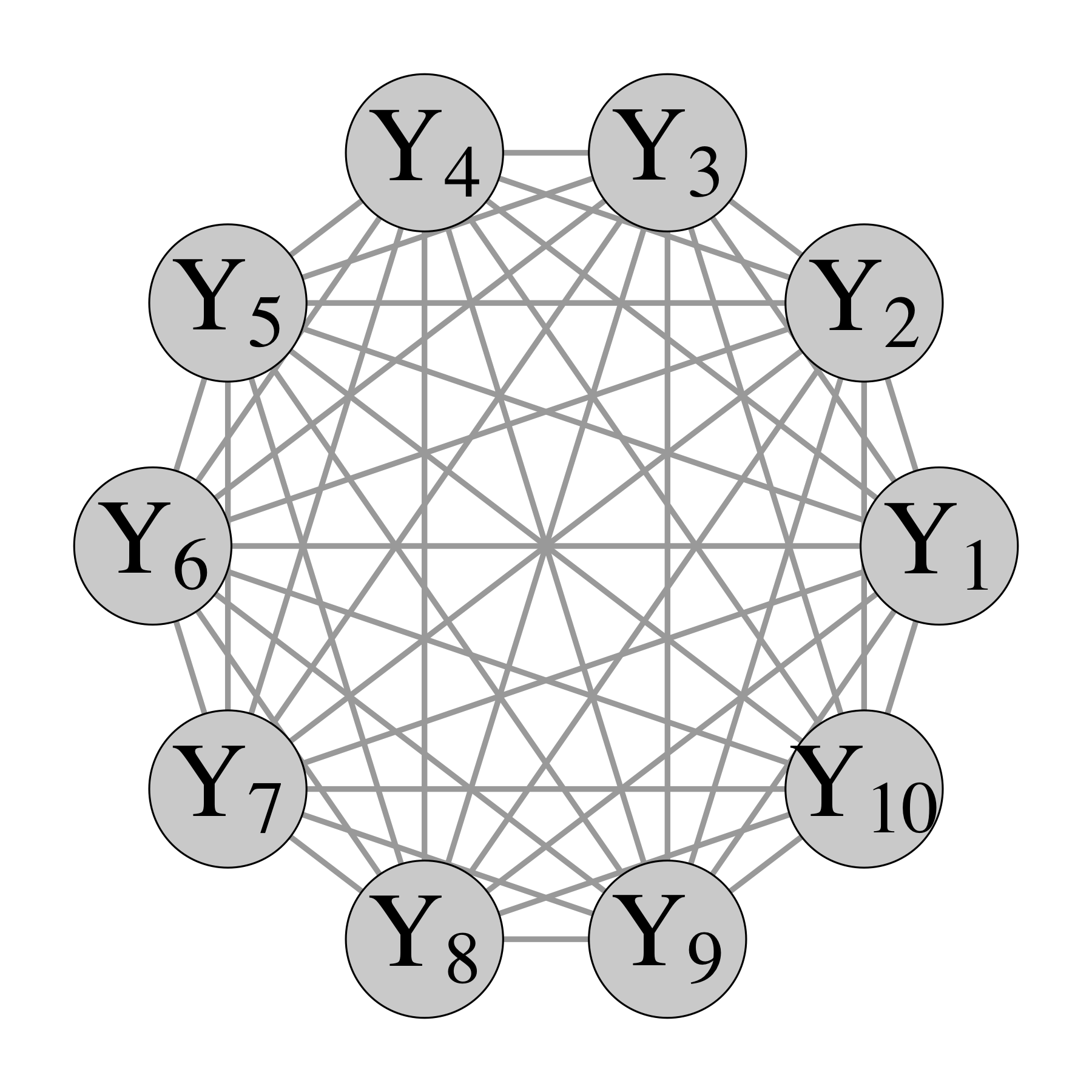}
\end{subfigure}	
\caption{Supports of each $\Omega$ for $q = 10$ (top) and corresponding graph (bottom). Gray cells in top row indicate non-zero entries $\omega_{k,k'}$ and white cells indicate zeros}
\label{fig:simulation_design}
\end{figure}

In total, we considered 21 combinations of dimensions $(n,p,q)$ and $\Omega.$
We generated $\Psi$ by randomly selecting 20\% of entries to be non-zero and drawing the non-zero entries uniformly from [-2,2].
For each combination of $(n,p,q), \Omega$ and $\Psi,$ we generated 100 synthetic datasets from the Gaussian chain graph model in Equation~\eqref{eq:cg_model}.
The design matrix $X$ contained independent standard normal entries.

\subsection{Results}

To assess estimation performance, we computed the Frobenius norm between the estimated matrices and the true data generating matrices.
We additionally computed the coverage of our 95\% bootstrap intervals, averaged over all entries in $\Psi$ and $\Omega$ (\texttt{cgSSL-dpe+BB})
To assess the support recovery performance, we counted the number of elements in each of $\Psi$ and $\Omega$ that were (i) correctly estimated as non-zero (true positives; TP); (ii) correctly estimated as zero (true negatives; TN); (iii) incorrectly estimated as non-zero (false positives; FP); and (iv) incorrectly estimated as zero (false negatives; FN).
We report the sensitivity (TP/(TP + FN)) and precision (TP/(TP + FP)). 
We also report the sensitivity and precision of estimating the supports of $\Psi$ and $\Omega$ by checking whether zero is contained in the bootstrap interval for each $\psi_{j,k}$ and $\omega_{k,k'}.$
Generally speaking, we prefer methods with high sensitivity and high precision.
High sensitivity indicates that the method has correctly estimated most of the true non-zero parameters as non-zero.
High precision, on the other hand, indicates that most of the estimated non-zero parameters are truly non-zero.
For brevity, we only report results for the $(n,p,q) = (400,100,30)$ setting in Table~\ref{tab:psi_omega_relationship}.
Results for the other dimensions were similar; see Tables~S1 and S2 of the Supplementary Materials .

We performed all our experiments in a shared high-throughput computing environment \citep{chtc} on nodes with 5 GB RAM and 2 CPU cores running \textsf{R} \citep[v.\ 4.13;][]{R_citation}.
We ran \texttt{cgSSL-DPE} using the default hyperparameter settings recommended in Section~\ref{sec:map_estimation}.
Similarly, we ran \texttt{mSSL} using the default settings that were recommended in \citet{Deshpande2019}. 
For each MCMC method (\texttt{CAR}, \texttt{CAR-A}, and \texttt{OFBF}), we ran four Markov chains for 10,000 iterations each, discarding the first 5,000 samples from each as ``burn-in,'' and retaining all subsequent samples.
In general, although the simulated Markov chains did not mix (see below for more discussion), we nevertheless report the results based on the 20,000 obtained samples. 
Because \texttt{cgLASSO}'s cross-validation step often did not finish within 72 hours (the maximum time limit set by our cluster) when run serially, we parallelized our \texttt{cgLASSO} implementation, allowing our cluster to schedule separate jobs running each combination of fold and penalty.
Because the scheduler sometimes delayed running certain jobs, we were unable to reliably time \texttt{cgLASSO} and do not report its runtimes in Table~\ref{tab:psi_omega_relationship}.
Extrapolating from preliminary runs, however, we estimate that running one full fold would take around ten hours. 

In terms of identifying non-zero direct effects (i.e., estimating the support of $\Psi$), \texttt{cgLASSO} consistently achieved the highest sensitivity. 
But, as the precision results indicate, the majority of \texttt{cgLASSO}'s ``discoveries'' were in fact false positives.
On further inspection, we determined that such behavior was the result of the cross-validation step in \texttt{cgLASSO}, which tended to select very small penalty values that promoted very little shrinkage.
The other fixed penalty method, \texttt{CAR}, similarly displayed high sensitivity and low precision.
In contrast, methods that deployed adaptive penalties (\texttt{CAR-A} and \texttt{cgSSL-DPE}), displayed higher precision in estimating the support of $\Psi.$
In fact, at least for estimating the support of $\Psi,$ \texttt{cgSSL-DPE} made no false positives in the vast majority of simulation replications.

We observed essentially the same phenomenon for $\Omega$: although \texttt{cgSSL-DPE} generally returned fewer non-zero estimates of $\omega_{k,k'}$, the vast majority of these estimates were true positives.
In a sense, the fixed penalty methods (\texttt{cgLASSO} and \texttt{CAR}) cast a very wide net when searching for non-zero signal in $\Psi$ and $\Omega,$ leading to large number of false positive identifications in the supports of these matrices.
Adaptive penalty methods, on the other hand, were much more discerning. 
Interestingly, \texttt{mSSL} recovered $\Omega$'s support substantially better than \texttt{OBFB}. 
We suspect that the discrepancy stems from the fact that \texttt{mSSL} deploys adaptive penalization while \texttt{OBFB} utilizes a fixed Wishart prior for $\Omega.$

\newpage
\begin{table}[H]
\centering
\caption{Sensitivity, precision, and Frobenius error for $\Psi$ and $\Omega$ when $(n,p,q) = (400, 100, 30).$ For each $\Omega,$ the best performance is bold-faced. For dense $\Omega,$ \texttt{cgLASSO} with tuned penalties did not converge within 72 hours.}
\label{tab:psi_omega_relationship}
\footnotesize
\begin{tabular}{lccccccc}
\hline
~ & \multicolumn{3}{c}{$\Psi$ recovery} & \multicolumn{3}{c}{$\Omega$ recovery} & Runtime \\ 
Method & SEN & PREC & FROB & SEN & PREC & FROB & Time (min) \\ \hline
\multicolumn{8}{c}{$AR(1)$ model} \\ \hline
\texttt{cgLASSO} & \textbf{1} & 0.2 & 0.07 & 0.94 & 0.46 & 28.0 & --\\
\texttt{CAR} & 0.82 & 0.46 & 0.02 & \textbf{1} & 0.27 & \textbf{2.2} & 472.7 \\
\texttt{CAR-A} & 0.86 & 0.73& 0.01& \textbf{1} & \textbf{0.89} & 7.3 & 564.8 \\
\texttt{OBFB} & -- & -- & -- &  0.07 &0.08 & --  & 54.0 \\
\texttt{cgSSL} & 0.87 & \textbf{0.99} & \textbf{0.00} & \textbf{1}& 0.78 & 3.4 & \textbf{26.7} \\ 
\texttt{cgSSL+BB} & 0.88 & \textbf{0.99} & -- & \textbf{1} & 0.83 & -- & 55.4 \\ 
\texttt{mSSL} & 0.95 & 0.25 & 0.06  & \textbf{1} & 0.82 & 9.9 & 51.6 \\ \hline
\multicolumn{8}{c}{$AR(2)$ model} \\ \hline
\texttt{cgLASSO} & \textbf{1} & 0.2  & 0.2 & 0.79 & 0.63  & 10.8 & --\\
\texttt{CAR} & 0.85 & 0.5  & 0.01 & 0.98 & 0.49 & 0.4 & 493.8 \\
\texttt{CAR-A} & 0.89 & 0.77 & 0.01 & 1 & 0.94 & 1.2 & 458.1 \\
\texttt{OBFB} & -- & -- & -- & 0.06 &0.13 & -- & 49.6 \\
\texttt{cgSSL} & 0.92& \textbf{1} & \textbf{0.00} & \textbf{1 }& 0.47 & \textbf{0.3} & 13.0 \\ 
\texttt{cgSSL+BB} & 0.92& \textbf{1} & -- & \textbf{1} & 0.61 & --  & 30.3 \\ 
\texttt{mSSL} & 0.99 & 0.23 & 0.03 & 1 & \textbf{0.95} & 2.5 & \textbf{0.08} \\ \hline
\multicolumn{8}{c}{Block model} \\ \hline
\texttt{cgLASSO} & \textbf{1} & 0.20 & 0.4 & 0.87& 0.97 & 10.1 & --\\
\texttt{CAR} & 0.84 & 0.46 & 0.02 & 0.71 & 0.76 & 3.4 & 480.6 \\
\texttt{CAR-A} & 0.88 & 0.70 & \textbf{0.01} & 0.75 & \textbf{0.99} & 4.1 & 512.4 \\
\texttt{OBFB} & -- & -- & -- & 0.06 &0.50 &  -- & 74.5\\
\texttt{cgSSL} & 0.86 & \textbf{0.99} & \textbf{0.01} & 0.98 & 0.98& \textbf{1.4} & 17.3 \\ 
\texttt{cgSSL+BB} & 0.88 & \textbf{0.99} & -- & \textbf{0.99} & 0.98 & -- &  42.6 \\ 
\texttt{mSSL} & 0.99 & 0.21 & 0.1 & 0.44 & \textbf{0.99} & 28.4 & \textbf{2.0} \\ \hline
\multicolumn{8}{c}{Star model} \\ \hline
\texttt{cgLASSO} & \textbf{0.93} & 0.83 & 0.01 & 0.53  & 0.59 & 4.7 & -- \\
\texttt{CAR} & 0.89 & 0.48 & 0.01& 0.73 & 0.25 & 0.6 & 493.2 \\
\texttt{CAR-A} & 0.90 & 0.70  & 0.01 & 0.87  & 0.74 & 1.1 & 431.8 \\
\texttt{OBFB} & -- & -- & -- &  0.06 & 0.08 & -- & 65.2 \\
\texttt{cgSSL} & 0.89 & \textbf{1} & \textbf{0.00} & \textbf{1} & 0.90 & \textbf{0.3} & 1.1 \\ 
\texttt{cgSSL+BB} & 0.89 & \textbf{1} & -- & \textbf{1} & 0.90 & -- & 3.7 \\ 
\texttt{mSSL} & 0.90  & 0.85 & 0.02 & \textbf{1}& \textbf{1}  & 0.7 & \textbf{0.1} \\ \hline
\multicolumn{8}{c}{Small world model} \\ \hline
\texttt{cgLASSO} & \textbf{0.99} & 0.20 & 0.6 & 0.38 & 0.49 & 468.1 & -- \\
\texttt{CAR} & 0.82 & 0.43  & 0.03  & 0.92  & 0.22 & \textbf{10.7} & 633.7 \\
\texttt{CAR-A} & 0.85 & 0.68 & \textbf{0.01} & 0.92 & \textbf{0.79} & 29.3 & 431.5 \\
\texttt{OBFB} & -- & -- & -- &  0.06 &0.08 & --  & 47.7 \\
\texttt{cgSSL} & 0.82& \textbf{0.99} & \textbf{0.01} & \textbf{0.95 } & 0.78 & 25.8 & 359.7 \\ 
\texttt{cgSSL+BB} &0.82 & \textbf{0.99} & -- & \textbf{0.95} & 0.78 & -- & 749.6 \\
\texttt{mSSL} & 0.88 & 0.34 & 0.1& 0.58 & 0.7 & 122.8 & \textbf{10.9} \\\hline
\multicolumn{8}{c}{Tree model} \\ \hline
\texttt{cgLASSO} & \textbf{0.99} & 0.20 & 0.6 & 0.69 & 0.48 &  381.0 & --\\
\texttt{CAR} & 0.79 & 0.46 & 0.03 & 0.95  & 0.24 & \textbf{14.2} & 716.1 \\
\texttt{CAR-A} & 0.84 & 0.72  & 0.02  & 0.95& \textbf{0.86} & 22.7 &  519.9\\
\texttt{OBFB} & -- & -- & -- &  0.07 &0.08 & -- & 44.7 \\ 
\texttt{cgSSL} & 0.84 & \textbf{0.99} & \textbf{0.01} & \textbf{0.97} & 0.61  & 18.8 & 398.6 \\ 
\texttt{cgSSL+BB} & 0.84 & \textbf{0.99} & -- & 0.96 & 0.61 & -- & 1139.1 \\
\texttt{mSSL} &  0.92 & 0.28 & 0.2  & 0.9 & 0.76& 25.1 & \textbf{28.3} \\ \hline
\multicolumn{8}{c}{Dense model} \\ \hline
\texttt{cgLASSO} & -- & -- & -- & -- & -- & --  & --\\
\texttt{CAR} & 0.87 & 0.39  & \textbf{0.01} & 0 & -- & 964.2 & 522.7 \\
\texttt{CAR-A} & 0.88  & 0.52  & \textbf{0.01}  & 0 & -- & 970.0 &  431.8\\
\texttt{OBFB} & -- & -- & -- & 0.06 &\textbf{1} & --& 53.8 \\
\texttt{cgSSL} & 0.86  & \textbf{0.98} & 0.04 & \textbf{0.26} & \textbf{1} & \textbf{918.4} & \textbf{1.7} \\ 
\texttt{cgSSL+BB} & 0.86 & \textbf{0.98} & -- & 0.24 & \textbf{1} & -- & 10.2 \\ 
\texttt{mSSL} & \textbf{0.96} & 0.27 & 0.06 & 0.18  & \textbf{1} & 960.0 & 8.7  \\ \hline
\end{tabular}
\end{table}

In terms of estimation performance, with the exception of the dense $\Omega$ setting, the fixed penalty methods tended to have much larger Frobenius error than the adaptive penalty methods.
Interestingly, for the six sparse $\Omega$'s, no method had high Frobenius for $\Omega$ but low Frobenius error for $\Psi.$
This finding corroborates our intuition about Corollary~\ref{coro:contraction_B_cg}: in order to estimate $\Psi$ well, we must estimate $\Omega$ well.
Additionally, our bootstrap intervals for individual parameters were reasonably well-calibrated and achieved near-nominal coverage (0.9 for $\Psi$ and $\Omega$).
Using these intervals for support recovery was comparable to \texttt{cgSSL-DPE}.

Across all choices of $\Omega,$ \texttt{cgSSL-DPE} was faster than the two MCMC methods.
As alluded to above, the Markov chains simulated by those methods did not appear to have mixed, even after 10,000 iterations: across our experiments, about 25\% of the $\omega_{k,k'}$'s had effective sample sizes less than 1,000 and around 5\% of the parameters had marginal $\hat{R}$ estimates in excess of $1.1.$
Interestingly, \texttt{mSSL} was sometimes faster than \texttt{cgSSL-DPE}, depending on $\Omega.$

\section{Re-analysis of the gut microbiome data with cgSSL}
\label{sec:real_data_experiments}
\citet{claesson2012gut} studied the gut microbiota of elderly adults using data sequenced from fecal samples taken from 178 subjects.
They were primarily interested in understanding differences in the gut microbiome composition across several residence types (in the community, day-hospital, rehabilitation, or in long-term residential care) and across several different types of diet.
They found that the gut microbiomes of residents in long-term care facilities were considerably less diverse than those of residents dwelling in the community.
They additionally reported that diet had a large marginal effect on gut microbe diversity but they did not examine direct effects, which might align more closely with the underlying biological mechanism.
We re-analyzed their data using the cgSSL to estimate the direct effects of each type of diet and residence type on gut microbiome composition.
Before proceeding, we note that while raw microbiome data consists of counts, we used cgSSL to model the logarithms of relative abundances of each taxa. 
Section S4 of the Supplementary Materials  describes how we pre-processed the raw 16s-rRNA data to obtain these log-abundances.

In all, the dataset contains $n = 178$ observations of $p = 11$ predictors and $q = 14$ outcomes.
We computed two graphs for these data, which are shown in Figure~\ref{fig:humangut2}.
In Figure~\ref{fig:humangut2_cgSSL}, edges correspond to the estimated non-zero entries of $\Psi$ and $\Omega$ returned by cgSSL-DPE.
Edges in Figure~\ref{fig:humangut2_cgSSL_BB} instead correspond to those parameters whose bootstrapped uncertainty intervals did not contain zero.

In both graphs in Figure~\ref{fig:humangut2}, we observed many more edges between the different species (corresponding to non-zero $\omega_{k,k'}$'s) than edges between covariates and species (corresponding to non-zero $\psi_{j,k}$'s).
In both graphical models, we estimated that percutaneous endoscopic gastronomy (PEG), in which a feeding tube is inserted into the abdomen, had a direct effect on the abundance of \textit{Veillonella}, which is involved in lactose fermentation. 
Our findings reassuringly align with those in \citet{takeshita2011enteral}, who reported a negative effect of PEG on this genus.
Although cgSSL-DPE additionally identified staying in a day hospital as having a direct effect on \textit{Caloramator}, the corresponding bootstrap interval contained zero.

\begin{figure}[ht]
	\centering
	\begin{subfigure}[b]{2.4in}
		\centering
		\includegraphics[width=2.4in]{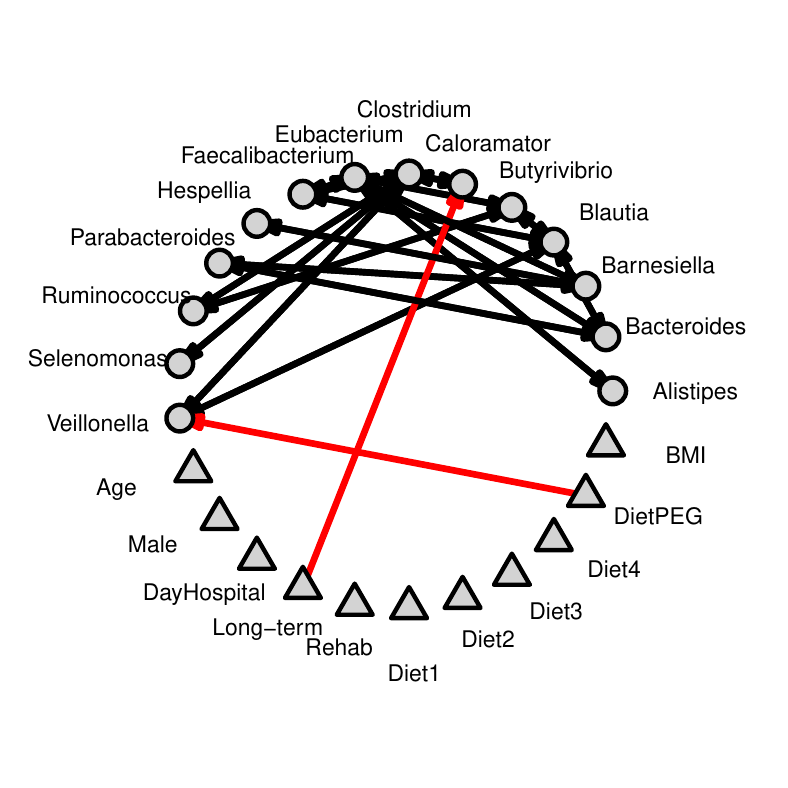}
		\caption{}
		\label{fig:humangut2_cgSSL}
		\end{subfigure}
	\begin{subfigure}[b]{2.4in}
		\centering
		\includegraphics[width=2.4in]{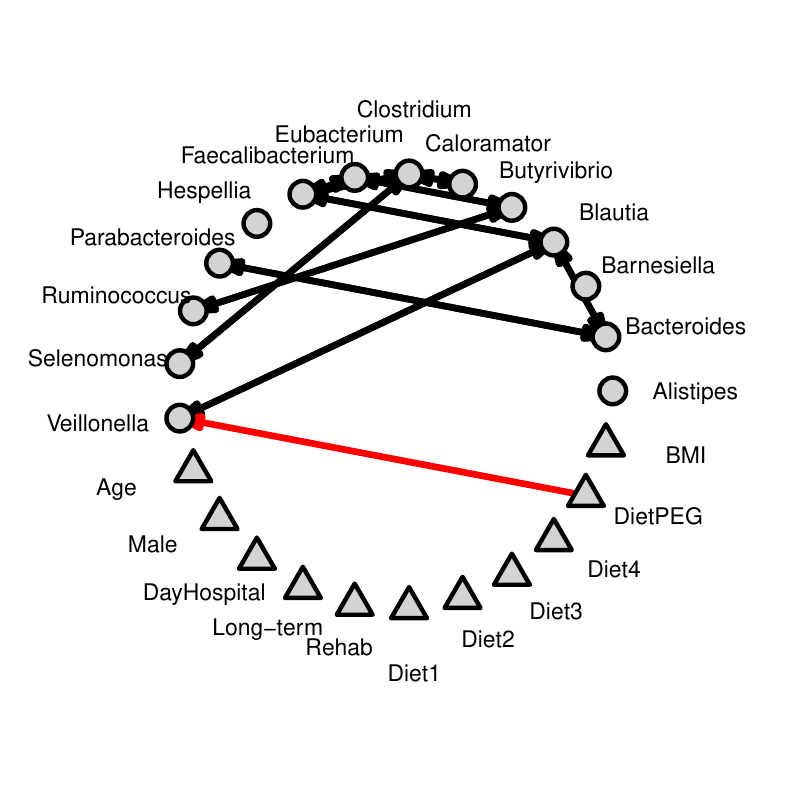}
		\caption{}
		\label{fig:humangut2_cgSSL_BB}
	\end{subfigure}
	\caption[microbe network]{Graphical models for \citet{claesson2012gut}'s gut microbiome dataset estimated by cgSSL-DPE (a) and cgSSL-DPE followed by the weighted Bayesian bootstrap (b). Triangles correspond to covariates and circles correspond to responses. Edges corresponding to non-zero $\psi_{j,k}$'s are colored red.}
	\label{fig:humangut2}
\end{figure}

Our results suggest that the large marginal effects reported by \citet{claesson2012gut} are a by-product of only a few direct effects and substantial residual conditional dependence between species.
For instance, because PEG has a direct effect on \textit{Veillonella}, which is conditionally correlated with \textit{Clostridium}, \textit{Butyrivibrio}, and \textit{Blautia}, PEG displays a marginal effect on each of these other genus.
In this way, the cgSSL can provide a more nuanced understanding of the underlying biological mechanism than simply estimating the matrix of marginal effects $B = \Psi\Omega^{-1}.$
We note, however, that \citet{claesson2012gut}'s dataset does not contain an exhaustive set of environmental and patient life-style predictors.
Accordingly, our re-analysis is limited in the sense that were we able to incorporate additional predictors, the estimated graphical model may be quite different.
Further, although we followed a relatively standard pre-processing workflow, we anticipate that our overall findings will be somewhat sensitive to some of the choices made in converting the raw microbiome data into the log abundances that we modeled.

\section{Discussion}
\label{sec:discussion}
In the Gaussian chain graph model in Equation~\eqref{eq:cg_model}, $\Psi$ is a matrix containing all of the direct effects of $p$ predictors on $q$ outcomes while $\Omega$ is the residual precision matrix that encodes the conditional dependence relationships between the outcomes that remain after adjusting for the predictors.
We have introduced the cgSSL procedure for obtaining simultaneously sparse estimates of $\Psi$ and $\Omega.$
In our procedure, we formally specify spike-and-slab LASSO priors on the free elements of $\Psi$ and $\Omega$ and use an ECM algorithm to maximize the posterior density.
Our ECM algorithm iteratively solves a sequence of penalized maximum likelihood problem with self-adaptive penalties. 
Across several simulated datasets, cgSSL demonstrated excellent support recovery and estimation performance, substantially out-performing competitors that deployed constant penalties.
We further characterized the asymptotic properties of cgSSL posteriors, establishing posterior contraction rates under relatively mild assumptions.
To the best of our knowledge, these are the first posterior contraction results for sparse Gaussian chain graph models with element-wise priors. 

Our main asymptotic result (Theorem~\ref{thm:posterior_contraction}) notwithstanding, quantifying the finite sample uncertainty around the MAP estimate returned by our cgSSL procedure remains a challenging problem.
We found that \citet{Newton2021}'s weighted Bayesian bootstrap can produce uncertainty intervals for individual parameters with close-to-nominal frequentist coverage.
Although the bootstrap does not exactly sample from the posterior, randomly re-centering the prior densities in Equation~\eqref{eq:master_for_bb}, as suggested by \citet{Nie2020}, may improve the approximation.
Understanding the extent to which solving randomly re-weighted and shifted MAP estimation problems faithfully approximate posterior sampling remains an important open question.

Although our motivating example involves only continuous outcomes, many applications feature outcomes of mixed type --- that is both continuous and discrete outcomes.
We anticipate that cgSSL could be extended to such settings using a strategy similar to that in \citet{Kowal2020}.
Specifically, one would model the discrete outcomes as a truncated and transformed latent Gaussian vector and fit a sparse Gaussian chain graphical model to the latent vector. 
The main challenge of such an extension lies in adaptively learning an appropriate transformation.


\section*{Acknowledgements}
The authors are grateful to Ray Bai and Bo Ning for helpful comments on the theoretical results and to Gemma Moran and Anna Menacher for their feedback on an early draft of the manuscript.

This work was supported by the National Institute of Food and Agriculture, United States Department of Agriculture, Hatch project 1023699.
This work was also supported by the Department of Energy [DE-SC0021016 to C.S.L.]. 
Support for S.K.D. was provided by the University of Wisconsin--Madison, Office of the Vice Chancellor for Research and Graduate Education with funding from the Wisconsin Alumni Research Foundation.

This research was performed using the compute resources and assistance of the UW--Madison Center For High Throughput Computing (CHTC) in the Department of Computer Sciences. 
The CHTC is supported by UW--Madison, the Advanced Computing Initiative, the Wisconsin Alumni Research Foundation, the Wisconsin Institutes for Discovery, and the National Science Foundation, and is an active member of the OSG Consortium, which is supported by the National Science Foundation and the U.S. Department of Energy's Office of Science.

\bibliographystyle{apalike}
\bibliography{references}

\renewcommand{\theequation}{S\arabic{equation}}
\renewcommand{\thesection}{S\arabic{section}}  
\renewcommand{\thefigure}{S\arabic{figure}}  
\renewcommand{\thetable}{S\arabic{table}} 
\renewcommand{\themyLemma}{S\arabic{myLemma}} 
\renewcommand{\themyTheorem}{S\arabic{myTheorem}} 
\renewcommand{\themyProposition}{S\arabic{myProposition}}
\setcounter{equation}{0}
\setcounter{section}{0}
\setcounter{subsection}{0}
\setcounter{subsubsection}{0}
\setcounter{myLemma}{0}
\setcounter{myTheorem}{0}

\newpage
\begin{center}
{\LARGE \textbf{Supplementary Materials}}
\end{center}

In Section~\ref{appendix:cgSSL_algorithm} we derive the Expectation Conditional Maximization (ECM) algorithm used to find the \textit{maximum a posteriori} (MAP) estimates of $\Psi$ and $\Omega$ in the cgSSL model.
One of the conditional maximization steps of that algorithm involves solving a CGLASSO problem.
We introduce a new algorithm, cgQUIC, to solve the general CGLASSO problem in Section~\ref{appendix:cgQUIC}. 
Specifically, we show that the problem has unique global optimum (Theorem~\ref{thm:cgquic_unique}) and that our cgQUIC algorithm converges to this optimum (Theorem~\ref{thm:cgquic_convergence}).
Then, we present additional results from the simulation study described in Section 5 of the main text in Section~\ref{appendix:synthetic_experiment}.
In Section~\ref{app:preprocessing}, we detail the preprocessing steps we took to prepare the gut microbiome data for analysis with the cgSSL.
Finally, we state and prove our main asymptotic results in Section~\ref{appendix:posterior_contraction}.

\section{The cgSSL algorithm}
\label{appendix:cgSSL_algorithm}
In this section, we provide full details of the Expectation Conditional Maximization (ECM) algorithm that is used in the cgSSL procedure.
We describe the algorithm for a fixed set of spike-and-slab penalties ($\lambda_{0},\lambda_{1}, \xi_{0},\xi_{1}$) and fixed set of hyperparameters ($a_{\theta}, b_{\theta}, a_{\eta}, b_{\eta}$).
For notational brevity, we will let $\Theta = \{\Psi, \theta, \Omega, \eta\}$ denote the four parameters of interest.

\subsection{High level overview of the algorithm}
\label{sec:penalty_mixing}

Before describing how we compute $\Psi^{(t)}$ and $\Omega^{(t)},$ we introduce two important functions:
\begin{align*}
p^{\star}(x, \theta)&=\frac{\theta\lambda_1e^{-\lambda_1 \lvert x \rvert}}{\theta\lambda_1e^{-\lambda_1 \lvert x \rvert}+(1-\theta)\lambda_0e^{-\lambda_0 \lvert x \rvert}} &
q^\star(x,\eta)&=\frac{\eta\xi_1e^{-\xi_1 \lvert x \rvert}}{\eta\xi_1e^{-\xi_1 \lvert x \rvert}+(1-\eta)\xi_0e^{-\xi_0 \lvert x \rvert}}.
\end{align*}
Note that $p^{\star}(\psi_{j,k}, \theta)$ is the conditional posterior probability that $\psi_{j,k}$ was drawn from the $\text{Laplace}(\lambda_{1})$ slab distribution.
Similarly, $q^{\star}(\omega_{k,k'}, \eta)$ is just the conditional posterior probability that $\omega_{k,k'}$ was drawn from the $\text{Laplace}(\xi_{1})$ slab: $q^{\star}(\omega_{k,k'}, \eta) = \E[\delta_{k,k'} \vert \bY, \Psi, \Omega, \theta, \eta].$

\textbf{Updating $\Psi$.}  Fixing the value $\Omega = \Omega^{(t-1)},$ computing $\Psi^{(t)}$ is equivalent to solving the following penalized optimization problem
\begin{equation}
\label{eq:psi_update}
\Psi^{(t)} = \argmax_{\Psi}  \left\{-\frac{1}{2}\tr\left((Y\Omega-X\Psi)\Omega^{-1}(Y\Omega-X\Psi)^{\top}\right)+\sum_{jk}\text{pen}(\psi_{j,k} ; \theta)\right\}
\end{equation}
where 
\begin{align*}
\text{pen}(\psi_{j,k}; \theta)&=\log\left(\frac{\pi(\psi_{j,k} \vert \theta)}{\pi(0 \vert \theta)}\right)=-\lambda_1 \lvert \psi_{j,k} \rvert +\log\left(\frac{p^{\star}(\psi_{j,k},\theta)}{p^{\star}(0,\theta)}\right).
\end{align*}
Note that the first term in~\eqref{eq:psi_update} can obtained by distributing a factor of $\Omega$ through the quadratic form that appears in the log-likelihood (see Equations~\eqref{eq:updating_B_eqval} and~\eqref{eq:Psi_penalty}). 

The Karush-Kuhn-Tucker (KKT) conditions for~\eqref{eq:psi_update} tells us that
\begin{equation}
	\label{eq:KKT_Psi}
	\psi^{(t)}_{j,k}=n^{-1}\left[ \lvert z_{j,k} \rvert-\lambda^\star(\psi^{(t)}_{j,k},\theta)\right]_{+}\sign(z_{j,k}),
\end{equation}
where 
\begin{align*}
	z_{jk}&= n\psi^{(t)}_{j,k}+ X_{j}^{\top}\mathbf{r}_{k} + \sum_{k' \neq k}\frac{(\Omega^{-1})_{k,k'}}{(\Omega^{-1})_{k,k}}X_j^{\top} \mathbf{r}_{k'},
\end{align*}
$\mathbf{r}_{k'}=(Y\Omega-X\Psi^{(t)})_{k'},$ and $\lambda^\star(\psi^{(t)}_{j,k},\theta)=\lambda_1 p^\star (\psi^{(t)}_{j,k},\theta)+\lambda_0(1-p^\star (\psi^{(t)}_{j,k},\theta)).$

The KKT conditions suggest a natural coordinate-ascent strategy for computing $\Psi^{(t)}$: starting from some initial guess, we cyclically update the entries $\psi_{j,k}$ by soft-thresholding $\psi_{j,k}$ at $\lambda_{j,k}^{\star}.$
When the current value of $\psi_{j,k}$ is very large, the corresponding value of $p^{\star}(\psi_{j,k},\theta)$ will be close to one, and the threshold $\lambda_{j,k}^{\star}$ will be close to the slab penalty $\lambda_{1}.$
On the other hand, when $\psi_{j,k}$ is very small, the corresponding $p^{\star}$ value will be close to zero and the threshold $\lambda^{\star}_{j,k}$ will be close to the spike penalty $\lambda_{0}.$
Since $\lambda_{1} \ll \lambda_{0},$ we are therefore able to apply a stronger penalty to the smaller entries of $\Psi$ and a weaker penalty to the larger entries.
Essentially, our cyclical coordinate ascent adaptively shrinks the estimates of $\psi_{j,k}$ by iteratively refining the thresholds $\lambda^{\star}_{j,k}.$ 
Before proceeding, observe that the quantity $z_{j,k}$ depends not only on the inner product between the $X_{j},$ the $j^{\text{th}}$ column of the design matrix, and the partial residual $\mathbf{r}_{k}$ but also on the inner product between $X_{j}$ and all other partial residuals $\mathbf{r}_{k'}$ for $k' \neq k.$
Practically this means that in our cyclical coordinate ascent algorithm, our estimate of the direct effect of predictor $X_{j}$ on outcome $Y_{k}$ can depend on how well we have fit all other outcomes $Y_{k'}.$
Moreover, the entries of $\Omega^{-1}$ determine the degree to which $\psi_{j,k}$ depends on the outcomes $Y_{k'}$ for $k' \neq k.$
Specifically, if $(\Omega^{-1})_{k,k'} = 0,$ then we are unable to leverage information contained in $Y_{k'}$ to inform our estimate of $\psi_{j,k}.$

\textbf{Updating $\Omega.$} 
Fixing $\Psi = \Psi^{(t)},$ $S=n^{-1}Y^{\top}Y,$ and $M=n^{-1}(X\Psi)^\top X\Psi,$ we compute
\begin{equation}
\label{eq:cglasso}
\Omega^{(t)} = \argmax_{\Omega\succ 0}\left\{\frac{n}{2}\left(\log \lvert \Omega \rvert - \tr(S\Omega)-\tr(M\Omega^{-1})\right) - 
\sum_{k = 1}^{q}{\left[\xi_{1} \omega_{k,k} + \sum_{k' > k}{\xi^{\star}_{k,k'}\lvert \omega_{k,k'}\rvert}\right]}\right\}
\end{equation}
where $\xi^{\star}_{k,k'} = \xi_{1}q^{\star}(\omega^{(t-1)}_{k,k'}, \eta^{(t-1)}) + \xi_{0}(1 - q^{\star}(\omega^{(t-1)}_{k,k'}, \eta^{(t-1)})).$ 

The objective in Equation~\eqref{eq:cglasso} is extremely similar to the conventional graphical LASSO \citep[GLASSO;][]{Friedman2008} objective.
However, there are two crucial differences.
First, because the conditional mean of $Y$ depends on $\Omega$ in Equation (1), we have an additional term $\tr(M\Omega^{-1})$ that is absent in the GLASSO objective.
The objective in Equation~\eqref{eq:cglasso} also contains \textit{individualized} penalties $\xi^{\star}_{k,k'}$ on the off-diagonal entries of $\Omega.$
Here, the penalty $\xi^{\star}_{k,k'}$ is large (resp. small) whenever the previous estimate of $\omega_{k,k'}^{(t-1)}$ is small (resp. large).
In other words, our ECM algorithm iteratively refines the amount of penalization applied to each off-diagonal entry in $\Omega.$

Although the objective in Equation~\eqref{eq:cglasso} is somewhat different than the GLASSO objective, we can solve it by suitably modifying an existing GLASSO algorithm.
Specifically, we solve the optimization problem in Equation~\eqref{eq:cglasso} with a modified version of \citet{Hsieh2011}'s QUIC algorithm that iteratively (i) forms a quadratic approximation of the objective and (ii) follows a suitable Newton direction for a step size chosen with an Armijo rule.
In Section~\ref{sec:cgquic_proof} of the Supplementary Materials, we show that the optimization problem in Equation~\eqref{eq:cglasso} has a unique solution and that our modification to QUIC converges to the unique solution. 

\subsection{Detailed derivation}
Recall that we wish to maximize the log posterior density, which for $\Omega \succ 0$ is given by,
\begin{align}
\begin{split}
\label{eq:log_density_supp}
\log \pi(\Theta \vert \bY) &= \frac{n}{2}\log\lvert\Omega\rvert -\frac{1}{2}\tr\left((Y-X\Psi\Omega^{-1})\Omega(Y-X\Psi\Omega)^\top\right) \\
&+ \sum_{j = 1}^{p}{\sum_{k = 1}^{q}{\log\left(\theta\lambda_{1}e^{-\lambda_{1}\lvert \psi_{j,k}\rvert} + (1-\theta)\lambda_{0}e^{-\lambda_{0}\lvert \psi_{j,k} \rvert}\right)}} \\
&+ \sum_{k = 1}^{q-1}{\sum_{k' > k}^{q}{\log\left(\eta\xi_{1}e^{-\xi_1\lvert\omega_{k,k'}\rvert}+(1-\eta)\xi_{0}e^{-\xi_0\lvert\omega_{k,k'}\rvert}\right)}} \\
&- \sum_{k = 1}^{q}{\xi_{1} \omega_{k,k}}\\
&+ (a_{\theta}-1)\log(\theta) + (b_{\theta} - 1)\log(1-\theta) \\
&+ (a_{\eta} - 1)\log(\eta) + (b_{\eta} - 1)\log(1-\eta) \\
\end{split}
\end{align}

Instead of optimizing $\log \pi(\Theta \vert \bY)$ directly, we use an ECM algorithm and iteratively update the surrogate objective
$$
F(\Theta) = \E_{\bdelta \vert \cdot}[\log \pi(\Theta, \delta \vert \bY) \vert \Theta],
$$
where $\log \pi(\Theta, \delta \vert \bY)$ is the log-density of the posterior in an \textit{augmented} model involving the spike-and-slab indicators $\bdelta = \{\delta_{k,k'}: 1 \leq k < k' \leq q\}.$
Note that the expectation is taken with respect to the conditional posterior distribution of $\bdelta$ given $\Theta.$
The log density of the augmented posterior is
\begin{align}
\begin{split}
\label{eqn:augmented_log_density}
\log \pi(\Theta, \bdelta \vert \by) &=\frac{n}{2}\log\lvert\Omega\rvert -\frac{1}{2}\tr\left((Y-X\Psi\Omega^{-1})\Omega(Y-X\Psi\Omega)^\top\right)\\
	&+\sum_{k = 1}^{q}{\sum_{j = 1}^{p}{\log\left(\theta\lambda_1e^{-\lambda_1|\psi_{j,k}|}+(1-\theta)\lambda_0e^{-\lambda_0|\psi_{j,k}|}\right)}} \\
	&-\sum_{k = 1}^{q}{\xi_{1}\omega_{k,k}} \\
	&+\sum_{k = 1}^{q}{\sum_{k' = k+1}^{q}{\left(\xi_{1}\delta_{k,k'} + \xi_{0}(1-\delta_{k,k'})\right)\lvert \omega_{k,k'}\rvert}} \\
	&+ \sum_{k = 1}^{q}{\sum_{k' = k+1}^{q}{\delta_{k,k'}\log(\eta) + (1 - \delta_{k,k'})\log(1-\eta)}} \\
	&+(a_\theta-1)\log\theta+(b_\theta-1)\log(1-\theta) \\
	&+(a_\eta-1)\log\eta+(b_\eta-1)\log(1-\eta).
\end{split}
\end{align}

In our augmented model, $\delta_{k,k'}$ indicates whether $\omega_{k,k'}$ was drawn from the spike $(\delta_{k,k'} = 0$) or the slab ($\delta_{k,k'} = 1$).
Given $\Theta$ and the data $\bY,$ these indicators are conditionally independent with
$$
\E[\delta_{k,k'} \vert \bY, \Theta] = \frac{\eta \xi_{1}e^{-\xi_{1}\lvert \omega_{k,k'} \rvert}}{\eta \xi_{1}e^{-\xi_{1}\lvert \omega_{k,k'} \rvert} + (1 - \eta) \xi_{0}e^{-\xi_{0}\lvert \omega_{k,k'} \rvert}}.
$$

The surrogate objective $F(\Theta)$ is given by
\begin{align}
\begin{split}
\label{eqn:M_step}
	F(\Theta) &=\frac{n}{2}\log(\vert\Omega\vert)+\tr(Y^{\top}X\Psi)-\frac{1}{2}\tr(Y^{\top}Y\Omega)-\frac{1}{2}\tr((X\Psi)^{\top}(X\Psi)\Omega^{-1})\\
	&+\sum_{ij} \log\left(\theta\lambda_1e^{-\lambda_1|\psi_{j,k}|}+(1-\theta)\lambda_0e^{-\lambda_0|\psi_{j,k}|}\right) \\
	& -\sum_{k<k'}\xi_{k,k'}^\star|\omega_{k,k'}|-\xi_1\sum_{k=1}^q\omega_{k,k}\\
	&+(a_\theta-1)\log\theta+(b_\theta-1)\log(1-\theta) \\ 
	&+(a_\eta-1)\log\eta+(b_\eta-1)\log(1-\eta)
\end{split}
\end{align}
where $\xi_{k,k'}^\star=\xi_1q^\star_{k,k'}+\xi_{0}(1-q^\star_{k,k'})$ and
$$
q^\star(x,\eta) =\frac{\eta\xi_{1}e^{-\xi_{1}\lvert x \rvert }}{\eta\xi_1e^{-\xi_{1} \lvert x \rvert }+(1-\eta)\xi_{0}e^{-\xi_0 \lvert x \rvert }}.
$$

Our ECM algorithm iteratively computes $F(\Theta)$ based on the current value of $\Theta$ (the E-step) and then updates the value of $\Theta$ by performing two conditional maximizations (the CM-step).
More specifically, for $t \geq 1,$ if $\Theta^{(t-1)}$ is the value of $\Theta$ at the start of the $t^{\text{th}}$ iteration, in the E-step we compute
$$
F^{(t)}(\Theta) = \E_{\bdelta \vert \cdot}[\log \pi(\Theta, \delta \vert \bY) \vert \Theta = \Theta^{(t-1)}].
$$

We then compute $\Theta^{(t)}$ by first optimizing $F^{(t)}(\Theta)$ with respect to $(\Psi,\theta)$ while fixing $(\Omega,\eta) = (\Omega^{(t-1)}, \eta^{(t-1)})$ to obtain $(\Psi^{(t)},\theta^{(t)})$.
Then we optimizing $F^{(t)}(\Theta)$ with respect to $(\Omega, \eta)$ while fixing $(\Psi,\theta) = (\Psi^{(t)},\theta^{(t)})$ to obtain $(\Omega^{(t)}, \eta^{(t)}).$
That is, in the CM-step we solve the following optimization problems
\begin{align}
(\Psi^{(t)}, \theta^{(t)}) &= \argmax_{\Psi,\theta}~ F^{(t)}(\Psi,\theta, \Omega^{(t-1)},  \eta^{(t-1)})  \label{eq:psi_theta_update_supp} \\
(\Omega^{(t)}, \eta^{(t)}) &= \argmax_{\Omega,\eta}~ F^{(t)}(\Psi^{(t)}, \theta^{(t)},\Omega,  \eta). \label{eq:omega_eta_update_supp}
\end{align}
Once we solve the optimization problems in Equations~\eqref{eq:psi_theta_update_supp} and~\eqref{eq:omega_eta_update_supp}, we set 
$$
\Theta^{(t)} = (\Psi^{(t)}, \theta^{(t)}, \Omega^{(t)}, \eta^{(t)}).
$$

Our ECM algorithm iterates between the E-step and CM-step until the percentage change in the estimated entries of $\Psi$ and $\Omega$ or log posterior density is below some user-defined tolerance.
In our implementation, we have found that tolerance of $10^{-3}$ works well.
The following subsections detail how we carry out each conditional maximization step.

\subsection{Updating $\Psi$ and $\theta$}.

Fixing $(\Omega, \eta) = (\Omega^{(t-1)}, \eta^{(t-1)}),$ observe that
\begin{align}
\begin{split}
\label{eq:updating_B_eqval}
F^{(t)}(\Psi, \theta, \Omega^{(t-1)}, \eta^{(t-1)}) &= -\frac{1}{2}\tr\left((Y-X\Psi\Omega^{-1})\Omega(Y-X\Psi\Omega^{-1})^{\top}\right)+\log\pi(\Psi, \theta) \\
&= -\frac{1}{2}\tr\left((Y-X\Psi\Omega^{-1})\Omega\Omega^{-1}\Omega(Y-X\Psi\Omega^{-1})^{\top}\right)+\log\pi(\Psi, \theta) \\
&= -\frac{1}{2}\tr\left((Y\Omega-X\Psi)\Omega^{-1}(Y\Omega-X\Psi)^{\top}\right)+\log\pi(\Psi, \theta)
\end{split}
\end{align}
where
\begin{align}
\begin{split}
\label{eq:log_prior_Psi_theta}
\log \pi(\Psi,\theta) &= \sum_{j = 1}^{p}{\sum_{k = 1}^{q}{\log\left(\theta\lambda_{1}e^{-\lambda_{1}\lvert \psi_{j,k} \rvert} + (1-\theta)\lambda_{0}e^{-\lambda_{0}\lvert \psi_{j,k} \rvert}\right)}} \\
&+ (a_{\theta} - 1)\log(\theta) + (b_{\theta} - 1)\log(1-\theta)
\end{split}
\end{align}
We solve the optimization problem in Equation~\eqref{eq:updating_B_eqval} using a coordinate ascent strategy that iteratively updates $\Psi$ (resp. $\theta$) while holding $\theta$ (resp. $\Psi$) fixed.
We run the coordinate ascent until the estimates of each $\psi_{j,k}$ no longer change by more than 0.1\%.

\textbf{Updating $\theta$ given $\Psi$}.
Notice that the objective in Equation~\eqref{eq:updating_B_eqval} depends on $\theta$ only through the $\log \pi(\Psi,\theta)$ term.
Accordingly, to updating $\theta$ conditionally on $\Psi,$ it is enough to maximize the expression in Equation~\eqref{eq:log_prior_Psi_theta} as a function of $\theta$ while keeping all $\psi_{j,k}$ terms fixed.
We use Newton's method for this optimization and we terminate once the Newton step has a step size less than the user defined tolerance.

\textbf{Updating $\Psi$ given $\theta$}. With $\theta$ fixed, optimizing Equation~\eqref{eq:updating_B_eqval} is equivalent to solving
\begin{align}
\begin{split}
\label{eq:Psi_penalty}
\Psi^{(t)}&=\argmax_{\Psi} \left\{-\frac{1}{2}\tr\left((Y\Omega-X\Psi)\Omega^{-1}(Y\Omega-X\Psi)^\top\right)+\log\pi(\Psi\vert\theta)\right\}\\
&=\argmax_{\Psi} \left\{-\frac{1}{2}\tr\left((Y\Omega-X\Psi)\Omega^{-1}(Y\Omega-X\Psi)^\top\right)+\sum_{j,k}\log\left(\frac{\pi(\psi_{j,k}\vert\theta)}{\pi(0\vert\theta)}\right)\right\}\\
&=\argmax_{\Psi} \left\{-\frac{1}{2}\tr\left((Y\Omega-X\Psi)\Omega^{-1}(Y\Omega-X\Psi)^\top\right)+\sum_{j,k}\pen(\psi_{j,k}\vert\theta)\right\}
\end{split}
\end{align}
where
\begin{align*}
	\pen(\psi_{j,k}\vert\theta)&=\log\left(\frac{\pi(\psi_{j,k}\vert\theta)}{\pi(0\vert\theta)}\right)=-\lambda_1|\psi_{j,k}|+\log\left(\frac{p^\star(\psi_{j,k},\theta)}{p^\star(0,\theta)}\right).
\end{align*}

Following essentially the same arguments as those in \citet{Deshpande2019} and using the fact that the columns of $X$ have norm $\sqrt{n},$ the Karush-Kuhn-Tucker (KKT) condition for optimization problem in the final line of Equation~\eqref{eq:Psi_penalty} tells us that the optimizer $\Psi^{*}$ satisfies
\begin{equation}
	\label{eq:KKT_Psi_supp}
	\psi^{*}_{j,k}=n^{-1}\left[ \lvert z_{j,k} \rvert-\lambda^\star(\psi^{*}_{j,k},\theta)\right]_{+}\sign(z_{j,k}),
\end{equation}
where 
\begin{align*}
	z_{jk}&= n\psi^{*}_{j,k}+ X_{j}^{\top}\mathbf{r}_{k} + \sum_{k' \neq k}\frac{(\Omega^{-1})_{k,k'}}{(\Omega^{-1})_{k,k}}X_j^{\top} \mathbf{r}_{k'}\\
	\mathbf{r}_{k'}&=(Y\Omega-X\Psi^{*})_{k'} \\
	\lambda^\star(\psi^{*}_{j,k},\theta)&=\lambda_1 p^\star (\psi^{*}_{j,k},\theta)+\lambda_0(1-p^\star (\psi^{*}_{j,k},\theta)).
\end{align*}

The KKT condition immediately suggests a cyclical coordinate ascent strategy for solving the problem in Equation~\eqref{eq:Psi_penalty} that involves soft thresholding the running estimates of $\psi_{j,k}.$
Like \citet{RockovaGeorge2018_ssl} and \citet{Deshpande2019}, we can, however, obtain a more refined characterization of the global model $\tilde{\Psi} = (\tilde{\psi}_{j,k})$:
$$
\tilde\psi_{j,k} = n^{-1}\left[\vert z_{j,k}\vert-\lambda^\star(\tilde{\psi}_{j,k},\theta)\right]_{+}\sign(z_{j,k}) \times \mathbbm{1}\left(\lvert z_{j,k} \rvert > \Delta_{j,k}\right),
$$
where
\begin{align*}
	\Delta_{j,k}=\inf_{t>0}\left\{\frac{nt}{2}-\frac{\pen(\tilde{\psi}_{j,k},\theta)}{(\Omega^{-1})_{k,k}t}\right\}.	
\end{align*}
Though the exact thresholds $\Delta_{j,k}$ are difficult to compute, they can be bounded use an analog to Theorem 2.1 of \citet{RockovaGeorge2018_ssl} and Proposition 2 of \citet{Deshpande2019}. 
Specifically, suppose we have $(\lambda_1-\lambda_0)>2\sqrt{n(\Omega^{-1})_{k,k}}$ and $(\lambda^\star(0,\theta)-\lambda_1)^2>-2n(\Omega^{-1})_{k,k}p^\star(0,\theta)$. 
Then we have $\Delta_{j,k}^L\le \Delta_{j,k}\le \Delta^U_{j,k}$, where:

\begin{align*}
	\Delta_{j,k}^L&=\sqrt{-2n((\Omega^{-1})_{k,k})^{-1}\log p^\star(0,\theta)-((\Omega^{-1})_{k,k})^{-2}d}+\lambda_1/(\Omega^{-1})_{k,k}\\
	\Delta_{j,k}^U&=\sqrt{-2n((\Omega^{-1})_{k,k})^{-1}\log p^\star(0,\theta)}+\lambda_1/(\Omega^{-1})_{k,k}\\
\end{align*}
where $d=-(\lambda^\star (\delta_{c+},\theta)-\lambda_1)^2-2n(\Omega^{-1})_{k,k}\log p^\star(\delta_{c+},\theta)$ and $\delta_{c+}$ is the largest root of $\pen''(x|\theta)=(\Omega^{-1})_{k,k}.$

Our refined characterization of $\tilde{\Psi}$ suggests a cyclical coordinate descent strategy that combines hard thresholding at $\Delta_{j,k}$ and soft-thresholding at $\lambda^{\star}_{j,k}.$

\begin{myRemark}
Equation~\eqref{eq:updating_B_eqval} and our approach to solving the optimization problem are extremely similar to Equation 3 and the coordinate ascent strategy used in \citet{Deshpande2019}, who fit sparse marginal multivariate linear models with spike-and-slab LASSO priors.
This is because if $Y \sim \mathcal{N}(X\Psi\Omega^{-1},\Omega^{-1})$ in our chain graph model, then $Y\Omega \sim \mathcal{N}(X\Psi, \Omega).$ 
Thus, if we fix the value of $\Omega,$ we can use any computational strategy for estimating marginal effects in the multivariate linear regression model to estimate $\Psi$ by working with the transformed data $Y\Omega.$
\end{myRemark}

\subsection{Updating $\Omega$ and $\eta$}

Fixing $\Psi = \Psi^{(t)}$ and $\theta = \theta^{(t)}$, we compute $\Omega^{(t)}$ and $\eta^{(t)}$ by optimizing the function

\begin{align}
\begin{split}
\label{eq:Omega_eta_objective}
F^{(t)}(\Psi^{(t)},\theta^{(t)},\Omega,\eta) &= \frac{n}{2}\left[\log\lvert\Omega\rvert - \tr(S\Omega) - \tr(M\Omega^{-1})\right] -\sum_{k < k'}{\xi^{\star}_{k,k'}\lvert \omega_{k,k'}\rvert} - \sum_{k=1}^{q}{\xi_1\omega_{k,k}} \\
~&+\left(a_{\eta} - 1 + \sum_{k < k'}{q_{k,k}^{\star}}\right) \times \log(\eta) \\
~&+ \left(b_{\eta} - 1 + q(q-1)/2 - \sum_{k < k'}{q_{k,k}^{\star}}\right) \times \log(1-\eta)
\end{split}
\end{align}
where $S=\frac{1}{n}Y^{\top}Y$ and $M=\frac{1}{n}(X\Psi)^{\top}X\Psi$. 

We immediately observe that expression in Equation~\eqref{eq:Omega_eta_objective} is separable in $\Omega$ and $\eta,$ meaning that we can compute $\Omega^{(t)}$ and $\eta^{(t)}$ separately.
Specifically, we have
\begin{equation}
\label{eq:eta_update}
\eta^{(t)}=\frac{a_{\eta}-1+\sum_{k<k'}q_{k,k'}^{\star}}{a_{\eta}+b_{\eta}-2+q(q-1)/2}
\end{equation}
and
\begin{align}
\label{eq:Omega_update}
\Omega^{(t)} &= \argmax_{\Omega>0}\left\{\frac{n}{2}\left[\log(\lvert\Omega\rvert)-\tr(S\Omega)-\tr(M\Omega^{-1})\right] -\sum_{k<k'}\xi^\star_{k,k'}|\omega_{k,k'}|-\xi_1\sum_{k=1}^q \omega_{k,k} \right\}.
\end{align}

The objective function in Equation~\eqref{eq:Omega_update} is similar to a graphical LASSO \citep[GLASSO;][]{Friedman2008} problem insofar as both problems involve a term like $\log\lvert\Omega\rvert +\tr(S\Omega)$ and separable $L1$ penalties on the off-diagonal elements of $\Omega.$
However, Equation~\eqref{eq:Omega_update} includes an additional term $\tr(M\Omega^{-1})$, which does not appear in the GLASSO.
This term arises from through the entanglement of $\Psi$ and $\Omega$ in the Gaussian chain graph model and we accordingly call the problem in Equation~\eqref{eq:Omega_update} the CGLASSO problem.
We solve this problem by (i) forming a quadratic approximation of the objective, (ii) computing a suitable Newton direction, and (iii) following that Newton direction for a suitable step size.
We detail this solution strategy in Section~\ref{appendix:cgQUIC}.

\subsection{Effect of prior truncation for $\Omega$}

Recall that we constructed the $\Omega$ prior by first placing independent priors on its elements and then truncating that prior to the positive definite cone.
By truncating the prior, we introduced non-trivial dependence between the off-diagonal elements $\omega_{k,k'}.$
While the untruncated prior helps shrink individual $\omega_{k,k'}$'s toward zero, whether the truncated prior induces similar marginal shrinkage is unclear.
We would be especially concerned if, after truncation, the marginal prior on $\omega_{k,k'}$ was concentrated away from zero, thereby inducing systematic shrinkage towards some non-zero target.

To investigate, we simulated draws from the truncated prior in the $q = 5$ setting where $\eta = 0.5, \xi_{1} = 1$ and $\xi_{0} = 10$ with rejection sampling.
For these choices of spike and slab parameters, the acceptance rate was 3.5\%; see Lemma~\ref{lemma:prior_bounds_omega} for a bound on the probability placed on the positive definite cone by the un-truncated prior.
Figure~\ref{fig:marginal_omega} shows the marginal density of each off-diagonal entry $\omega_{k,k'}.$
The estimated marginal priors are reassuringly symmetric around zero.
More interestingly, they appear to be more concentrated around zero after truncation (black lines) than they were before truncation (red lines).
It would appear, then, that truncation encourages even stronger shrinkage towards zero than the un-truncated prior. 
\begin{figure}[ht]
    \centering
    \includegraphics[width = \textwidth]{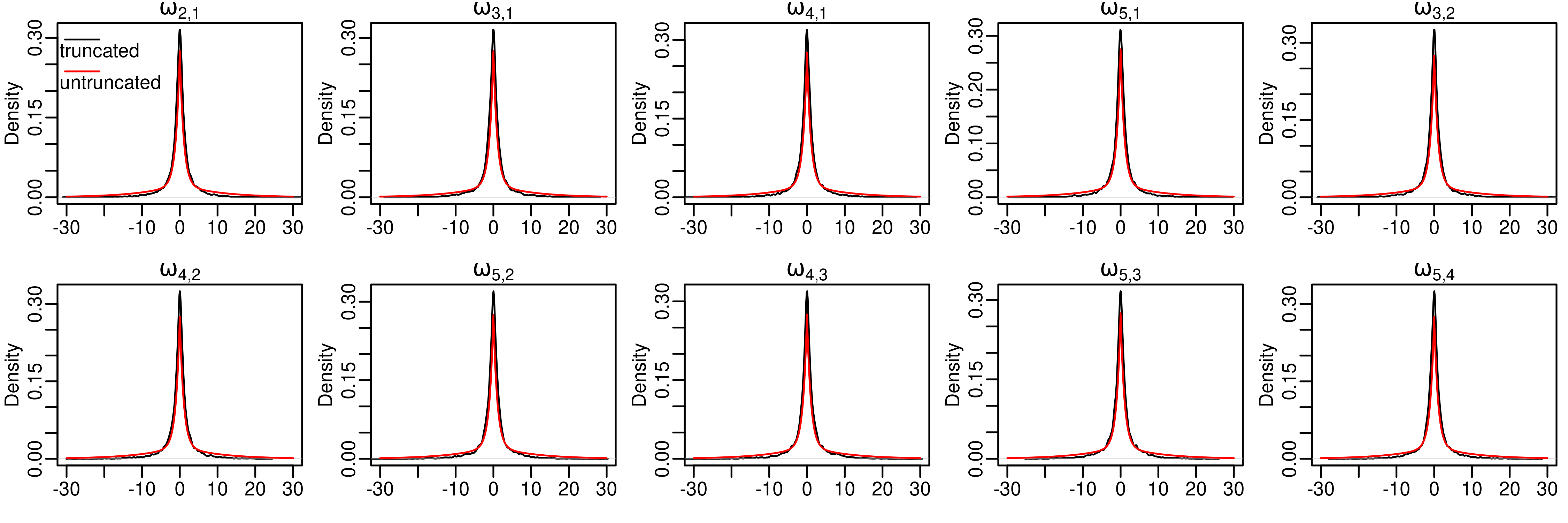}
    \caption{Marginal prior for off diagonal entries with $q=5$, $\eta=\theta=0.5$, $\lambda_1=\xi_1=1$, $\lambda_0=\xi_0=10$ using rejection sampling. The truncated prior is close to the untruncated one with little more concentration around the origin.}
    \label{fig:marginal_omega}
\end{figure}

\subsection{Stabilization of dynamic posterior exploration}

When run with the default settings recommended in Section 3.2 of the main text, our dynamic posterior exploration (DPE) returned posterior modes that appeared to stabilize.
That is, for large penalty values $\lambda_{0}^{(s)}, \lambda_{0}^{(s')}, \xi_{0}^{(t)},$ and $\xi_{0}^{(t')},$ we have $(\Psi^{(s,t)}, \theta^{(s,t)}, \Omega^{(s,t)}, \eta^{(s,t)}) \approx (\Psi^{(s',t')}, \theta^{(s',t')}, \Omega^{(s',t')}, \eta^{(s',t')}).$
Figure~\ref{fig:stabilize} illustrates this phenomenon from a single replication from our simulation study with $p = 10, q = 10,$ and $n = 100.$
Each line represents how the modal estimate of different parameters ($\psi_{j,k}$'s in the top pane and $\omega_{k,k'}$'s in the bottom pane) changes as the spike penalty increases.
\begin{figure}[H]
    \centering
    \includegraphics[width = \textwidth]{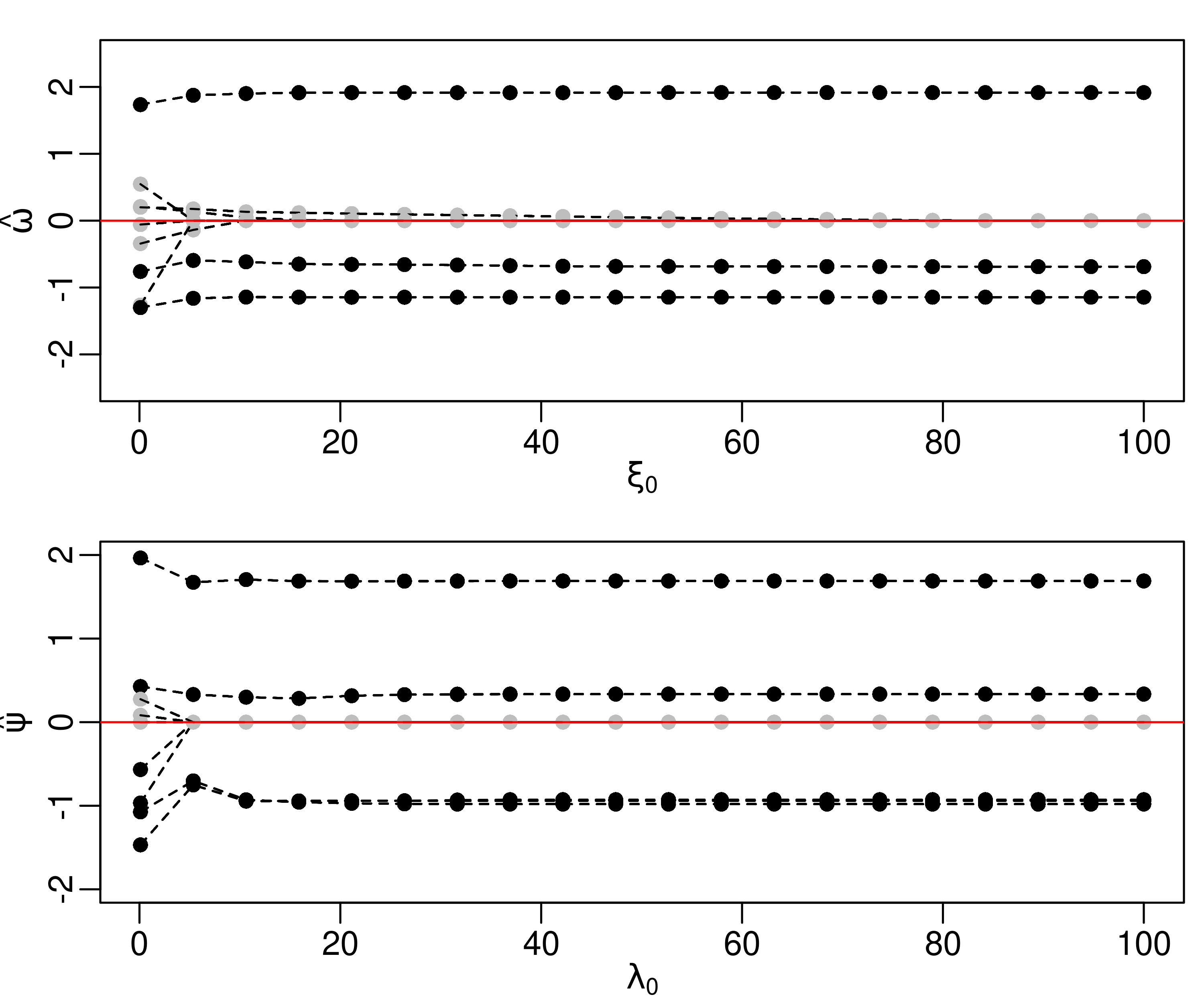}
    \caption{Path of modal estimates returned by DPE as functions of the spike penalties increase.}
    \label{fig:stabilize}
\end{figure}
Such stability has been observed in many other studies of the spike-and-slab LASSO.
It is important to stress that there is no general guarantee that DPE will stabilize for every choice of slab penalties and grids of spike penalties $\mathcal{I}_{\lambda}$ and $\mathcal{I}_{\xi}.$
That said, we have observed stabilization in every experiment we have tried with our recommended default settings.

We conjecture that stabilization will occur so long as the grids contain several large values 
To develop some intuition about why this might happen, consider the adaptive threshold from the $\psi_{j,k}$ update
$$
\lambda^{\star}_{j,k} = \lambda_{1}p^{\star}(\psi_{j,k},\theta) + \lambda_{0}(1 - p^{\star}(\psi_{j,k},\theta))
$$
where
$$
p^{\star}(x, \theta) = \left[1 + \frac{1-\theta}{\theta} \times \frac{\lambda_{0}}{\lambda_{1}} \times \exp\left\{-\left(\lambda_{0} - \lambda_{1}\right)\lvert x \rvert\right\}\right]^{-1}
$$
No suppose that $\psi_{j,k}$ is very small in absolute value so that $p^{\star}(\psi_{j,k},\theta)$ is closer to zero.
When we increase $\lambda_{0}$ in the next step of DPE, the initial $p^{\star}(\psi_{j,k},\theta)$ will become even smaller, making the ensuing penalty $\lambda^{\star}_{j,k}$ even larger.
Similarly, if $\psi_{j,k}$ is non-zero and $p^{\star}(\psi_{j,k},\theta)$ is close to one, then increasing $\lambda_{0}$ will only drive $p^{\star}(\psi_{j,k},\theta)$ closer to one, thereby driving $\lambda^{\star}_{j,k}$ towards the fixed slab penalty $\lambda_{1}.$
As a result, for large enough spike penalties $\lambda_{0},$ zero-estimates tend to remain zero and non-zero estimates tend to remain non-zero.
For a non-zero $\psi_{j,k},$ as $\lambda^{\star}_{j,k}$ converges to $\lambda_{1},$ we end up applying essentially the same amount of regularization from one spike penalty to the next, which manifests as stabilization. 

\section{Chain graphical LASSO with cgQUIC}
\label{appendix:cgQUIC}
Equation~\eqref{eq:Omega_update} is a specific instantiation of what we term the ``chain graphical LASSO'' (CGLASSO) problem, whose general form is
\begin{equation}
\label{eqn:mod_glasso}
\argmin_{\Omega} \left\{-\log(\lvert\Omega\rvert) + \tr(S\Omega) + \tr(M\Omega^{-1})+\sum_{k,k'}{\xi_{k,k'}\lvert \omega_{k,k'}\rvert}\right\}
\end{equation}
where $S$ and $M$ are symmetric positive semi-definite $q \times q$ matrices; $\Omega$ is a symmetric positive definite $q \times q$ matrix; and the $\xi_{k,k'}$'s are symmetric non-negative penalty weights (i.e. we have $\xi_{k,k'}=\xi_{k',k}$).

Notice that when $M$ is the 0 matrix, the CGLASSO problem reduces to a general GLASSO problem, which admits several computational solutions. 
One well known solution is to solve the dual problem, which involves minimization of a log determinant under $L_\infty$ constraint \citep{Banerjee2008}. 

Unfortunately, the dual form of the CGLASSO problem does not have such a simple form.
To wit, the dual of the CGLASSO problem with uniform penalty $\xi$ is given by: 
\begin{align*}
\min_{||U||_\infty <\xi} \max_{\Omega\succ 0} \log\lvert \Omega\rvert - \tr[(S+U)\Omega]-\tr[MX^{-1}]
\end{align*}
The inner optimization about $\Omega$ can be solved by setting the derivative to 0; the optimal value of $\Omega$ solves a special case of continuous time algebraic Riccati equation (CARE) \citep{boyd1991linear}:
\begin{equation*}
\Omega-\Omega(S+U)\Omega+M=0
\end{equation*}
Unfortunately, this problem does not have a closed form solution and solving it numerically in every step of the cgSSL is computationally prohibitive.

We instead solve the CGLASSO problem using a suitably modified version of \citet{Hsieh2011}'s QUIC algorithm for solving the GLASSO problem.
At a high level, instead of using the first order gradient or solving dual problem, the algorithm is based on Newton’s method and uses a quadratic approximation. 
Basically, we sequentially cycle over the parameters $\omega_{k,k'}$ and update each parameter by following a Newton direction for suitable step-size.
The step-size is chosen to ensure that our running estimate of $\Omega$ remains positive definite while also ensuring sufficient decrease in the overall objective.
We call our solution CGQUIC, which we summarize in Algorithm~\ref{algo:cgQUIC}.

\begin{algorithm}[htp]
	\caption{The CGQUIC algorithm for CGLASSO problem}
	\label{algo:cgQUIC}
	\KwData{$S=Y^{\top}Y/n$, $M=(X\Psi)^{\top}(X\Psi)/n$, regularization parameter matrix $\Xi$, initial $\Omega_0$, inner stopping tolerance $\epsilon$, parameters $0 < \sigma < 0.5$, $0 < \beta < 1$}
	\KwResult{path of positive definite $\Omega_t$ that converge to $\argmin_\Omega f$ with $f(\Omega)=g(\Omega) +\sum_{k,k'}{\xi_{k,k'}\lvert \omega_{k,k'}\rvert}$, where $g(\Omega)=-\log(\lvert\Omega\rvert) + \tr(S\Omega) + \tr(M\Omega^{-1})$}
	
	Initialize $W_0=\Omega_0^{-1}$;
	
	\For{$t=1,2,\dots$}{
		$D = 0, U = 0$\;
		$Q=MW_{t-1}$\;
		\While{not converged}{
			Partition the variables into fixed and free sets based on gradient\footnote{For compactness, we omit the subscript $t-1$ for $\Omega$ and $W$} \\
			$S_{fixed}:= \{(k,k'):|\nabla_{k,k'} g(\Omega)|<\xi_{k,k'}\text{ and } \omega_{k,k'}=0\}$\;
			$S_{free}:= \{(k,k'):|\nabla_{k,k'} g(\Omega)|\ge\xi_{k,k'}\text{ or } \omega_{k,k'}\ne0\}$\;
			\For{$(k,k')\in S_{free}$}{
				Calculate Newton direction:\\
				$b=S_{k,k'}-W_{k,k'}+w_k^{\top}Dw_{k'}-w_k^{\top}Mw_{k'}+w_{k'}^{\top} DWM w_k+ w_k^{\top} DWM w_{k'}$\;
				$c=\Omega_{k,k'}+D_{k,k'}$\;
				\If{$i\ne j$}{
					$a=W_{k,k'}^2+W_{k,k}W_{k',k'}+W_{k,k}w_{k'}^{\top}Mw_{k'}+W_{k',k'}w_k^{\top}Mw_k+2W_{k,k'}w_k^{\top}Mw_{k'}$\;
				}
				\Else{
					$a=W_{k,k}^2+2W_{k,k}w_k^{\top}Mw_k$\;
				}
				$\mu=-c+[\lvert c-b/a\rvert-\xi_{k,k'}/a]_+ \sign(c-b/a)$ \;
				$D_{k,k'}+=\mu$\;
				$\boldsymbol{u}_k+=\mu \boldsymbol{w}_{k'}$\;
				$\boldsymbol{u}_{k'}+=\mu \boldsymbol{w}_{k}$\;
				
			}
		}
		\For{$\alpha=1,\beta,\beta^2,\dots$}{
			Compute the Cholesky decomposition of $\Omega_{t-1}+\alpha D$\;
			\If{$\Omega_t$ is not positive definite}{
				continue\;
			}
			Compute $f(\Omega_{t-1}+\alpha D)$\;
			\If{$f(\Omega_{t-1}+\alpha D^*)\le f(\Omega_{t-1})+\alpha \sigma\delta, \delta = \tr[\nabla g(\Omega_{t-1})^{\top} D^*] + ||\Omega_{t-1}+D^*||_{1,\Xi}-||\Omega_{t-1}||_{1,\Xi}$}{
				break\;
			}
		}
	$\Omega_t = \Omega_{t-1}+\alpha D$\;
	$W_{t}=\Omega_t^{-1}$ using Cholesky decomposition result\;
	}
return $\{\Omega_t\}$\;
\end{algorithm}

To describe CGQUIC, we first define the ``smooth'' part of the CGLASSO objective as $g(\Omega)$ and the objective function as $f(\Omega)$:

\begin{align}
\begin{split}
\label{eq:cglasso_smooth_part}	
g(\Omega)&=-\log(\lvert\Omega\rvert)+\tr(S\Omega)+\tr(M\Omega^{-1}),\\
f(\Omega)&=g(\Omega)+\sum_{k,k'} \lambda_{k,k'}\omega_{k,k'}.
\end{split}
\end{align}

The function $g(\Omega)$ is twice differentiable and strictly convex. 
To see this, observe that $g(\Omega)$ is just the log-likelihood of a Gaussian chain graph model with known $\Psi.$
Its Hessian is just the negative Fisher information of $\Omega$ and is positive definite. 
The second-order Taylor expansion of the smooth part $g(\Omega)$ based on $\Omega$ and evaluated at $\Omega+\Delta$ for a symmetric $\Delta$ is

\begin{align}
\begin{split}
	\label{eq:cglasso_smooth_part_quad_approx}
\bar{g}_{\Omega}(\Delta)=&-\log(\lvert\Omega\rvert)+\tr(S\Omega)+\tr(MW)\\
&+ \tr(S\Delta)-\tr(W\Delta)-\tr(WMW\Delta )\\
&+\frac{1}{2}\tr(W\Delta W\Delta)+\tr(WMW\Delta W\Delta)
\end{split}
\end{align}


\subsection{Newton Direction}
We now consider the coordinate descent update for the variable $\Omega_{k,k'}$ for $k \leq k'.$ 
Let $D$ denote the current approximation of the Newton direction and let $D'$ be the updated direction.
To preserve symmetry, we set $D'= D+\mu(e_{k}e_{k'}^{\top}+e_{k'}e_{k}^{\top}).$ 
Our goal, then, is to find the optimal $\mu$:
\begin{equation}
\argmin_\mu \{\bar{g}(D+\mu(e_{k}e_{k'}^{\top}
+e_{k'}e_{k}^{\top} ))+2\xi_{k,k'}|\Omega_{k,k'}+D_{k,k'}+\mu|\}
\end{equation}

We begin by substituting $\Delta = D'$ into $\bar{g}(\Delta).$
Note that terms not depending on $\mu$ do not affect the line search.
Compared to QUIC, we have two additional terms, $\tr(WMW\Delta)$ and $\tr(WMW\Delta W \Delta).$
The first term turns out to be linear $\mu$ and the second is quadratic in $\mu.$
To see this, first observe
\begin{align}
\begin{split}
-\tr(WMW\Delta )&=-\tr(WMW(D+\mu(e_{k}e_{k'}^{\top}
+e_{k'}e_{k}^{\top} ))\\
&=C-\mu\tr(WMWe_{k}e_{k'}^{\top}+WMWe_{k'}e_{k}^{\top})\\
&=C-\mu\tr(e_{k'}^{\top}WMWe_{k}+e_{k}^{\top}WMWe_{k'})\\
&=C-2\mu e_{k}^{\top}WMWe_{k'}\\
&=C-2\mu w_k^{\top}Mw_{k'}
\end{split}
\end{align}
where $w_k$ is the $k^{\text{th}}$ column of $\Omega$. 

Furthermore, we have
\begin{align}
\begin{split}
\tr(WMW\Delta W\Delta)&=\tr[WMW(D+\mu(e_{k}e_{k'}^{\top}
+e_{k'}e_{k}^{\top} )) W(D+\mu(e_{k}e_{k'}^{\top}
+e_{k'}e_{k}^{\top} ))]\\
&=\tr[DWMW+2\mu DWMW(e_{k}e_{k'}^{\top}
+e_{k'}e_{k}^{\top})W] \\
&~~~+ \tr[\mu^2(e_{k}e_{k'}^{\top} +e_{k'}e_{k}^{\top})WMW(e_{k}e_{k'}^{\top}
+e_{k'}e_{k}^{\top})W]\\
~ & ~ \\
=&C+\tr[2\mu DWMW(e_{k}e_{k'}^{\top}
+e_{k'}e_{k}^{\top})W] \\
&~~~+\tr[\mu^2(e_{k}e_{k'}^{\top}
+e_{k'}e_{k}^{\top})WMW(e_{k}e_{k'}^{\top}
+e_{k'}e_{k}^{\top})W]\\
~ & ~ \\
=&C+2\mu w_{k'}^{\top} DWM w_k+2\mu w_k^{\top} DWM w_{k'}^{\top} \\
&~~~+\mu^2\tr[(e_{k}e_{k'}^{\top}
+e_{k'}e_{k}^{\top})WMW(e_{k}e_{k'}^{\top}
+e_{k'}e_{k}^{\top})W]\\
~ & ~ \\
=&C+2\mu (w_{k'}^{\top} DWM w_k+ w_k^{\top} DWM w_{k'}^{\top}) \\
&~~~+\mu^2(W_{k,k}w_{k'}^{\top}Mw_{k'}+W_{k',k'}w_k^{\top}Mw_k+2W_{k,k'}w_k^{\top}Mw_{k'})
\end{split}
\end{align}


By combining the above simplifications, we can minimize the objective with coordinate descent. The update for $\omega_{k,k'}$ is given by:
\begin{align}
\begin{split}
&\frac{1}{2}[W_{k,k'}^2+W_{k,k}W_{k',k'}+W_{k,k}w_{k'}^{\top}Mw_{k'}+W_{k',k'}w_k^{\top}Mw_k+2W_{k,k'}w_k^{\top}Mw_{k'}]\mu^2\\
+&[S_{k,k'}-W_{k,k'}+w_kDw_{k'}-w_kMw_{k'}+w_{k'}^{\top} DWM w_k+ w_k^{\top} DWM w_{k'}]\mu\\
+&\xi_{k,k'}|\Omega_{k,k'}+D_{k,k'}+\mu|
\end{split}
\end{align}

The optimal solution (for off-diagonal $\omega_{k,k'}$) is given by
\begin{equation}
\mu=-c+[\lvert c-b/a \lvert-\xi_{k,k'}/a]_+\sign(c-b/a)
\label{eqn:newton_mu}
\end{equation}
where

\begin{align*}
a&=W_{k,k'}^2+W_{k,k}W_{k',k'}+W_{k,k}w_{k'}^{\top}Mw_{k'}+W_{k',k'}w_k^{\top}Mw_k+2W_{k,k'}w_k^{\top}Mw_{k'}\\
b&=S_{k,k'}-W_{k,k'}+w_k^{\top}Dw_{k'}-w_k^{\top}Mw_{k'}+w_{k'}^{\top} DWM w_k+ w_k^{\top} DWM w_{k'}\\
c&=\omega_{k,k'}+D_{k,k'}
\end{align*}

For diagonal entries, we take $D'=D+\mu e_{k}e_{k}^{\top}$, the two terms involving $D$ are then:
\begin{equation}
\begin{aligned}
-\tr(WMW\Delta )&= C-\mu w_k^{\top}Mw_k\\
\tr(WMW\Delta W\Delta)&=C+2\mu w_k^{\top}DWMw_k+\mu^2 W_{k,k}w_k^{\top}Mw_k
\end{aligned}
\end{equation}

Then we can take 
\begin{align*}
a&=W_{k,k}^2+2W_{k,k}w_k^{\top}Mw_k\\
b&=S_{k,k}-W_{k,k}+w_k^{\top}Dw_k-w_k^{\top}Mw_k+2 w_k^{\top}DWMw_k\\
c&=\omega_{k,k}+D_{k,k}
\end{align*}
and use Equation~\eqref{eqn:newton_mu} to obtain the optimal $\mu$ and thus the updated Newton direction $D'$.  

Note that computing the optimal $\mu$ requires repeated calculation of quantities like $w_k^{\top}Mw_{k'}$ and $w_k^{\top}UMw_{k'}.$
To enable rapid computation, we track and update the values of $U = DW$ and $Q = MW$ during our optimization.

\subsection{Step Size}
Like \citet{Hsieh2011}, we use Armijo's rule to set a step size $\alpha$ that simultaneously ensures our estimate of $\Omega$ remains positive definite and sufficient decrease of our overall objective function. We denote the Newton direction after a complete update over all active coordinates as $D^*$ (see Appendix \ref{supp:active_set_cgquic} for active sets). We require our step size to satisfy the line search condition~\eqref{eqn:decrease_cond}. 
\begin{equation}
f(\Omega+\alpha D^*)\le f(\Omega)+\alpha \sigma\delta, \delta = \tr[\nabla g(\Omega)^{\top} D^*] + ||\Omega+D^*||_{1,\Xi}-||\Omega||_{1,\Xi}
\label{eqn:decrease_cond}
\end{equation}
Three important properties can be established following \cite{Hsieh2011}:
\begin{itemize}
	\item[P1.] The condition was satisfied for small enough $\alpha$. This property is satisfied exactly following proposition 1 of \cite{Hsieh2011}.
	\item[P2.] We have $\delta<0$ for all $\Omega \succ 0$, which ensures that the objective function decreases. This property generally follow Lemma 2 and Proposition 2 of \cite{Hsieh2011}, which requires the Hessian of the smooth part $g(\Omega)$ to be positive definite. In our case the Hessian of $g(\Omega)$ is the Fisher information of the chain graph model, ensuring its positive definiteness. 
	\item[P3.] When $\Omega$ is close to the global optimum, the step size $\alpha=1$ will satisfy the line search condition. To establish this, we follow the proof of Proposition 3 in \cite{Hsieh2011}. 
\end{itemize}

\subsection{Thresholding to Decide the Active Sets}
\label{supp:active_set_cgquic}

Similar to the QUIC procedure, our algorithm does not need to update every $\omega_{k,k'}$ in each iteration. 
We instead follow \citet{Hsieh2011} and only update those parameters exceeding a certain threshold.
More specifically, we can partition the parameters $\omega_{k,k'}$ into a fixed set $S_{\text{fixed}},$ containing those parameters falling below the threshold, and a free set $S_{\text{free}},$ containing those parameters exceeding the threshold.
That is
\begin{equation}
\begin{aligned}
\omega_{k,k'}&\in S_{fixed} \text{ if } |\nabla_{k,k'}g(\Omega)|\le \xi_{k,k'}\text{ and } \omega_{k,k'}=0\\
\omega_{k,k'}&\in S_{free} \text{ otherwise}
\end{aligned}
\end{equation}

We can determine the free set $S_{\text{free}}$ using the minimum-norm sub-gradient $\text{grad}_{k,k'}^Sf(\Omega),$ which is defined in Definition 2 of \citet{Hsieh2011}.
In our case $\nabla g(\Omega)=S-\Omega^{-1}-\Omega^{-1}M\Omega^{-1}$, so the minimum-norm sub-gradient is


\begin{equation}
\text{grad}_{k,k'}^Sf(\Omega)=
\begin{cases}
(\nabla g(\Omega))_{k,k'}+\xi_{k,k'} & \text{if } \omega_{k,k'}>0\\
(\nabla g(\Omega))_{k,k'}-\xi_{k,k'} & \text{if } \omega_{k,k'}<0\\
\sign((\nabla g(\Omega))_{k,k'})[\lvert (\nabla g(\Omega))_{k,k'}\rvert-\xi_{k,k'}]_+ & \text{if } \omega_{k,k'}=0\\
\end{cases}
\end{equation}
Note that the subgradient evaluated on the fixed set is always equal to 0.
Thus, following Lemma 4 in \citet{Hsieh2011}, the elements of the fixed set do not change during our coordinate descent procedure.
It suffices, then, to only compute the Newton direction on the free set and update those parameters.



\subsection{Unique minimizer}
\label{sec:cgquic_proof}

In this subsection, we show that the CGLASSO problem admits a unique minimizer. 
Our proof largely follows the proofs of Lemma 3 and Theorem 1 of \cite{Hsieh2011} but makes suitable modifications to account for the extra $\tr(M\Omega^{-1})$ term in the CGLASSO objective.

\begin{myTheorem}[Unique minimizer]
	There is a unique global minimum for the CGLASSO problem ~\eqref{eqn:mod_glasso}.
	\label{thm:cgquic_unique}
\end{myTheorem}

We first show the entire sequence of iterates $\{\Omega_{t}\}$ lies in a particular, compact level set.
To this end, let 
\begin{equation}
	\label{eqn:level_set}
	U=\{\Omega|f(\Omega)\le f(\Omega_0),\Omega\in S_{++}^p\}.
\end{equation}

To see that all iterations lies in $U$, we check need to check the line search condition Equation~\eqref{eqn:decrease_cond} has a $\delta<0$. By directly applying \citet{Hsieh2011}'s Lemma 2 to $g(\Omega)$, we have that

\begin{align*}
	\delta\le -\vect(D^*)^\top \nabla^2g(\Omega)\vect(D^*)
\end{align*}

where $D^*$ is the Newton direction. 
Since $g(\Omega)$ (Equation~\eqref{eq:cglasso_smooth_part}) is convex, $\nabla^2g(\Omega)$ is positive definite, so the function value $f(\Omega_t)$ is always decreasing. 

Now we need to check that the level set is actually contained in a compact set, by suitably adapt Lemma 3 of \citet{Hsieh2011}. 

\begin{myLemma}
\label{lem:level_set}
The level set $U$ defined in ~\eqref{eqn:level_set} is contained in the set $\{mI\le \Omega \le NI\}$ for some constants $m,N > 0$, if we assume that the off-diagonal elements of $\Xi$ and the diagonal elements of $S$ are positive.
\end{myLemma}

\begin{proof}
We begin by showing that the largest eigenvalue of $\Omega$ is bounded by some constant that does not depend on $\Omega.$
Recall that $S$ and $M$ are positive semi-definite.
Since $\Omega$ is positive definite, we have $\tr(S\Omega)+\tr(M\Omega^{-1})>0$ and $||\Omega||_{1,\Xi}+\tr(M\Omega^{-1})>0$. 

Therefore we have

	\begin{equation}
		\label{eqn:level_inequal}
		\begin{aligned}
			f(\Omega_0)&>f(\Omega)\ge -\log(\lvert\Omega\rvert) +||\Omega||_{1,\Xi}\\
			f(\Omega_0)&>f(\Omega)\ge -\log(\lvert\Omega\rvert) +\tr(S\Omega)
		\end{aligned}
	\end{equation}

	Since $||\Omega||_2$ is the largest eigenvalue of $\Omega$, we have $\log(\lvert\Omega\rvert)\le q\log(||\Omega||_2)$. 

	Using the assumption that off-diagonal entries of $\Xi$ is larger than some positive number $\xi,$ we know that

	\begin{equation}
		\label{eqn:inequal_Lambda}
		\xi\sum_{i\ne j}|\Omega_{k,k'}|\le ||\Omega||_{1,\Xi}\le f(\Omega_0)+q\log(||\Omega||_2)
	\end{equation}

	Similarly, we have

	\begin{equation}
		\label{eqn:inequal_S}
		\tr(S\Omega)\le f(\Omega_0)+q\log(||\Omega||_2)
	\end{equation}

	Let $\alpha=\min_k(S_{k,k})$ and $\beta = \max_{k\ne k'} S_{k,k'}$. 
	We can split the $\tr(S\Omega)$ into two parts, which can be further lower bounded:

	\begin{equation}
		\label{eqn:trSX}
		\tr(S\Omega)=\sum_k S_{k,k}\Omega_{k,k}+\sum_{k\ne k'}S_{k,k'}\Omega_{k,k'}\ge \alpha \tr(\Omega)-\beta\sum_{k\ne k'} |\Omega_{k,k'}|
	\end{equation}

	Since $||\Omega||_2\le \tr(\Omega)$, by using Equation~\eqref{eqn:trSX}, we have,

	\begin{equation}
		\label{eq:Omega_2_bound}
		\alpha ||\Omega||_2\le \alpha \tr(\Omega) \le \tr(\Omega S)+\beta\sum_{k\ne k'} |\Omega_{k,k'}|
	\end{equation}

	By combining Equations~\eqref{eqn:inequal_Lambda},~\eqref{eqn:inequal_S}, and\eqref{eq:Omega_2_bound}, we conclude that
	\begin{equation}
		\alpha||\Omega||_2\le (1+\beta/\xi)(f(\Omega_0)+q\log(||\Omega||_2))
	\end{equation}
	The left hand side as a function of $||\Omega||_2$ grows much faster than then right hand side.
	 Thus $||\Omega||_2$ can be bounded by a quantity depending only on the value of $f(\Omega_0)$, $\alpha$, $\beta$ and $\xi$. 
	
	We now consider the smallest eigenvalue denoted by $a$.  We use the upper bound of other eigenvalues to bound the determinant. By using the fact that $f(\Omega)$ always decreasing during iterations, we have

	\begin{equation}
		f(\Omega_0)>f(\Omega)>-\log(\lvert\Omega\rvert)\ge -\log(a)-(q-1)\log(N)
	\end{equation}
	Thus we have $m=e^{-f(\Omega_0)M^{-q+1}}$ is a lower bound for smallest eigenvalue $a$. 
\end{proof}

We are now ready to prove Theorem \ref{thm:cgquic_unique}, by showing the objective function is strongly convex on a compact set. 

\begin{proof}
	Because of Lemma~\ref{lem:level_set}, the level set $U$ contains all iterates produced by cgQUIC. The set $U$ is further contained in the compact set $\{mI\le \Omega\le N I\}$. By the Weierstrass extreme value theorem the continuous function $f(\Omega)~\eqref{eq:cglasso_smooth_part}$ attains its minimum on this set. 

	Further, the modified objective function is also strongly convex in its smooth part.
	This is because $\tr(M\Omega^{-1})$ and $\tr(S\Omega)$ are convex and $-\log(\lvert\Omega\rvert)$ is strongly convex. 
	Since $\tr(M\Omega^{-1})$ is convex, the Hessian of the smooth part has the same lower bound as in Theorem 1 of \cite{Hsieh2011}. 
	By following the argument of in the proof of Theorem 1 of \citet{Hsieh2011}, we can show the objective function $f(\Omega)$ is strongly convex on the compact set $\{mI\le \Omega\le N I\}$, and thus has a unique minimizer.

\end{proof}

We can further show that the cgQUIC procedure converges to the unique minimizer, using the general results on quadratic approximation methods studied in \citet{Hsieh2011}.

\begin{myTheorem}[Convergence]
\label{thm:cgquic_convergence}
	The cgQUIC converge to global optimum.
\end{myTheorem}

	\begin{proof}
		cgQUIC is an example of quadratic approximation method investigated in Section 4.1 of \cite{Hsieh2011} with a strongly convex smooth part $g(\Omega)$ in~\eqref{eq:cglasso_smooth_part}. Convergence to the global optimum follows from their Theorem 2. 
	\end{proof}

\section{Synthetic experiment results}
\label{appendix:synthetic_experiment}

\subsection{Hyperparameter sensitivity analysis}

Our proposed computational procedure, cgSSL with Dynamic Posterior Exploration (hereafter, cgSSL-DPE), involves setting several hyperparameters.
Namely, one must specify (i) two slab penalties $\lambda_{1}$ and $\xi_{1}$ and two grids of $L$ increasing spike penalties $\mathcal{I}_{\lambda} = \{\lambda_{0}^{(s)}\}_{s = 1}^{L}$ and $\mathcal{I}_{\xi} = \{\xi_{0}^{(t)}\}_{t = 0}^{L};$ and (ii) Beta priors for the $\theta$ and $\eta$ parameters that help determine the overall sparsity of the estimate of the final estimates.
cgSSL-DPE returns $L^{2}$ posterior modes, one for every combination of spike penalties.
In the main text, we recommended reporting the final mode (corresponding to the largest spike penalties) once the estimates appeared to stabilize in parameter space.
We recommended to run cgSSL-DPE with the following default settings: (i) $\lambda_{1} = 1$ and $\xi_{1} = 0.01n$; (ii) setting$\mathcal{I}_{\lambda}$ to contain ten evenly spaced values ranging from $10$ to $n$ and $\mathcal{I}_{\xi}$ to contain ten evenly spaced values from $0.1n$ to $n;$ and (iii) $\theta \sim \textrm{Beta}(1,pq)$ and $\eta \sim \textrm{Beta}(1,q).$
In this subsection, we investigate how sensitive the final posterior mode estimates are to these choices.

For our sensitivity analysis, we generated 100 datasets with $(n,p,q) = (400, 100, 30)$ a sparse Gaussian chain graph model where $\Omega$'s sparsity structure corresponds to a star graph.
Like in our main simulation study, 20\% of the entries in the true $\Psi$ were non-zero and generated uniformly from the interval $[-2,2].$
We compared DPE run with our default settings (\texttt{default}) to the following competitors: (i) shorter (\texttt{grid\_5}) and longer (\texttt{grid\_20}) grids of spike penalties spanning the default ranges; (ii) much more diffuse slab distributions with $\lambda_{1} = 0.1$ and $\xi_{1} = n/100 (\texttt{small\_slab});$ (iii) grids of spike penalties that were ten times smaller or larger than the recommended defaults (\texttt{small\_spike} and \texttt{large\_spike}); and (iv) uniform priors on $\theta$ and $\eta$ (\texttt{unif}). 
Note that we did not run DPE with combinations of alternate settings.
That is, each competitor differs from \texttt{default} in exactly one respect
For instance, \texttt{small\_slab} and \texttt{default} differ only in the ranges spanned by $\mathcal{I}_{\lambda}$ and $\mathcal{I}_{\xi}:$ in \texttt{small\_slab}, these grids span $1$ to $n/10$ and $n/1000$ and $n/10$ while in \texttt{default}, these grids span $10$ to $n$ and $0.1n$ to $n.$

Figure~\ref{fig:hyperparam} shows the distribution of each settings' F1 score for recovering $\Psi$'s and $\Omega$'s support across the 100 simulated datasets.

\begin{figure}[H]
    \centering
    \begin{subfigure}{0.45\textwidth}
    \centering
    \includegraphics[width = \textwidth]{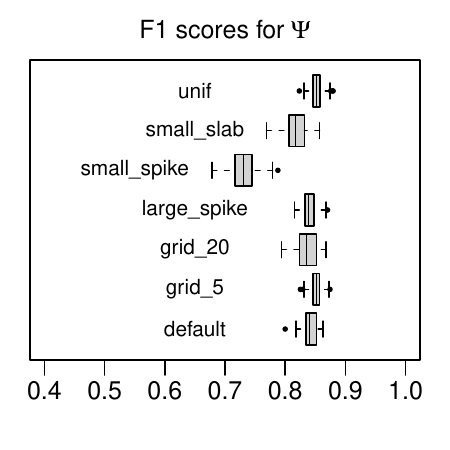}
    \caption{}
    \label{fig:hyperparameter_psi}
    \end{subfigure}
        \begin{subfigure}{0.45\textwidth}
    \centering
    \includegraphics[width = \textwidth]{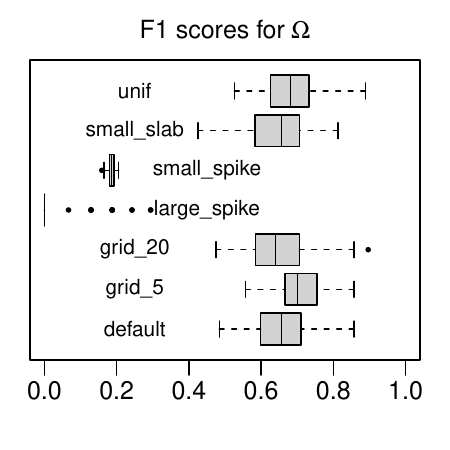}
    \caption{}
    \label{fig:hyperparameter_omega}
    \end{subfigure}
    \caption{F1 scores for recovering the supports of $\Psi$ (a) and $\Omega$ (b) for different hyperparameter settings.}
    \label{fig:hyperparam}
\end{figure}

Running DPE with along a finer (i.e. longer) grid of spike penalties resulted in F1 scores very similar to running with the default settings.
Interestingly, for both $\Psi$ and $\Omega,$ we observed a very slight improvement in support recovery when we ran DPE with shorter grids of spike penalties (\texttt{grid\_5}).
In terms of recovering $\Psi$'s support, increasing the spike penalties (\texttt{large\_spike}) yielded similar performance to \texttt{default} but reducing the range of spike penalties (\texttt{small\_spike}) and reducing the slab penalty (\texttt{small\_slab}) tended to produce worse results.
Support recovery of $\Omega$ was much more sensitive to changes in the spike penalties. 
We additionally observed very little sensitivity to the choice of prior for $\theta$ and $\eta.$

\subsection{Additional simulation results}
Given grids $\mathcal{I}_{\lambda} = \{\lambda_{0}^{(s)}\}_{s = 1}^{L}$ and $\mathcal{I}_{\xi} = \{\xi_{0}^{(t)}\}_{t = 0}^{L}$ spike penalties, the dynamic posterior exploration (DPE) implementation of cgSSL (hereafter \texttt{cgSSL-DPE}) computes $L^{2}$ posterior modes, one for each combination of spike penalties.
This implementation can be computationally intensive, as it requires running the ECM algorithm to convergence for every pair of penalties.
Following \citet{Deshpande2019}, we also implemented a faster variant, called dynamic \textit{conditional} posterior exploration (DCPE).
In DCPE, we first run our ECM algorithm with warm-starts over the grid $\mathcal{I}_{\lambda},$ while keeping $\Omega = I$ fixed.
Then, fixing $\Psi$ and $\Theta$ at their final values from the first step, we run our ECM algorithm with warm-starts over the ladder $\mathcal{I}_{\Omega}.$
Finally, we run our ECM algorithm one more time starting from the final estimates of the parameters obtained in the first two steps with $(\lambda_{0}, \xi_{0}) = (\lambda_{0}^{(L)}, \xi_{0}^{(L)}).$
In this way, DCPE maximizes $2L$ \textit{conditional} posteriors and one joint posterior.
Generally speaking, DPE and DCPE trace different paths through the parameter space and return different final estimates.

We now present the remaining results from our simulation experiments.
These results are qualitatively similar to those from the $(n,p,q) = (100,10,10)$ setting presented in the main text.
Generally speaking, in terms of support recovery, the methods that deployed a single fixed penalty (\texttt{cgLASSO} and \texttt{CAR}) displayed higher sensitivity but lower precision than \texttt{cgSSL-DPE} and \texttt{cgSSL-DCPE}.
The only exception was when $\Omega$ was dense.
Furthermore, methods with adaptive penalties (both cgSSL procedures and \texttt{CAR-A}) tended to return a fewer number of non-zero estimates than the fixed penalty.
Most of these non-zero estimates were in fact true positives.
Across all settings of $(n,p,q)$, \texttt{cgSSL-DPE} makes virtually no false positive identifications in the support of $\Psi.$
In terms of parameter estimation, the fixed penalty methods tended to have larger Frobenius error in estimating both $\Psi$ and $\Omega$ than the cgSSL.
Note that \texttt{cgLASSO} uses ten-fold cross-validation to set the two penalty levels.

\newpage
\begin{table}[H]
    \centering
    \caption{Average sensitivity, precision, and Frobenius error for $\Psi$ and $\Omega$ when $(n,p,q) = (100, 10, 10)$ across 100 simulated datasets. Best performance is bold-faced.}
    \label{tab:psi_omega_relationship_3}
    \scriptsize
    \begin{tabular}{lccccccc}
    \hline
    ~ & \multicolumn{3}{c}{$\Psi$ recovery} & \multicolumn{3}{c}{$\Omega$ recovery} & Runtime \\ 
    Method & SEN & PREC & FROB & SEN & PREC & FROB & Time (min) \\ \hline
    \multicolumn{8}{c}{$AR(1)$ model} \\ \hline
    \texttt{cgLASSO} & \textbf{0.88} & 0.44 & 0.1 & 0.78 & 0.55 & 31.9 &  \\
    \texttt{CAR} & 0.86 & 0.31& 0.04 & \textbf{1} & 0.3 & 4.2 & 0.4 \\
    \texttt{CAR-A} & 0.87 & 0.59 & \textbf{0.02} & \textbf{1} & 0.83 & 2.8 & 0.4 \\
    \texttt{OBFB} & -- & -- & -- & 0.73 &0.32 & NA & 16.5 \\
    \texttt{cgSSL-dcpe} & 0.64 & 0.8 & 0.08 & 0.94 & 0.96 & 6.3 & 0.02 \\ 
    \texttt{cgSSL-dpe} & 0.65  & \textbf{0.99} & 0.04 & \textbf{1} & 0.97 & \textbf{2.5} & 0.2 \\ 
    \texttt{cgSSL-dpe+BB} & 0.64  & \textbf{0.99} & -- & \textbf{1} & \textbf{0.99} & -- & 0.6 \\ 
    \texttt{mSSL} & 0.89 & 0.29  & 0.1 & 1 & 0.99 & 3.2 & 0.01 \\ \hline
     \multicolumn{8}{c}{$AR(2)$ model} \\ \hline
     \texttt{cgLASSO} & \textbf{1} & 0.22 & 0.2 & 0.84  & 0.55 & 2.7 & \\
     \texttt{CAR }& 0.9 & 0.34& 0.03 & 0.98 & 0.57 & 0.6 & 0.4 \\
     \texttt{CAR-A} & 0.89 & 0.67& \textbf{0.02} & \textbf{1} & \textbf{0.91} & 0.5 & 0.4 \\
     \texttt{OBFB} & -- & -- & -- & 0.82&0.35 & -- & 15.4 \\
     \texttt{cgSSL-dcpe} & 0.96 & 0.43 & 0.5 & 0.24 & 0.63 & 5.0 & 0.4 \\
     \texttt{cgSSL-dpe} & 0.73  & \textbf{1} & \textbf{0.02} & \textbf{1} & 0.86 & \textbf{0.4} & 0.6 \\ 
     \texttt{cgSSL-dpe+BB} & 0.74 & \textbf{1} & -- & 0.99 & 0.89 & -- & 1.1 \\ 
     \texttt{mSSL} & 0.98  & 0.29  & 0.2 & 0.61& 0.95  & 2.5 & 0.004 \\ \hline
    \multicolumn{8}{c}{Block model} \\ \hline
    \texttt{cgLASSO} & \textbf{0.95} & 0.39  & 0.1 & 0.73 & 0.78& 5.2 & \\
    \texttt{CAR} & 0.89  & 0.31  & \textbf{0.03} & \textbf{0.95} & 0.61 & \textbf{1.9} & 0.4 \\
    \texttt{CAR-A} & 0.87 & 0.57 & \textbf{0.03 } & 0.86  & 0.93 & 3.0 & 0.4 \\
    \texttt{OBFB} & -- & -- & -- & 0.78&0.52 & -- & 14.5 \\
    \texttt{cgSSL-dcpe} & 0.76 & 0.29 & 0.3 & 0.0075 & 0.71 & 8.8 & 0.02 \\ 
    \texttt{cgSSL-dpe} & 0.69  & \textbf{0.99} & \textbf{0.03} & 0.71  & \textbf{0.95 } & 3.3 & 0.1 \\ 
     \texttt{cgSSL-dpe+BB} & 0.68  & \textbf{0.99} & -- & 0.72  & \textbf{0.95} & -- & 0.3 \\
     \texttt{mSSL} & 0.86 & 0.29 & 0.1 & 0.57 & 0.99 & 6.5 & 0.002 \\ \hline
     \multicolumn{8}{c}{Star model} \\ \hline
     \texttt{cgLASSO} & \textbf{0.96 } & 0.48  & 0.04 & 0.36 & 0.20  & 0.9 &  \\
     \texttt{CAR} & 0.91 & 0.34 & 0.02 & \textbf{0.55} & 0.25 & 0.6 & 0.4 \\
     \texttt{CAR-A} & 0.91 & 0.6 & 0.02  & 0.2 & 0.46 & 0.6 & 0.4 \\
     \texttt{OBFB} & -- & -- & -- & 0.15& \textbf{1} & -- & 16.5 \\
     \texttt{cgSSL-dcpe} & 0.83 & 0.96 & 0.01 & 0.049 & 0.90 & 0.2 & 0.01 \\ 
     \texttt{cgSSL-dpe} & 0.79 & \textbf{0.99} & \textbf{0.01} & 0.09 & 0.71 & \textbf{0.3} & 0.3 \\ 
     \texttt{cgSSL-dpe+BB} & 0.79  & \textbf{0.99} & -- & 0.09  & 0.38  & -- & 0.7 \\ 
     \texttt{mSSL} & 0.83 & 0.96 & 0.015 & 0 & -- & 0.5 & 0.0005 \\ \hline
    \multicolumn{8}{c}{Small world model} \\ \hline
    \texttt{cgLASSO} & 0.54 & 0.44 & 0.2 & 0.66 & 0.45 & 44.7 & \\
    \texttt{CAR} & \textbf{0.86} & 0.28 & 0.05 & \textbf{1} & 0.31& 9.7 & 0.4 \\
    \texttt{CAR-A} & 0.81& 0.55 & \textbf{0.04} & 0.94& 0.75 & 13.3 & 0.4 \\
    \texttt{OBFB} & -- & -- & -- & 0.52&0.46 & -- & 17.6 \\
    \texttt{cgSSL-dcpe} & 0.46 & 0.85 & 0.1 & 0.83& \textbf{0.93} & 19.6 & 0.009 \\
    \texttt{cgSSL-dpe} & 0.60 & \textbf{0.99} & 0.05 & 0.91& \textbf{0.93} & \textbf{8.0} & 0.06 \\ 
    \texttt{cgSSL-dpe+BB}  &0.60 & 1 & -- &0.9 & \textbf{0.94} & -- & 0.2 \\
    \texttt{mSSL} & 0.81& 0.41 & 0.2 & 0.6  & 0.95 & 17.6 & 0.005\\\hline
     \multicolumn{8}{c}{Tree model} \\ \hline
    \texttt{cgLASSO} & 0.83 & 0.47 & 0.1 & 0.83 & 0.37 & 41.6 & \\
    \texttt{CAR} & \textbf{0.87} & 0.27 & 0.06 & \textbf{0.87} & 0.24 & 12 & 0.4 \\
    \texttt{CAR-A} & 0.84 & 0.55 & \textbf{0.04} & 0.63 & 0.58 & 13 & 0.4 \\
    \texttt{OBFB} & -- & -- & -- & 0.35&0.53 & -- & 18.43 \\
    \texttt{cgSSL-dcpe} & 0.33 & \textbf{1} & 0.2 & 0.59 & 0.78 & 20 & 0.001 \\
    \texttt{cgSSL-dpe} & 0.54& \textbf{1} & 0.05& 0.57& 0.87 & \textbf{8.8} & 0.003 \\ 
    \texttt{cgSSL-dpe+BB}  &0.55 & 1 & -- &0.56 & \textbf{0.88} & -- & 0.05 \\
    \texttt{mSSL} & 0.59& 0.63 & 0.09 & 0.46 & 0.99& 14 & 0.001\\ \hline
    \multicolumn{8}{c}{Dense model} \\ \hline
     \texttt{cgLASSO} & \textbf{0.92} & 0.57 & \textbf{0.03} & \textbf{0.88} & \textbf{1} & \textbf{16.9} & \\
     \texttt{CAR} & 0.85 & 0.28 & 0.04 & 0.03 & \textbf{1} & 92.5 & 0.4 \\
     \texttt{CAR-A} & 0.84 & 0.4 & 0.04 & 0 & \textbf{1} & 96.0 & 0.4 \\
     \texttt{OBFB} & -- & -- & -- & 0.76&1 & -- & 16.04 \\
     \texttt{cgSSL-dcpe} & 0.82 & 0.84 & 0.02 & 0.011 & 1 & 99.9 & 0.01 \\ 
     \texttt{cgSSL-dpe} & 0.72 & \textbf{0.93} & \textbf{0.03} & 0.05 & \textbf{1} & 100.0 & 0.02 \\ 
     \texttt{cgSSL-dpe+BB} & 0.72  & \textbf{0.93} & -- & 0.05 & \textbf{1} & --& 0.3 \\ 
     \texttt{mSSL} & 0.82 & 0.81& 0.03& 0& -- & 99.84 & 0.3 \\ \hline
    \end{tabular}
\end{table}
    
\begin{table}[H]
    \centering
    \caption{Sensitivity, precision, and Frobenius error for $\Psi$ and $\Omega$ when $(n,p,q) = (100, 20, 30)$ for each specification of $\Omega.$ For each choice of $\Omega,$ the best performance is bold-faced.}
    \label{tab:psi_omega_relationship_2}
    \scriptsize
    \begin{tabular}{lccccccc}
    \hline
    ~ & \multicolumn{3}{c}{$\Psi$ recovery} & \multicolumn{3}{c}{$\Omega$ recovery} & Runtime \\ 
    Method & SEN & PREC & FROB & SEN & PREC & FROB & Time (min) \\ \hline
    \multicolumn{8}{c}{$AR(1)$ model} \\ \hline
    \texttt{cgLASSO} & \textbf{0.94} & 0.30 & 0.2 & 0.74& 0.48& 111.7 & \\
    \texttt{CAR} & 0.54 & 0.39& 0.1& 1 & 0.21 & 11.2 & 8.9 \\
    \texttt{CAR-A} & 0.69& 0.69& 0.04& \textbf{1} & 0.74 & 15.1 & 9.8 \\
    \texttt{OBFB} & -- & --& -- & 0.06 &0.76 & -- & 56.2 \\
    \texttt{cgSSL-dcpe} & 0.66 & 0.82& 0.08& 0.94 & 0.68 & 34.1 & 3.0 \\
    \texttt{cgSSL-dpe} & 0.69& \textbf{1} & \textbf{0.02} & \textbf{1} & 0.82 & \textbf{4.9} & 30.2 \\ 
    \texttt{cgSSL-dpe+BB} & 0.70 & \textbf{1} & --& \textbf{1} & \textbf{0.88} & -- & 34.3 \\
    \texttt{mSSL} & 0.83 & 0.32 & 0.1& 0.97& 0.89 & 18.0 & 3.1 \\ \hline
     \multicolumn{8}{c}{$AR(2)$ model} \\ \hline
     \texttt{cgLASSO} & \textbf{0.94} & 0.3& 0.2 & 0.98  & 0.18  & 14.4 & \\
    \texttt{CAR} & 0.42& 0.37 & 0.2 & 0.38 & 0.25 & 15.8 & 9.9 \\
    \texttt{CAR-A} & 0.64& 0.76 & 0.07 & 0.91 & \textbf{0.87} & 3.2 & 10.7 \\
    \texttt{OBFB} & -- & -- & -- & 0.06&0.15 & -- & 53.0 \\
    \texttt{cgSSL-dcpe} & 0.73 & 0.7 & 0.05 & 0.81& 0.32& 14.5 & 0.6 \\
    \texttt{cgSSL-dpe} & 0.72& \textbf{0.99} & \textbf{0.01} & \textbf{1} & 0.48 & \textbf{1.0} & 26.3 \\ 
    \texttt{cgSSL-dpe+BB} & 0.72& \textbf{0.99} & -- & 0.99& 0.71& -- & 29.5 \\ 
    \texttt{mSSL} & 0.88& 0.32 & 0.1 & 0.33 & 0.61 & 21.0 & 1.3 \\ \hline
    \multicolumn{8}{c}{Block model} \\ \hline
    \texttt{cgLASSO} & \textbf{0.92} & 0.4 & 0.5 & 0.62 & 0.93 & 27.5 & \\
    \texttt{CAR} & 0.51 & 0.4& 0.1 & 0.53 & 0.68 & \textbf{13.0} & 10.8 \\
    \texttt{CAR-A} & 0.66 & 0.64 & \textbf{0.06 } & 0.47 & \textbf{0.96} & 23.1 & 9.3 \\
    \texttt{OBFB} & -- & -- & -- & 0.06&0.50 & -- & 57.1 \\
    \texttt{cgSSL-dcpe} & 0.82 & 0.29 & 0.7 & 0.1& 0.88& 30.2 & 1.6 \\
    \texttt{cgSSL-dpe} & 0.61& \textbf{0.99} & 0.07& \textbf{0.66 } & 0.9  & 30.4 & 53.4 \\ 
    \texttt{cgSSL-dpe+BB} & 0.61& \textbf{0.99} & -- & 0.38 & 0.93& -- & 57.3 \\
    \texttt{mSSL} & 0.89& 0.25& 0.6 & 0.16& 0.99  & 70.0 & 0.5 \\ \hline
     \multicolumn{8}{c}{Star model} \\ 
      \hline
     \texttt{cgLASSO} & \textbf{0.91} & 0.45& 0.08& 0.7 & 0.31 & 6.5 & \\
    \texttt{CAR} & 0.45& 0.41& 0.1& 0.32& 0.12 & 4.4 & 9.9 \\
    \texttt{CAR-A} & 0.69& 0.68& 0.06& 0.31& 0.35& 2.8 & 9.4 \\
    \texttt{OBFB} & -- & -- & -- & 0 & 0 & -- & 42.6 \\
    \texttt{cgSSL-dcpe} & 0.77& 0.94 & \textbf{0.01} & 0.61& 0.57& \textbf{0.6} & 0.09 \\
    \texttt{cgSSL-dpe} & 0.73& \textbf{1}& \textbf{0.01} & \textbf{0.83} & 0.54& 0.9 & 1.8 \\ 
    \texttt{cgSSL-dpe+BB} & \textbf{0.72} & \textbf{1.0 } & -- & 0.62 & \textbf{0.59} & -- & 3.6  \\ 
    \texttt{mSSL} & 0.76 & 0.91& 0.02& 0.32 & 1 & 1.9 &  0.2\\ \hline
    \multicolumn{8}{c}{Small world model} \\ \hline
    \texttt{cgLASSO} & 0.74 & 0.49 & 0.19 & 0.83 & 0.18 & 307.0 & \\
    \texttt{CAR} & 0.55& 0.35 & 0.12 & \textbf{0.84} & 0.15 & 40.1 & 8.6\\
    \texttt{CAR-A} & \textbf{0.7} & 0.63 & 0.05 & \textbf{0.84} & 0.56& 43.8 & 10.3 \\
    \texttt{cgSSL-dcpe} & 0.52 & 0.9& 0.09 & 0.58 & 0.47& 73.6 & 11.4 \\ 
    \texttt{OBFB} & -- & -- & -- & 0.08&0.10 & -- & 40.8 \\
    \texttt{cgSSL-dpe} & 0.6 & \textbf{0.99} & \textbf{0.04} & 0.73 & 0.58 & \textbf{30.1} & 130.1\\ 
    \texttt{cgSSL-dpe+BB}  &0.6 & 0.99 & --&0.72 & \textbf{0.6} & -- & 135.3 \\
    \texttt{mSSL} & 0.67 & 0.43& 0.16& 0.57& 0.8 & 64.1 & 3.3 \\\hline
     \multicolumn{8}{c}{Tree model} \\ \hline
    \texttt{cgLASSO} & 0.74 & 0.51 & 0.16 & 0.64 & 0.36 & 398 & \\
    \texttt{CAR} & 0.63 & 0.36  & 0.092 & \textbf{0.94 } & 0.15  & 92 & 8.7\\
    \texttt{CAR-A} & \textbf{0.74} & 0.61& \textbf{0.05} & 0.92 & 0.51 & 69 & 9.8 \\
    \texttt{OBFB} & -- & -- & -- & 0.10&0.11 & -- & 48.0 \\
    \texttt{cgSSL-dcpe} & 0.47 & 0.91 & 0.098& 0.71& \textbf{0.88} & 110 & 3.0 \\
    \texttt{cgSSL-dpe} & 0.56 & \textbf{0.99} & 0.06 & 0.86& \textbf{0.88} & \textbf{55} & 18.7 \\ 
    \texttt{cgSSL-dpe+BB}  &0.56 & 0.99 & -- &0.87 & \textbf{0.88} & -- & 21.1\\
    \texttt{mSSL} & 0.62& 0.5& 0.12& 0.65 & 0.89& 120 & 3.6 \\ \hline
    \multicolumn{8}{c}{Dense model} \\ \hline
     \texttt{cgLASSO} & \textbf{0.89} & 0.43& 0.07& \textbf{0.34} & 1& \textbf{712.5} & \\
    \texttt{CAR} & 0.49& 0.39 & 0.1& 0.05 & 1 & 914.5 & 9.4 \\
    \texttt{CAR-A} & 0.7 & 0.64& 0.05& 0.01& 1  & 897.9 & 9.4 \\
    \texttt{OBFB} & -- & -- & -- & 0 & 0 & -- & 56.2 \\
    \texttt{cgSSL-dcpe} & 0.77 & 0.99 & \textbf{0.01} & 0  & 1& 900.0 & 0.09 \\
    \texttt{cgSSL-dpe} & 0.72 & \textbf{1} & \textbf{0.01} & 0.03 & 1  & 901.5 & 1.8 \\ 
    \texttt{cgSSL-dpe+BB} & \textbf{0.72} & \textbf{1} & -- & 0.02& 1 & -- & 15.2 \\ 
    \texttt{mSSL} & 0.76& 0.99 & 0.01 & 0 & -- & 900.0 & 0.005 \\ \hline
    \end{tabular}
\end{table}

\section{Preprocessing for real data experiment}
\label{app:preprocessing}
\citet{claesson2012gut}’s dataset consisted of measurements of the following $q = 14$ most common microbial taxa: \textit{Alistipes, Bacteroides, Barnesiella, Blautia, Butyrivibrio, Caloramator, Clostridium, Eubacterium, Faecalibacterium, Hespellia, Parabacteroides, Ruminococcus, Selenomonas} and \textit{Veillonella}. 
They estimated the marginal effects of each of the following $p = 11$ predictors on each taxa: age, gender, three different stratum type (day hospital, long-term residential care and rehabilitation), five different diet groups including PEG diet, and BMI.
We refer the reader to the Supplementary Notes and Supplementary Table 3 of \citet{claesson2012gut} for specific details about these variables.

To conduct our analysis, we preprocessed the raw 16s-rRNAseq data following the workflow provided on the MG-RAST server \citep{keegan2016mg}. 
We first ``annotated'' the sequences to get genus counts (i.e.~number of segments belongs to one genus). 
The annotation process compares the rRNA segments detected during sequencing to the reference sequence of each genus of microbes, then counts the number of rRNA segments match with each genus. 
We used the MG-RAST server's default tuning parameters during the annotation process.
That is, we set e-value to be 5 and annotated with 60\% identity, alignment length of 15 bp, and set a minimal abundance of 10 reads. 

Following standard practices of analyzing microbiome data, we transformed raw counts into relative abundance.
We selected genera with more than 0.5\% relative abundance in more than 50 samples as the focal genus and all other genera aggregated as the reference group. 
We further took the log-odds (with respect to the reference group described above) to stabilize the variances \citep{aitchison1982statistical} in order to fit our normal model.

\section{Proofs of posterior contraction for cgSSL}
\label{appendix:posterior_contraction}
This section provides detail on the posterior contraction results for the cgSSL. 
Our proof was inspired by \citet{Ning2020} and \citet{Bai2020_groupSSL}. 
We first show the contraction in log-affinity by verifying KL condition and test conditions following \citet{ghosal2017fundamentals}. Then we use the results in log-affinity to show recovery of parameters.

To establish our results, we work with a slightly modified prior on $\Omega$ that has density
\begin{align}
\begin{split}
    f_\Omega(\Omega) & \propto \prod_{k>k'}\left[\frac{(1-\eta) \xi_0}{2}\exp\left(-\xi_0|\omega_{k,k'}|\right)+\frac{\eta\xi_1}{2}\exp\left(\xi_1|\omega_{k,k'}|\right)\right] \\
    ~&~~~\times \prod_{k}\xi_1\exp\left[-\xi_1\omega_{k,k}\right] \times \mathbbm{1}(\Omega\succ \tau I) 
    \end{split}
    \\
    f_\Psi(\Psi) & =\prod_{jk}\left[\frac{(1-\theta) \lambda_0}{2}\exp\left(-\lambda_0|\psi_{j,k}|\right)+\frac{\theta\lambda_1}{2}\exp\left(\lambda_1|\psi_{j,k}|\right)\right]
\end{align}
where $0 < \tau < 1/b_{2}.$
This way $\tau$ is less than the lower bound of the smallest eigenvalue of the true precision matrix $\Omega_{0}.$

We will need to estimate the mass of some events under this truncated prior, however, the truncation of the prior on $\Omega$ makes computing these masses intractable.
To overcome this, we first bound the prior probability of events of the form $A\cap \{\Omega\succ \tau I\}$ by observing the prior on $\Omega$ can be viewed as a particular conditional distribution.

Specifically, let $\tilde{\Pi}$ be the \textit{untruncated} spike-and-slab LASSO prior with density
$$
\tilde{f}(\Omega) = \prod_{k>k'}\left[\frac{(1-\eta) \xi_0}{2}\exp\left(-\xi_0|\omega_{k,k'}|\right)+\frac{\eta\xi_1}{2}\exp\left(\xi_1|\omega_{k,k'}|\right)\right] \times \prod_{k}\xi_1\exp\left(-\xi_1\omega_{k,k}\right).
$$

The following Lemma shows that we can bound $\Pi$ probabilities using $\tilde{\Pi}$ probabilities.
\begin{myLemma}[Bounds of the graphical prior]
    \label{lemma:prior_bounds_omega}
Let $\tilde{\Pi}$ be the untruncated version of the prior on $\Omega.$
Then for all events $A,$ for large enough $n$ there is a number $R$ that does not depend on $n$ such that
    \begin{equation}
        \label{eqn:graphical_prior_bound}
        \tilde{\Pi}(\Omega\succ \tau I|A)\tilde{\Pi}(A)\le \Pi_\Omega(A\cap \{\Omega\succ \tau I\})\le \exp(2 \xi_1 Q-\log(R))\tilde{\Pi}(A)
    \end{equation}
    where $Q=q(q-1)/2$ is the total number of free off-diagonal entries in $\Omega.$
\end{myLemma}

\begin{proof}

    Consider an event of form $A\cap \{\Omega\succ \tau I\} \subset \mathbb{R}^{q\times q}$. The prior mass $\Pi_\Omega(A\cap \{\Omega\succ \tau I\})$ can be viewed as a conditional probability:

    \begin{equation}
        \label{eqn:graphical_prior_cond_exp}
        \Pi_\Omega(A\cap \{\Omega\succ \tau I\})=\tilde{\Pi}(A|\Omega\succ \tau I)=\frac{\tilde{\Pi}(\Omega\succ \tau I|A)\tilde{\Pi}(A)}{\tilde{\Pi}(\Omega\succ \tau I)}
    \end{equation}

    The lower bound follows because the denominator is bounded from above by 1.

    For the upper bound, we first observe that
    \begin{equation}
        \Pi_\Omega(A\cap \{\Omega\succ \tau I\})=\tilde{\Pi}(A|\Omega\succ \tau I)=\frac{\tilde{\Pi}(\Omega\succ \tau I|A)\tilde{\Pi}(A)}{\tilde{\Pi}(\Omega\succ \tau I)}\le (\tilde{\Pi}(\Omega\succ \tau I))^{-1} \tilde{\Pi}(A)
    \end{equation}

To upper bound the probability in Equation~\eqref{eqn:graphical_prior_cond_exp}, we find a lower bound of the denominator $\tilde{\Pi}(\Omega\succ \tau I)$. 
To this end, let
$$
\mathcal{G}=\left\{\Omega: \omega_{k,k}> q-1, |\omega_{k,k'}|\le 1-\frac{\tau}{q-1} \text{ for } k'\ne k\right\}
$$
and consider an $\Omega \in \mathcal{G}.$
Since all of $\Omega$'s eigenvalues are real, they must each be contained in at least one Gershgorin disc.
Consider the $k^{\text{th}}$ Gershgorin disc, whose intersection with the real line is an interval centered at $\omega_{k,k}$ with half-width $\sum_{k' \neq k}{\lvert \omega_{k,k'}\rvert}.$
Any eigenvalue of $\Omega$ that lies in this disc must be greater than
$$
\omega_{k,k} - \sum_{k' \neq k}{\lvert \omega_{k,k'}\rvert} > q-1 - (q - 1 - \tau) = \tau
$$
Thus, we have $\mathcal{G}\subset \left\{\Omega \succ \tau I \right\}.$  
    
Since the entries of $\Omega$ are independent under $\tilde{\Pi}$, we compute
    \begin{equation}
        \label{eqn:graphical_prior_lower_bound_PDcone}
        \begin{aligned}
            \tilde{\Pi}(\mathcal{G})
            &\ge\prod_k \int_{q-1}^\infty \xi_1 \exp(-\xi_1 \omega_{k,k})d\omega_{k,k}(1-\eta)^{Q}\prod_{k>k'}\int_{|\omega_{k,k'}|\le 1-\frac{\tau}{q-1}}\frac{\xi_0}{2}\exp(-\xi_0|\omega_{k,k'}|)d\omega_{k,k'}\\
            &\ge \exp(-2\xi_1 Q)(1-\eta)^Q \left[1-\frac{\E|\omega_{k,k'}|}{1-\frac{\tau}{q-1}}\right]^Q\\
            &=\exp(-2\xi_1 Q)(1-\eta)^Q\left[1-\frac{1}{\xi_0(1-\frac{\tau}{q-1})}\right]^Q\\
            &\ge \exp(-2\xi_1 Q)\left[1-\frac{1}{1+K_1Q^{2+a}}\right]^Q \left[1-\frac{1}{K_3Q^{2+b}(1-\tau)}\right]^Q\\
            &\ge \exp(-2\xi_1 Q+\log(R)),\\
        \end{aligned}
        \end{equation}
    where $R>0$ does not depend on$n$. 
    Note that the first inequality holds by ignoring the contribution to the probability from slab distribution.
    The second inequality is Markov's inequality and the third inequality follows from our assumptions about how $\xi_{0}$ and $\eta$ are tuned.
\end{proof}

\subsection{The Kullback-Leibler condition}

The first technical lemma we established is the so-called Kullback-Leibler condition, on a high level meaning that our prior places enough probability mass in small neighborhoods around each of the possible values of the true parameters. These neighborhoods are defined by KL divergence and KL variance.

The KL divergence between a Gaussian chain graph model with parameters $(\Psi_{0},\Omega_{0})$ and one with parameters $(\Psi,\Omega)$ is
\begin{equation}
    \begin{aligned}
        &\frac{1}{n}K(f_0,f)=\E_0 \left[\log\left(\frac{f_0}{f}\right)\right]\\
       =&\frac{1}{2}\left(\log\left(\frac{|\Omega_0|}{|\Omega|}\right)-q+\tr(\Omega_0^{-1}\Omega)+\frac{1}{n}\sum_{i=1}^n||\Omega^{1/2}(\Psi\Omega-\Psi_0\Omega_0)^{\top}X_i^{\top}||_2^2\right)
    \end{aligned}
\end{equation}

The KL variance is:
\begin{equation}
    \begin{aligned}
        &\frac{1}{n}V(f_0,f)=Var_{0} \left[\log\left(\frac{f_0}{f}\right)\right]\\
       =&\frac{1}{2}\left(\tr((\Omega_0^{-1}\Omega)^2)-2\tr (\Omega_0^{-1}\Omega) + q \right)+\frac{1}{n}\sum_{i=1}^n||\Omega_0^{-1/2}\Omega(\Psi\Omega^{-1}-\Psi_0\Omega^{-1}_0)^{\top}X_i^{\top}||_2^2
    \end{aligned}
\end{equation}

\begin{myLemma}[KL conditions]
Let $	\epsilon_n=\sqrt{\max\{p,q,s_0^\Omega,s_0^\Psi\}\log(\max\{p,q\})/n}.$ Then for all true parameters $(\Psi_{0},\Omega_{0})$ we have
    \label{lemma:KL_cg}
    \begin{align*}
        -\log\Pi\left[(\Psi,\Omega):K(f_0,f)\le n\epsilon_n^2, V(f_0,f)\le n\epsilon_n^2\right]\le C_1n\epsilon_n^2
    \end{align*}
\end{myLemma}

\begin{proof}
    Let $S_{0}^{\Psi}$ and $S_{0}^{\Omega}$ respectively denote the supports of $\Psi$ and $\Omega.$
Similarly, let $s^{\Psi}_{0}$ be the number of true non-zero entries in $\Psi_{0}$ and let $s^{\Omega}_{0}$ be the true number of non-zero off-diagonal entries in $\Omega_{0}$

We need to lower bound the prior probability of the event 
$$
\{(\Psi,\Omega):K(f_0,f)\le n\epsilon_n^2, V(f_0,f)\le n\epsilon_n^2\}
$$ 
for large enough $n$. 

We first obtain an upper bound of the average KL divergence and variance so that the mass of such event can serve as a lower bound. To simplify the notation, we denote $(\Psi-\Psi_0)=\Delta_\Psi$ and $\Omega-\Omega_0=\Delta_\Omega$. We observe that $\Psi\Omega^{-1}-\Psi_0\Omega^{-1}_0=(\Delta_\Psi-\Psi_0\Omega_0^{-1}\Delta_\Omega)\Omega^{-1}$. 

Using the fact that $||A-B||_2^2\le(||A||_2+||B||_2)^2\le 2||A||_2^2+2||B||_2^2$ for any two matrices $A$ and $B,$ we obtain a simple upper bound:
\begin{equation}
    \begin{aligned}
        &\frac{1}{n}K(f_0,f)\\
       =&\frac{1}{2}\left(\log\left(\frac{|\Omega_0|}{|\Omega|}\right)-q+\tr(\Omega_0^{-1}\Omega)+\frac{1}{n}\sum_{i=1}^n||\Omega^{-1/2}\Delta_\Psi^{\top}X_i^{\top}-\Omega^{-1/2}\Delta_ \Omega\Omega_0^{-1} \Psi_0^{\top} X_i^{\top} ||_2^2\right)\\
       \le&\frac{1}{2}\left(\log\left(\frac{|\Omega_0|}{|\Omega|}\right)-q+\tr(\Omega_0^{-1}\Omega)\right)+\frac{1}{n}\sum_{i=1}^n||\Omega^{-1/2}\Delta_ \Omega\Omega_0^{-1} \Psi_0^{\top} X_i^{\top} ||_2^2\\
       &+\frac{1}{n}\sum_{i=1}^n||\Omega^{-1/2}\Delta_\Psi^{\top}X_i^{\top}||_2^2\\
       =&\frac{1}{2}\left(\log\left(\frac{|\Omega_0|}{|\Omega|}\right)-q+\tr(\Omega_0^{-1}\Omega)\right)+\frac{1}{n}||X\Psi_0\Omega_0^{-1}\Delta_\Omega\Omega^{-1/2}||_F^2\\
       &+\frac{1}{n}||X\Delta_\Psi\Omega^{-1/2}||_F^2\\
    \end{aligned}
\end{equation}

The last line holds because $\Omega^{-1/2}\Delta_\Psi^{\top}X_i^{\top}$ is the $i^{\text{th}}$ row of $X\Delta_\Psi\Omega^{-1/2}$. 

Using the same inequality, we derive a similar upper bound for the average KL variance:
\begin{equation}
    \begin{aligned}
        &\frac{1}{n}V(f_0,f)\\
       =&\frac{1}{2}\left(\tr((\Omega_0^{-1}\Omega)^2)-2\tr (\Omega_0^{-1}\Omega) + q \right)+\frac{1}{n}\sum_{i=1}^n ||\Omega_0^{-1/2}\Delta_\Psi^{\top} X_i^{\top}- \Omega_0^{-1/2} \Delta_\Omega\Omega_0^{-1} \Psi_0^{\top} X_i^{\top}||_2^2\\
       \le&\frac{1}{2}\left(\tr((\Omega_0^{-1}\Omega)^2)-2\tr (\Omega_0^{-1}\Omega) + q \right)+\frac{2}{n}\sum_{i=1}^n || \Omega_0^{-1/2}\Delta_\Omega\Omega_0^{-1} \Psi_0^{\top} X_i^{\top}||_2^2\\
       &+\frac{2}{n}\sum_{i=1}^n ||\Omega_0^{-1/2}\Delta_\Psi^{\top} X_i^{\top}||_2^2\\
       =&\frac{1}{2}\left(\tr((\Omega_0^{-1}\Omega)^2)-2\tr (\Omega_0^{-1}\Omega) + q \right)+\frac{2}{n}|| X\Psi_0 \Omega_0^{-1} \Delta_\Omega \Omega_0^{-1/2}||_F^2+\frac{2}{n}||X\Delta_\Psi\Omega_0^{-1/2}||_F^2
    \end{aligned}
\end{equation}

We find event $\mathcal{A}_1$ involving only $\Delta_\Omega$ and event $\mathcal{A}_2$ involving both $\Delta_\Omega$ and $\Delta_\Psi$ such that $(\mathcal{A}_1\cap \{\Omega\succ 0\})\cap\mathcal{A}_2$ is a subset of the event of interest $\{ K/n\le \epsilon_n^2, V/n\le\epsilon_n^2 \}$. 

To this end, define
\begin{equation}
    \label{eqn:A1_due_to_Omega}
    \begin{aligned}
        \mathcal{A}_1=&\left\{ \Omega: \frac{1}{2}\left(\tr((\Omega_0^{-1}\Omega)^2)-2\tr (\Omega_0^{-1}\Omega) + q \right)+\frac{2}{n}|| X\Psi_0 \Omega_0^{-1} \Delta_\Omega \Omega_0^{-1/2}||_F^2\le \epsilon_n^2/2\right\} \\
        &\bigcap\left\{\frac{1}{2}\left(\log\left(\frac{|\Omega_0|}{|\Omega|}\right)-q+\tr(\Omega_0^{-1}\Omega)\right)+\frac{1}{n}||X\Psi_0\Omega_0^{-1}\Delta_\Omega\Omega^{-1/2}||_F^2\le \epsilon_n^2/2 \right\}
    \end{aligned}
\end{equation}
and 
\begin{equation}
    \label{eqn:A2_condition_on_A1}
    \begin{aligned}
        \mathcal{A}_2=\{ (\Omega, \Psi): &\frac{1}{n}||X\Delta_\Psi\Omega_0^{-1/2}||_F^2\le \frac{\epsilon_n^2}{2},\frac{2}{n}||X\Delta_\Psi\Omega^{-1/2}||_F^2  \le \frac{\epsilon_n^2}{2}  \}
    \end{aligned}
\end{equation}

We separately bound the prior probabilities $\Pi(\mathcal{A}_1)$ and $\Pi(\mathcal{A}_1|\mathcal{A}_2)$. 

To lower bound $\Pi(\mathcal{A}_1)$, first consider the event
\begin{align*}
    \mathcal{A}_1^\star=\{2\sum_{k>k'} |\omega_{0,k,k'}-\omega_{k,k'}|+\sum_k |\omega_{0,k,k}-\omega_{k,k}|\le \frac{\epsilon_n}{c_1\sqrt{p}}\}
\end{align*}
where $c_{1} > 0$ is a constant to be specified.
Since the Frobenius norm is bounded by the vectorized L1 norm, we immediately conclude that 
$$
A_{1}^{\star} \subset \left\{ \lVert \Omega_{0} - \Omega \rVert_{F} \leq \frac{\epsilon_{n}}{c_{1}\sqrt{p}}\right\}.
$$
We now show that $\left\{ \lVert \Omega_{0} - \Omega \rVert_{F} \leq \frac{\epsilon_{n}}{c_{1}\sqrt{p}} \right\}\subset \mathcal{A}_1$.

Since the Frobenius norm bounds the L2 operator norm, if $||\Omega_0-\Omega||_F\le \frac{\epsilon_n}{c_1 \sqrt{p}}$ then the absolute value of the eigenvalues of $\Omega - \Omega_{0}$ are bounded by $\frac{\epsilon_n}{c_1 \sqrt{p}}.$
Further, because we have assumed $\Omega_{0}$ has bounded spectrum, the spectrum of $\Omega=\Omega_0+\Omega-\Omega_0$ is bounded by $\lambda_{\min}-\frac{\epsilon_n}{c_1 \sqrt{p}}$ and $\lambda_{\max}+\frac{\epsilon_n}{c_1 \sqrt{p}}.$
When $n$ is large enough, these quantities are further bounded by $\lambda_{\min}/2$ and $2\lambda_{\max}.$
Thus, for $n$ large enough, if $||\Omega_0-\Omega||_F\le \frac{\epsilon_n}{c_1 \sqrt{p}},$ then we know $\Omega$ has bounded spectrum. 

Consequently, $\Omega^{-1/2}$ has bounded L2 operator norm.
Using the fact that $||AB||_F\le \min(|||A|||_2||B||_F,|||B|||_2||A||_F)$, we have for some constant $c_{2}$ not depending on $n,$

\begin{align*}
    \frac{2}{n}|| X\Psi_0 \Omega_0^{-1} \Delta_\Omega \Omega_0^{-1/2}||_F^2&\le \frac{2}{n}|||X\Psi_0|||^2_2||\Omega_0^{-1} \Delta_\Omega \Omega_0^{-1/2}||_F^2\\
    &\le \frac{2}{n}||X||_F^2|||\Psi_0|||_2^2|||\Omega_0^{-1} \Delta_\Omega \Omega_0^{-1/2}||_F^2\\
    &\le pc_2^2||\Delta_\Omega||_F^2,
\end{align*}
where we have used the fact that $||X||_F=\sqrt{np}.$
Thus $||\Delta_\Omega||_F\le \frac{\epsilon_n}{2c_2\sqrt{p}}$ implies 
$$\frac{1}{n}|| X\Psi_0 \Omega_0^{-1} \Delta_\Omega \Omega_0^{-1/2}||_F^2\le \epsilon_n^2/4.$$

Similarly, for some constant $c_3$, we have that 
\begin{align*}
    \frac{1}{n}||X\Psi_0\Omega_0^{-1}\Delta_\Omega\Omega^{-1/2}||_F^2&\le \frac{1}{n}||X||_F^2||\Psi_0||_2^2 ||\Omega_0^{-1}\Delta_\Omega\Omega^{-1/2}||_F^2\\
    &\le pc_3^2||\Delta_\Omega||_F^2
\end{align*}
Thus we have $||\Delta_\Omega||_F\le \frac{\epsilon_n}{2c_3\sqrt{p}}$ implies $\frac{1}{2n}||X\Psi_0\Omega_0^{-1}\Delta_\Omega\Omega^{-1/2}||_F^2\le \epsilon_n^2/4$

We have $||\Delta_\Omega||_F\le \frac{\epsilon_n}{2b_2\sqrt{p}} \le \epsilon_n/2b_2$ implies the following two inequalities
\begin{align*}
\frac{1}{2}(\tr((\Omega_0^{-1}\Omega)^2)-2\tr (\Omega_0^{-1}\Omega) + q ) &\le \epsilon_n^2/4 \\
\frac{1}{2}(\log(\frac{|\Omega_0|}{|\Omega|})-q+\tr(\Omega_0^{-1}\Omega)) &\le \epsilon_n^2/4.
\end{align*}
Thus by taking $c_1=2\max\{c_2,c_3,b_2\}$, we can conclude $\{||\Omega_0-\Omega||_F\le \frac{\epsilon_n}{c_1\sqrt{p}}\}\subset \mathcal{A}_1$. 
Thus $\mathcal{A}_1^\star \subset \mathcal{A}_1$

Since $\mathcal{A}_1^\star\in \{\Omega: ||\Omega_0-\Omega||_F\le \epsilon_n/c_1\sqrt{p}\},$ we know that  
$\tilde{\Pi}(\Omega\succ \tau I|\mathcal{A}_1^\star)=1.$
We can therefore lower bound $\Pi(\mathcal{A}_1)$ by $\Pi(\mathcal{A}_1^*\cap \{\Omega\succ \tau I\})$.
Instead of calculating the latter probability directly, we can lower bound it by observing
\begin{align*}
    &2\sum_{k>k'} |\omega_{0,k,k'}-\omega_{k,k'}|+\sum_k |\omega_{0,k,k}-\omega_{k,k}|\\
    =&2\sum_{(k,k')\in S_0^\Omega} |\omega_{0,k,k'}-\omega_{k,k'}| + 2\sum_{(k,k')\in (S_0^\Omega)^c} |\omega_{k,k'}|+\sum_{k} |\omega_{0,k,k}-\omega_{k,k}|.
\end{align*}

Consider the following events
\begin{align*}
    \mathcal{B}_1&=\left\{ \sum_{(k,k')\in S_0^\Omega} |\omega_{0,k,k'}-\omega_{k,k'}|\le \frac{\epsilon_n}{6c_1\sqrt{p}}  \right\}\\
    \mathcal{B}_2&=\left\{\sum_{(k,k')\in (S_0^\Omega)^c} |\omega_{k,k'}|\le \frac{\epsilon_n}{6c_1\sqrt{p}}\right\}\\
    \mathcal{B}_3&=\left\{ \sum_{k} |\omega_{0,k,k}-\omega_{k,k}| \le  \frac{\epsilon_n}{3c_1\sqrt{p}}\right\}
\end{align*}

Let $\mathcal{B}=\bigcap_{i=1}^3\mathcal{B}_i\subset \mathcal{A}_1^*\subset \mathcal{A}_1$. 
Since the prior probability of $\mathcal{B}$ lower bounds $\Pi(\mathcal{A}_1)$, we now focus on estimating $\tilde{\Pi}(\mathcal{B})$. 
Recall that the untruncated prior $\tilde{\Pi}$ is separable. 
Consequently, 
\begin{align*}
    \Pi(\mathcal{A}_1\cap \{\Omega\succ \tau I\})\ge \tilde{\Pi}(\mathcal{A}_1)\ge \tilde{\Pi}(\mathcal{B})=\prod_{i=1}^3\tilde{\Pi}(\mathcal{B}_i)
\end{align*}

We first bound the probability of $\mathcal{B}_{1}.$ 
Note that we can use only the slab part of the prior to bound this probability.
Specifically, we have
\begin{align*}
    \tilde{\Pi}(\mathcal{B}_1)&=\int_{\mathcal{B}_1}\prod_{(k,k')\in S_0^\Omega}   \pi(\omega_{k,k'}|\eta)d\mu\\
    &\ge \prod_{(k,k')\in S_0^\Omega}  \int_{|\omega_{0,k,k'}-\omega_{k,k'}|\le \frac{\epsilon_n}{6s_0^\Omega c_1\sqrt{p}}} \pi(\omega_{k,k'}|\eta) d\omega_{k,k'}\\
    &\ge \eta^{s_0^\Omega} \prod_{(k,k')\in S_0^\Omega}  \int_{|\omega_{0,k,k'}-\omega_{k,k'}|\le \frac{\epsilon_n}{6s_0^\Omega c_1\sqrt{p}}} \frac{\xi_1}{2} \exp(-\xi_1|\omega_{k,k'}|) d\omega_{k,k'}\\
    &\ge \eta^{s_0^\Omega} \exp(-\xi_1\sum_{(k,k')\in S_0^\Omega}|\omega_{0,k,k'}|)\prod_{(k,k')\in S_0^\Omega}  \int_{|\omega_{0,k,k'}-\omega_{k,k'}|\le \frac{\epsilon_n}{6s_0^\Omega c_1\sqrt{p}}} \frac{\xi_1}{2} \exp(-\xi_1|\omega_{0,k,k'}-\omega_{k,k'}|) d\omega_{k,k'}\\
    &=\eta^{s_0^\Omega}\exp(-\xi_1||\Omega_{0,S_0^\Omega}||_1) \prod_{(k,k')\in S_0^\Omega} \int_{|\Delta| \le \frac{\epsilon_n}{6s^\Omega_0c_1\sqrt{p}}}\frac{\xi_1}{2} \exp(-\xi_1|\Delta|)d\Delta\\
    &\ge \eta^{s_0^\Omega}\exp(-\xi_1||\Omega_{0,S_0^\Omega}||_1)\left[e^{-\frac{\xi_1\epsilon_n}{6c_1s^\Omega_0\sqrt{p}}}\left(\frac{\xi_1\epsilon_n}{6s^\Omega_0c_1\sqrt{p}}\right)\right]^{s^\Omega_0}
\end{align*}

The first inequality holds because the fact that $|\omega_{0,k,k'}-\omega_{k,k'}|\le \epsilon_n/(6s_0^\Omega c_1\sqrt{p})$ implies that the sum less than $\epsilon_n/(6 c_1\sqrt{p}).$

For $\mathcal{B}_2,$ we derive the lower bound using the spike component of the prior.
To this end, let $Q=q(q-1)/2$ denote the number of off-diagonal entries of matrix $\Omega$. 
We have
\begin{align*}
    \tilde{\Pi}(\mathcal{B}_2)&=\int_{\mathcal{B}_2}\prod_{(k,k')\in (S_0^\Omega)^c}   \pi(\omega_{k,k'}|\eta)d\mu\\
    &\ge \prod_{(k,k')\in (S_0^\Omega)^c} \int_{|\omega_{k,k'}|\le \frac{\epsilon_n}{6(Q-s_0^\Omega)c_1\sqrt{p}}} \pi(\omega_{k,k'}|\eta)d\mu\\
    &\ge (1-\eta)^{Q-s_0^\Omega} \prod_{(k,k')\in (S_0^\Omega)^c} \int_{|\omega_{k,k'}|\le \frac{\epsilon_n}{6(Q-s_0^\Omega)c_1\sqrt{p}}} \frac{\xi_0}{2} \exp(-\xi_0|\omega_{k,k'}|)d\omega_{k,k'} \\
    &\ge(1-\eta)^{Q-s_0^\Omega}\prod_{(k,k')\in (S_0^\Omega)^c}\left[1-\frac{6(Q-s_0^\Omega)c_1\sqrt{p}}{\epsilon_n}\mathbb{E}_{\pi}|\omega_{k,k'}|\right]\\
    &=(1-\eta)^{Q-s_0^\Omega}\left[1-\frac{6(Q-s_0^\Omega)c_1\sqrt{p}}{\epsilon_n\xi_0}\right]^{Q-s_0^\Omega}\\
    &\gtrsim (1-\eta)^{Q-s_0^\Omega} \left[1-\frac{1}{Q-s_0^\Omega}\right]^{Q-s_0^\Omega}\\
    &\asymp (1-\eta)^{Q-s_0^\Omega}
\end{align*}

To derive the last two lines, we used the assumption that $\xi_0\sim \max\{Q,n,pq\}^{4+b}$ for some $b>0$ to conclude that $\sqrt{n}/\max\{Q,n,pq\}^{1/2+b}\le 1.$
This inequality allows us to control the $Q$ in the numerator.
Since $s_{0}^{\Omega}$ grows slower than $Q,$ we can lower bound the above function some multiplier of the form $(1-\eta)^{Q-s_0^\Omega}.$
Thus, for large enough $n$, we have
\begin{align*}
	\frac{6(Q-s_0^\Omega)c_1\sqrt{p}}{\epsilon_n\xi_0}&\le \frac{6(Q-s_0^\Omega)c_1\sqrt{p}\sqrt{n}}{\sqrt{p\log(q)}Q^{2+b}}\\
	&=\frac{6c_1}{\sqrt{\log(q)}}\frac{Q-s^\Omega_0}{Q^2}\frac{\sqrt{n}}{Q^b}\\
	&\le \frac{Q-s^\Omega_0}{Q^2}\\
	&\le \frac{1}{Q-s^\Omega_0}
\end{align*}
The event $\mathcal{B}_3$ only involves diagonal entries. 
The untruncated prior mass can be directly bounded using the exponential distribution
\begin{align*}
    \tilde{\Pi}(\mathcal{B}_3)&=\int_{\mathcal{B}_3}\prod_{k=1}^q   \pi(\omega_{k,k})d\mu\\
    &\ge \prod_{k=1}^q  \int_{|\omega_{0,k,k}-\omega_{k,k}|\le \frac{\epsilon_n}{3q c_1\sqrt{p}}} \pi(\omega_{k,k}) d\omega_{k,k}\\
    &=\prod_{k=1}^q  \int_{\omega_{0,k,k}- \frac{\epsilon_n}{3 q c_1\sqrt{p}}}^{\omega_{0,k,k}+ \frac{\epsilon_n}{3 q c_1\sqrt{p}}} \xi_1 \exp(-\xi_1\omega_{k,k}) d\omega_{k,k}\\
    &\ge \prod_{k=1}^q  \int_{\omega_{0,k,k}}^{\omega_{0,k,k}+ \frac{\epsilon_n}{3 q c_1\sqrt{p}}} \xi_1 \exp(-\xi_1\omega_{k,k}) d\omega_{k,k}\\
    &=\exp(-\xi_1 \sum_{i=1}^q\omega_{0,k,k})\int_{0}^{ \frac{\epsilon_n}{3 q c_1\sqrt{p}}} \xi_1 \exp(-\xi_1\omega_{k,k}) d\omega_{k,k}\\
    &\ge \exp(-\xi_1 \sum_{i=1}^q\omega_{0,k,k})\left[e^{-\frac{\xi_1\epsilon_n}{3c_1q\sqrt{p}}}\left(\frac{\xi_1\epsilon_n}{3qc_1\sqrt{p}}\right)\right]^{q}
\end{align*}

Now we are ready to show that the log prior mass on $\mathcal{B}$ can be bounded by some $C_1n\epsilon_n^2$. To this end, consider the negative log probability
\begin{align*}
    -&\log(\Pi(\mathcal{A}_1\cap \{\Omega\succ \tau I\}))\\
    \le& \sum_{i=1}^3 -\log(\tilde{\Pi}(\mathcal{B}_i)) \\
    \lesssim& -s_0^\Omega \log(\eta) + \xi_1||\Omega_{0,S_0^\Omega}||_1 + \frac{\xi_1 \epsilon_n}{6c_1\sqrt{p}}-s_0^\Omega \log\left(\frac{\xi_1\epsilon_n}{6s^\Omega_0c_1\sqrt{p}}\right)-
    (Q-s_0^\Omega) \log(1-\eta)\\
    &+\xi_1\sum_k \omega_{0,k,k}+\frac{\xi_1\epsilon_n}{3c_1\sqrt{p}}-q\log\left(\frac{\xi_1\epsilon_n}{3qc_1\sqrt{p}}  \right)\\
    =&-\log\left( \eta^{s_0^\Omega} (1-\eta)^{Q-s_0^\Omega}  \right) + \xi_1||\Omega_{0,S_0^\Omega}||_1 + \frac{\xi_1 \epsilon_n}{6c_1\sqrt{p}}+\xi_1\sum_k \omega_{0,k,k}+\frac{\xi_1\epsilon_n}{3c_1\sqrt{p}}\\
    &-s_0^\Omega \log\left(\frac{\xi_1\epsilon_n}{6s^\Omega_0c_1\sqrt{p}}\right)-q\log\left(\frac{\xi_1\epsilon_n}{3qc_1\sqrt{p}}  \right)\\
\end{align*}

The $\frac{\xi_1 \epsilon_n}{6c_1\sqrt{p}}$ and $\frac{\xi_1\epsilon_n}{3c_1\sqrt{p}}$ terms are $O(\epsilon_n/\sqrt{p})\lesssim n\epsilon_n^2$ which goes to infinity.
The 4th term is of order $q$ since the diagonal entries is controlled by the largest eigenvalue of $\Omega,$ which was assumed to be bounded. 
\begin{align*}
    \xi_1||\Omega_{0,S_0^\Omega}||_1\le \xi_1 s_0^\Omega \sup|\omega_{0,k,k'}|
\end{align*}
is of order $s_0^\Omega$ as the entries of $\omega_{0,k,k'}$ is controlled. 

Without tuning of $\eta$, the first term $-\log\left( \eta^{s_0^\Omega} (1-\eta)^{Q-s_0^\Omega}  \right)$ has order of $Q$. 
But since we assumed $\frac{1-\eta}{\eta}\sim \max\{Q,pq\}^{2+a}$ for some $a>0$, we have $K_1 \max\{Q,pq\}^{2+a} \le \frac{1-\eta}{\eta}\le K_2 \max\{Q,pq\}^{2+a}$. 
That is, we have $1/(1+K_2 \max\{Q,pq\}^{2+a})\le \eta \le 1/(1+K_1\max\{Q,pq\}^{2+a}).$
We can derive a simple lower bound as
\begin{align*}
    \eta^{s_0^\Omega} (1-\eta)^{Q-s_0^\Omega}&\ge (1+K_2 \max\{Q,pq\}^{2+a})^{-s_0^\Omega}(1-\eta)^{Q-s_0^\Omega}\\
    &\ge (1+K_2 \max\{Q,pq\}^{2+a})^{-s_0^\Omega}\left(1-\frac{1}{1+K_1\max\{Q,pq\}^{2+a}}\right)^{Q-s_0^\Omega}\\
    &\gtrsim (1+K_2 \max\{Q,pq\}^{2+a})^{-s_0^\Omega}
\end{align*}

The last line is because $\max\{Q,pq\}^{2+a}$ grows faster than $Q-s^\Omega_0.$
Thus $(1-\frac{1}{1+K_1\max\{Q,pq\}^{2+a}})^{\max\{Q,pq\}-s_0^\Omega}$ can be bounded below by some constant. 
\begin{align*}
    -\log\left( \eta^{s_0^\Omega} (1-\eta)^{Q-s_0^\Omega}  \right) &\lesssim s_0^\Omega \log(1+K_2 \max\{Q,pq\}^{2+a})\lesssim s_0^\Omega \log(\max\{Q,pq\})\\&\asymp s_0^\Omega\log(\max\{q,p\})\le \max(p,q,s_0^\Omega) \log(\max\{q,p\})
\end{align*}

The last two terms can be treated in the same way, using the assumption $\xi_1\asymp 1/\max\{Q,n\}.$
\begin{align*}
	-s_0^\Omega \log\left(\frac{\xi_1\epsilon_n}{6s^\Omega_0c_1\sqrt{p}}\right)&=s_0^\Omega \log\left(\frac{6s^\Omega_0c_1\sqrt{p}}{\xi_1\epsilon_n}\right)\\
	&\lesssim s_0^\Omega \log\left(\frac{n^{1/2}\max\{n,Q\}s_0^\Omega \sqrt{p}}{\sqrt{\max\{s_0^\Omega,p,q\}\log(q)}}\right)\\
	&\le s_0^\Omega \log\left(n^{1/2}\max\{n,Q\}s_0^\Omega \right)\\
	&\lesssim s_0^\Omega\log(q^2)\\
	&\lesssim n\epsilon_n^2
\end{align*}
The third line holds because $\sqrt{p}\le \sqrt{\max\{s_0^\Omega, p,q\}}$ and $\log(q)\ge 1,$ which together imply that $\sqrt{p}/\sqrt{\max\{s_0^\Omega, p,q\}\log(q)}\le 1$. 
The fourth line follows from our assumption that $\log(n)\lesssim \log(q)$ because $s_0^\Omega <q^2$ and $Q<q^2$. The last line uses the definition of $\epsilon_n$.  

Finally, we have
\begin{align*}
-q \log\left(\frac{\xi_1\epsilon_n}{3qc_1\sqrt{p}}\right)&=q \log\left(\frac{3qc_1\sqrt{p}}{\xi_1\epsilon_n}\right)\\
&\lesssim q \log\left(\frac{n^{1/2}\max\{Q,n\}q \sqrt{p}}{\sqrt{\max\{s_0^\Omega,p,q\}\log(q)}}\right)\\
&\le q \log\left(n^{1/2}\max\{Q,n\}q \right)\\
&\lesssim q\log(q)\\
&\lesssim n\epsilon_n^2
\end{align*}

To bound $\Pi(\mathcal{A}_{2} \vert \mathcal{A}_{1}),$ we used a similar strategy as the one above.
but on $\Psi.$
We show that mass on a L1 norm ball serves as a lower bound similar to that of $\Omega$. 
To see that, we show that powers of $\Omega$ and $\Omega_{0}$ are bounded in operator norm. 
Thus the terms $\frac{1}{n}||X\Delta_\Psi\Omega_0^{-1/2}||_F^2$ and $\frac{2}{n}||X\Delta_\Psi\Omega^{-1/2}||_F^2 $ that appear in the KL condition are bounded by a constant multiplier of $n^{-1}||X\Delta_\Psi||_F^2.$
Using the fact that the columns of $X$ have norm $\sqrt{n},$ we can found this norm:
\begin{align*}
    ||X\Delta_\Psi||_F\le \sqrt{n}\sum_{j=1}^p||\Delta_{\Psi,j,.}||_F\le \sqrt{n}\sum_{j=1}^p\sum_{k=1}^q|\psi_{j,k}-\psi_{0,j,k}|
\end{align*}

Thus to bound $\Pi(\mathcal{A}_2|\mathcal{A}_1)$ from below, it suffices to bound $\Pi(\sum |\psi_{j,k}-\psi_{0,j,k}|\le c_4\epsilon_n)$ for some fixed constant $c_4>0$. 

We separate the sum based on whether the true value is 0, similar to our treatment on $\Omega$:
\begin{align*}
    \sum_{ij} |\psi_{j,k}-\psi_{0,j,k}|&=\sum_{(j,k)\in S_0^\Psi}|\psi_{j,k}-\psi_{0,j,k}|+\sum_{(j,k)\in (S_0^\Psi)^c}|\psi_{j,k}|
\end{align*}

Using the same argument as in $\Omega$, we can consider the events whose intersection is a subset of $\mathcal{A}_2$:
\begin{align*}
    \mathcal{B}_4&=\left\{ \sum_{(j,k)\in S_0^\Psi}|\psi_{j,k}-\psi_{0,j,k}|\le \frac{c_4\epsilon_n}{2}  \right\}\\
    \mathcal{B}_5&=\left\{ \sum_{(j,k)\in (S_0^\Psi)^c}|\psi_{j,k}-\psi_{0,j,k}|\le \frac{c_4\epsilon_n}{2}  \right\}\\
\end{align*}
We have $\mathcal{B}_4\cap \mathcal{B}_5\subset\mathcal{A}_2.$ 
Since the elements of $\Psi$ are \textit{a priori} independent of each other and of $\Omega,$ we compute 
\begin{align*}
    \Pi(\mathcal{A}_2|\mathcal{A}_1)\ge \Pi(\mathcal{B}_4|\mathcal{A}_1)\Pi(\mathcal{B}_5|\mathcal{A}_1)=\Pi(\mathcal{B}_4)\Pi(\mathcal{B}_5)
\end{align*}

We bound each of these terms using the same argument as in the previous subsection:
\begin{align*}
    \Pi(\mathcal{B}_4)&=\int_{\mathcal{B}_4}\prod_{(j,k)\in S_0^\Psi}\pi(\psi_{j,k}|\theta)d\mu\\
    &\ge \prod_{(j,k)\in S_0^\Psi}\int_{|\psi_{j,k}-\psi_{0,j,k}|\le \frac{c_4\epsilon_n}{2s_0^\Psi}}\pi(\psi_{j,k}|\theta)d\psi_{j,k}\\
    &\ge \theta^{s_0^\Psi}\prod_{(j,k)\in S_0^\Psi}\int_{|\psi_{j,k}-\psi_{0,j,k}|\le \frac{c_4\epsilon_n}{2s_0^\Psi}}\frac{\lambda_1}{2}\exp(-\lambda_1|\psi_{j,k}|)d\psi_{j,k}\\
    &\ge \theta^{s_0^\Psi}\exp(-\lambda_1\sum_{(j,k)\in S_0^\Psi}|\psi_{0,j,k}|)\prod_{(j,k)\in S_0^\Psi}\int_{|\psi_{j,k}-\psi_{0,j,k}|\le \frac{c_4\epsilon_n}{2s_0^\Psi}} \frac{\lambda_1}{2}\exp(-\lambda_1|\psi_{j,k}-\psi_{0,j,k}|)d\psi_{j,k}\\
    &=\theta^{s_0^\Psi}\exp(-\lambda_1\sum_{(j,k)\in S_0^\Psi}|\psi_{0,j,k}|)\prod_{(j,k)\in S_0^\Psi}\int_{|\Delta|\le \frac{c_4\epsilon_n}{2s_0^\Psi}} \frac{\lambda_1}{2}\exp(-\lambda_1|\Delta|)d\Delta\\
    &\ge \theta^{s_0^\Psi}\exp(-\lambda_1||\Psi_{0,S_0^\Psi}||_1)\left[e^{-\frac{c_4\lambda_1\epsilon_n}{2s_0^\Psi}}\frac{c_4\epsilon_n}{2s_0^\Psi}\right]^{s_0^\Psi}
\end{align*}

Similarly, we have
\begin{align*}
    \Pi(\mathcal{B}_5)&\ge (1-\theta)^{pq-s_0^\Psi}\left[1-\frac{2(pq-s_0^\Psi)c_4}{\epsilon_n\lambda_0}\right]^{pq-s_0^\Psi}\\
    &\gtrsim (1-\theta)^{pq-s_0^\Psi}
\end{align*}

From here we have
\begin{align*}
    -\log(\Pi(\mathcal{A}_2|\mathcal{A}_1))&\le -\log(\Pi(\mathcal{B}_4))-\log(\Pi(\mathcal{B}_5))\\
    &=-\log(\theta^{s_0^\Psi}(1-\theta)^{pq-s_0^\Psi})+\lambda_1||\Psi_{0,S_0^\Psi}||_1+\frac{\lambda_1c_4\epsilon_n}{2}-s_\Psi^0\log\left(\frac{c_4\epsilon_n}{2s_0^\Psi}\right)
\end{align*}

Since $\Psi_0$ has bounded L2 operator norm, we know that the entries of $\Psi_0$ are all bounded. 
Thus$\lambda_1||\Psi_{0,S_0^\Psi}||_1=O(s_0^\Psi)\lesssim n\epsilon_n^2$. 
The last two terms are $O(\epsilon_n)\lesssim n\epsilon_n^2$.

For the first term, recall that we assumed $\frac{1-\theta}{\theta}\sim (pq)^{2+b}$ for some $b>0.$
That is, there are constants $M_{3}$ and $M_{4}$ such that $M_3(pq)^{2+b}\le \frac{1-\theta}{\theta} \le M_4(pq)^{2+b}\le \frac{1-\theta}{\theta}$. 
Since $1/(1+M_4(pq)^{2+b})\le \theta \le 1/(1+M_3(pq)^{2+b}),$ we compute
\begin{align*}
    \theta^{s_0^\Psi}(1-\theta)^{pq-s_0^\Psi}&\ge(1+M_4(pq)^{2+b})^{-s_0^\Psi}(1-\theta)^{pq-s_0^\Psi}\\
    &\ge(1+M_4(pq)^{2+b})^{-s_0^\Psi}\left(1-1/(1+M_3(pq)^{2+b})\right)^{pq-s_0^\Psi}\\
    &\gtrsim (1+M_4(pq)^{2+b})^{-s_0^\Psi}
\end{align*}
Note that the last line is due to the fact that $(pq)^{2+b}$ grows faster than $pq-s_0^\Psi.$
Consequently, the term $\left(1-1/(1+M_3(pq)^{2+b})\right)^{pq-s_0^\Psi}$ can be bounded from below by a constant not depending on $n.$
Thus,
\begin{align*}
    -\log\left(\theta^{s_0^\Psi}(1-\theta)^{pq-s_0^\Psi}\right)\lesssim s_0^\Psi\log(1+M_4(pq)^{2+b})\lesssim s_0^\Psi\log(pq)\lesssim s_0^\Psi\max\{\log(q),\log(p)\}
\end{align*}

For the last term, we use the same argument as we did with $\Omega$ to conclude.
\begin{align*}
	-s_\Psi^0\log\left(\frac{c_4\epsilon_n}{2s_0^\Psi}\right)&=s_\Psi^0\log\left(\frac{2s_0^\Psi}{c_4\epsilon_n}\right)\\
	&\lesssim s_0^\Psi\log\left(\frac{\sqrt{n}}{\sqrt{\log(pq)}}\right)\\
	&\le s_0^\Psi\log(\sqrt{n})\\
	&\lesssim n\epsilon_n^2
\end{align*}

\end{proof}

\begin{myCorollary}
    Let $E_{n}$ be the event 
    $$
    E_n=\{Y:\iint f/f_0d\Pi(\Psi)\Pi d(\Omega)\ge e^{-C_1n\epsilon_n^2}\}.
    $$
    Then for all $(\Psi_{0}, \Omega_{0},$ we have $\P_{0}(E_{n}^{c}) \to 0$ as $n \to \infty.$
\end{myCorollary}

\begin{proof}
    By Lemma 8.1 of \citet{ghosal2017fundamentals} and KL condition. 
\end{proof}


To simplify the parameter space to be concerned in the test condition, we first show the dimension recovery result by bounding the prior probability, with our effective dimension defined as number of entries whose absolute value is larger than the intersection of spike and slab components. 
Then we find the proper vectorized L1 norm sieve in the ``lower-dimensional" parameter space. 
We construct tests based on the supremum of a collection of single-alternative Neyman-Pearson likelihood ratio tests in the subsets of the sieve that are norm balls, then we show that the number of such subsets needed to cover the sieve can be bounded properly. 

Unlike \citet{Ning2020}, our prior assigns no mass on exactly sparse solutions.
Nevertheless, similar to \citet{RockovaGeorge2018_ssl}, we can define a notion of ``effective sparsity'' and generalized dimension.
Intuitively the generalized dimension can be defined as how many coefficients are drawn from the slab rather than the spike part of the prior. 

\begin{myDefinition}[Generalized dimension]

    Then the generalized dimension can be defined as number of entries exceeding some threshold:
    \begin{equation}
        \label{eqn:eff_dim}
        \begin{aligned}
            |\nu_{\psi}(\Psi)|&=\sum_{jk}\nu_{\psi}(\psi_{j,k})\\
            |\nu_{\omega}(\Omega)|&=\sum_{k>k'}\nu_{\omega}(\omega_{k,k'})
        \end{aligned}
    \end{equation}

    where the generalized inclusion functions $\nu_{\psi}$ and $\nu_{\omega}$ for $\Psi$ and $\Omega$ is:
    \begin{align*}
        \nu_{\psi}(\psi_{j,k})&=\mathbbm{1}(|\psi_{j,k}|>\delta_{\psi})\\
        \nu_{\omega}(\omega_{k,k'})&=\mathbbm{1}(|\omega_{k,k'}|>\delta_{\omega})
    \end{align*}
    with $\delta_{\psi}$ and $\delta_{\omega}$ being the threshold where the spike and slab part has the same density. 
    
    \begin{align*}
        \delta_{\psi}&=\frac{1}{\lambda_0-\lambda_1}\log\left[\frac{1-\theta}{\theta}\frac{\lambda_0}{\lambda_1}\right]\\
        \delta_{\omega}&=\frac{1}{\xi_0-\xi_1}\log\left[\frac{1-\eta}{\eta}\frac{\xi_0}{\xi_1}\right]\\
    \end{align*}

\end{myDefinition}
Note that we only count the off-diagonal entries in $\Omega$.

We are now ready to prove Lemma 1 from the main text.
The main idea is to check the posterior probability directly. 

\begin{proof}
Let $\mathcal{B}^\Psi_n=\{\Psi:|\nu_{\psi}(\Psi)|<r^\Psi_n\}$ for some $r_n^\Psi=C_3'\max\{p,q,s_0^\Psi,s_0^\Omega\}$ with $C_3'>C_1$ in the KL condition. 
For $\Omega$, let $\mathcal{B}^\Omega_n=\{\Omega\succ \tau I:|\nu_{\omega}(\Omega)|<r^\Omega_n\}$ for $r_n^\Omega=C_3'\max\{p,q,s_0^\Psi,s_0^\Omega\}$ with some $C_3'>C_1$ in the KL condition. 
We aim to show that $\E_0\Pi(\Omega\in(\mathcal{B}^\Omega_n)^c|Y_1,\dots,Y_n )\to 0$ and $\E_0\Pi(\Psi\in(\mathcal{B}^\Psi_n)^c|Y_1,\dots,Y_n )\to 0$.

The marginal posterior can be expressed using log-likelihood $\ell_n$:
\begin{align}
    \label{eqn:posterior_B}
    \begin{split}
        \Pi(\Psi\in \mathcal{B}_n^\Psi|Y_1,\dots,Y_n)&=\frac{\iint_{\mathcal{B}_n^\Psi}\exp(\ell_n(\Psi,\Omega)-\ell_n(\Psi_0,\Omega_0))d\Pi(\Psi)d\Pi(\Omega)}{\iint\exp(\ell_n(\Psi,\Omega)-\ell_n(\Psi_0,\Omega_0))d\Pi(\Psi)d\Pi(\Omega)}\\
        \Pi(\Omega\in \mathcal{B}_n^\Omega|Y_1,\dots,Y_n)&=\frac{\iint_{\mathcal{B}_n^\Omega}\exp(\ell_n(\Psi,\Omega)-\ell_n(\Psi_0,\Omega_0))d\Pi(\Psi)d\Pi(\Omega)}{\iint\exp(\ell_n(\Psi,\Omega)-\ell_n(\Psi_0,\Omega_0))d\Pi(\Psi)d\Pi(\Omega)}
    \end{split}
\end{align}
By using the result of KL condition (Lemma \ref{lemma:KL_cg}), we know the denominators are bounded from below by $e^{-C_1n\epsilon_n^2}$ with large probability.
Thus, we focus now on upper bounding the numerators beginning with $\Psi.$

Consider the numerator:
\begin{align*}
    \E_0\left(\iint_{(\mathcal{B}_n^{\Psi})^c}f/f_0d\Pi(\Psi)d\Pi(\Omega)\right)&=\int\iint_{(\mathcal{B}_n^{\Psi})^c}f/f_0d\Pi(\Psi)d\Pi(\Omega)f_0 dy\\
    &=\iint_{(\mathcal{B}_n^{\Psi})^c}\int f dyd\Pi(\Psi)d\Pi(\Omega)\\
    &\le \int_{(\mathcal{B}_n^{\Psi})^c}d\Pi(\Psi)=\Pi(|\nu_{\psi}(\Psi)|\ge r_n^\Psi)
\end{align*}

We can bound the above display using the fact that when $|\psi_{j,k}|>\delta_\psi$ we have $\pi(\psi_{j,k})<2\theta\frac{\lambda_1}{2}\exp(-\lambda_1|\psi_{j,k}|)$, this is by definition of the effective dimension:

\begin{align*}
    \Pi(|\nu_{\psi}(\Psi)|\ge r_n^\Psi)&\le \sum_{|S|>r_n^\Psi}(2\theta)^{|S|}\prod_{(j,k)\in S}\int_{|\psi_{j,k}|>\delta_\psi} \frac{\lambda_1}{2}\exp(-\lambda_1|\psi_{j,k}|) d\psi_{j,k} \prod_{(j,k)\notin S} \int_{|\psi_{j,k}|<\delta_\psi}\pi(\psi_{j,k})d\psi_{j,k}\\
    &\le \sum_{|S|>r_n^\Psi}(2\theta)^{|S|}\\
\end{align*}

Using the assumption on $\theta$, and the fact $\binom{pq}{k}\le (epq/k)^k$, we can further upper bound the probability
\begin{align*}
    \Pi(|\nu_{\psi}(\Psi)|\ge r_n^\Psi)&\le \sum_{|S|>r_n^\Psi}(2\theta)^{|S|}
    \le \sum_{|S|>r_n^\Psi}(\frac{2}{1+M_4 (pq)^{2+b}})^{|S|}\\
    &\le \sum_{k=\left\lfloor r_n^\Psi\right\rfloor +1}^{pq} \binom{pq}{k}\left(\frac{2}{M_4(pq)^2}\right)^k
    \le \sum_{k=\left\lfloor r_n^\Psi\right\rfloor +1}^{pq} \left(\frac{2e}{M_4kpq}\right)^k\\
    &<\sum_{k=\left\lfloor r_n^\Psi\right\rfloor +1}^{pq} \left(\frac{2e}{M_4(\left\lfloor r_n^\Psi\right\rfloor +1)pq}\right)^k\\
    &\lesssim (pq)^{-(\left\lfloor r_n^\Psi\right\rfloor +1)}\\
    &\le \exp(-(\left\lfloor r_n^\Psi\right\rfloor)\log(pq)).\\
\end{align*}

Taking $r_n^\Psi=C_3'\max\{p,q,s_0^\Psi,s_0^\Omega\}$ for some $C_3'>C_1$, we have:
\begin{align*}
    \Pi(|\nu_{\psi}(\Psi)|\ge r_n^\Psi)&\le \exp(-C_3'\max\{p,q,s_0^\Psi,s_0^\Omega\}\log(pq))
\end{align*}

Therefore,
\begin{align*}
    \E_0\Pi((\mathcal{B}_n^\Psi)^c|Y_1,\dots,Y_n)\le \E_0\Pi((\mathcal{B}_n^\Psi)^c|Y_1,\dots,Y_n)I_{E_n}+P_0(E_n^c),
\end{align*}
where $E_n$ is the event in the KL condition. 
On $E_n$, the KL condition ensures that the denominator in Equation~\eqref{eqn:posterior_B} is lower bounded by $\exp(-C_1n\epsilon_n^2)$ while the denominator is upper bounded by $\exp(-C_3'\max\{p,q,s_0^\Psi,s_0^\Omega\}\log(pq)).$, 
Since $\P_0(E_n^c)$ is $o(1)$ per KL condition, we have the upper bound
\begin{align*}
    \E_0\Pi((\mathcal{B}_n^\Psi)^c|Y_1,\dots,Y_n)&\le \exp(C_1n\epsilon_n^2-C_3'\max\{p,q,s_0^\Psi,s_0^\Omega\}\log(pq))+o(1)\to 0
\end{align*}
This completes the proof of the dimension recovery result of $\Psi.$

The workflow for $\Omega$ is very similar, except we need to use the upper bound of the graphical prior in Equation~\eqref{eqn:graphical_prior_bound} to properly bound the prior mass. 

We upper bound the numerator:
\begin{align*}
    \E_0\left(\iint_{(\mathcal{B}_n^{\Omega})^c}f/f_0d\Pi(\Psi)d\Pi(\Omega)\right)&\le \int_{(\mathcal{B}_n^{\Omega})^c}d\Pi(\Omega)=\Pi(|\nu_{\omega}(\Omega)|\ge r_n^\Omega)\le \exp(2\xi_1 Q-\log(R))\tilde{\Pi}(|\nu_{\omega}(\Omega)|\ge r_n^\Omega)
\end{align*}

We bound the above display using the fact that when $|\omega_{k,k'}|>\delta_\omega$ we have $\pi(\omega_{k,k'})<2\eta\frac{\xi_1}{2}\exp(-\xi_1|\omega_{k,k'}|)$.
Note that this follows from the definition of the effective dimension.
We have
\begin{align*}
    \tilde{\Pi}(|\nu_{\omega}(\Omega)|\ge r_n^\Omega)&\le \sum_{|S|>r_n^\Omega}(2\eta)^{|S|}\prod_{(k,k')\in S}\int_{|\omega_{k,k'}|>\delta_\omega} \frac{\xi_1}{2}\exp(-\xi_1|\omega_{k,k'}|) d\omega_{k,k'} \prod_{(k,k')\notin S} \int_{|\omega_{k,k'}|<\delta_\omega}\pi(\omega_{k,k'})d\omega_{k,k'}\\
    &\le \sum_{|S|>r_n^\Omega}(2\eta)^{|S|}\\
\end{align*}

By using the assumption on $\eta$, and the fact $\binom{Q}{k}\le (eQ/k)^k$, we can further upper bound the probability:
\begin{align*}
    \tilde{\Pi}(|\nu_{\omega}(\Omega)|\ge r_n^\Omega)&\le \sum_{|S|>r_n^\Omega}(2\eta)^{|S|}
    \le \sum_{|S|>r_n^\Omega}(\frac{2}{1+K_4 \max\{pq,Q\}^{2+b}})^{|S|}\\
    &\le \sum_{k=\left\lfloor r_n^\Omega\right\rfloor +1}^{Q} \binom{Q}{k}\left(\frac{2}{K_4\max\{pq,Q\}^2}\right)^k
    \le \sum_{k=\left\lfloor r_n^\Omega\right\rfloor +1}^{\max\{pq,Q\}} \binom{\max\{pq,Q\}}{k}\left(\frac{2}{K_4\max\{pq,Q\}^2}\right)^k\\
    &\le \sum_{k=\left\lfloor r_n^\Omega\right\rfloor +1}^{\max\{pq,Q\}} \left(\frac{2e}{K_4k\max\{pq,Q\}}\right)^k
    <\sum_{k=\left\lfloor r_n^\Omega\right\rfloor +1}^{\max\{pq,Q\}} \left(\frac{2e}{K_4(\left\lfloor r_n^\Omega\right\rfloor +1)\max\{pq,Q\}}\right)^k\\
    &\lesssim \max\{pq,Q\}^{-(\left\lfloor r_n^\Omega\right\rfloor +1)}\\
    &\le \exp(-(\left\lfloor r_n^\Omega\right\rfloor)\log(\max\{pq,Q\}))\\
\end{align*}

Taking $r_n^\Omega=C_3'\max\{p,q,s_0^\Psi,s_0^\Omega\}$ and $C_3'>C_1$, we have
\begin{align*}
    \tilde{\Pi}(|\nu_{\omega}(\Omega)|\ge r_n^\Omega)&\le \exp(-C_3'\max\{p,q,s_0^\Psi,s_0^\Omega\}\log(\max\{pq,Q\}))\le \exp(-C_3n\epsilon_n^2)
\end{align*}

Thus, using the assumption $\xi_1 \asymp 1/max\{Q,n\}$, for some $R'$ not depending on $n$, we have
\begin{align*}
    \Pi(|\nu_{\omega}(\Omega)|\ge r_n^\Omega)\le \exp(-C_3n\epsilon_n^2+2\xi_1 Q-\log(R))\le \exp(-C_3n\epsilon_n^2+\log(R'))
\end{align*}

We therefore conclude that
\begin{align*}
    \E_0\Pi((\mathcal{B}_n^\Omega)^c|Y_1,\dots,Y_n)\le \E_0\Pi((\mathcal{B}_n^\Omega)^c|Y_1,\dots,Y_n)I_{E_n},+P_0(E_n^c)
\end{align*}
where $E_n$ is the event in KL condition. 
On $E_n$, the KL condition ensures that the denominator in Equation~\eqref{eqn:posterior_B} is lower bounded by $\exp(-C_1n\epsilon_n^2)$ while the denominator is upper bounded by $\exp(-C_3'n\epsilon_n^2+\log(R')).$
Since $\P_0(E_n^c)$ is $o(1)$ per KL condition, we conclude
\begin{align*}
    \E_0\Pi((\mathcal{B}_n^\Omega)^c|Y_1,\dots,Y_n)&\le \exp(C_1n\epsilon_n^2-C_3'n\epsilon_n^2+\log(R'))+o(1)\to 0
\end{align*}

\end{proof}

We pause now to reflect on how dimension recovery can help us establish contraction. Our end goal is to show the posterior distribution contract to the true value by first showing that event with log-affinity difference larger than any given $\epsilon>0$ has an $o(1)$ posterior mass. 
For any such event, we can take a partition based on whether it intersects with $\mathcal{B}^\Psi_n, \mathcal{B}^\Omega_n$ or their complements. 
Because the complements $(\mathcal{B}^\Psi_n)^c$ and $(\mathcal{B}^\Omega_n)^c$ have $o(1)$ posterior mass, we have the partition that intersects with any of these two complements also has $o(1)$ posterior mass. 
Thus, we only need to show that events with log-affinity difference larger than any given $\epsilon>0$ \textit{and} recovered the low dimension structure have an $o(1)$ posterior mass. 
The recovery condition reduces the complexity of the events (on the parameter space) that we need to deal with by reducing the effective dimension of such events. 
We will make use of this low dimension structure during checking the test condition. 

Formally for every $\epsilon>0$, we have
\begin{align*}
    &\E_0\Pi(\Psi,\Omega\succ \tau I:\frac{1}{n}\sum \rho(f_i, f_{0,i})>\epsilon|Y_1,\dots, Y_n)\\
    \le& \E_0\Pi(\Psi\in \mathcal{B}_n^\Psi,\Omega\succ \tau I:\frac{1}{n}\sum \rho(f_i, f_{0,i})>\epsilon|Y_1,\dots, Y_n)+\E_0\Pi((\mathcal{B}^\Psi_n)^c|Y_1,\dots,Y_n)\\
    \le& \E_0\Pi(\Psi\in \mathcal{B}_n^\Psi,\Omega\in \mathcal{B}_n^\Omega:\frac{1}{n}\sum \rho(f_i, f_{0,i})>\epsilon|Y_1,\dots, Y_n)\\
    &+\E_0\Pi((\mathcal{B}^\Psi_n)^c|Y_1,\dots,Y_n)+\E_0\Pi((\mathcal{B}^\Omega_n)^c|Y_1,\dots,Y_n)
\end{align*}
The last two terms are $o(1),$ as proved above. 

\subsubsection{Sieve}
\label{appendix:sieve_exists}

As shown in the previous section, we can concentrate on the events with proper dimension recovery, i.e. $\{\Psi\in\mathcal{B}^\Psi_n,\Omega\in \mathcal{B}^\Omega_n\}$. 
To apply \citet{ghosal2017fundamentals}'s general theory of posterior contraction, to establish contraction on the event of proper dimension recovery (i.e. $\E_0\Pi(\Psi\in \mathcal{B}_n^\Psi,\Omega\in \mathcal{B}_n^\Omega:\frac{1}{n}\sum \rho(f_i, f_{0,i})>\epsilon|Y_1,\dots, Y_n)\to 0$).

The next technical lemma established that an L1 sieve is sufficient to cover enough support of the prior. 

\begin{myLemma}   
    Define the sieve $\mathcal{F}_n$: 
    \begin{equation}
        \mathcal{F}_n=\left\{\Psi\in \mathcal{B}^\Psi_n,\Omega\in \mathcal{B}^\Omega_n:||\Psi||_1\le 2C_3p, ||\Omega||_1 \le 8C_3q \right\}
    \end{equation}
    
    for some large $C_3>C_1+2+\log(3)$ where $C_1$ is the constant in KL condition. 
    We have for some $C_2$
    \begin{equation}
        \Pi(\mathcal{F}_n^c)\le \exp(-C_2n\epsilon_n^2)
    \end{equation}
\end{myLemma}

\begin{proof}
    We first observe that  
    \begin{align*}
        \Pi(\mathcal{F}_n^c)&\le \Pi(||\Psi||_1> 2C_3p)+\Pi((||\Omega||_1 > 8C_3q)\cap \{\Omega\succ \tau I\})
    \end{align*}
    We upper bound each term using the bound in Equation~\eqref{eqn:graphical_prior_bound}, we know that
\begin{align*}
    \Pi((||\Omega||_1 > 8C_3q)\cap \{\Omega\succ \tau I\})\le \exp(2\xi_1 Q-\log(R))\tilde{\Pi}(||\Omega||_1 > 8C_3q).
\end{align*}
Since $||\Omega||_1=2\sum_{k>k'}|\omega_{k,k'}|+\sum_{k}|\omega_{k,k}|$, at least one of these two sums exceeds $8C_{3}q/2.$
Thus, we can form an upper bound on the L1 norm probability
\begin{align*}
    \tilde{\Pi}(||\Omega||_1 > 8C_3q)\le \tilde{\Pi}\left(\sum_{k>k'}|\omega_{k,k'}|>\frac{8C_3q}{4}\right)+\tilde{\Pi}\left(\sum_{k}|\omega_{k,k}|>\frac{8C_3q}{2}\right).
\end{align*}

To get an upper bound under $\tilde{\Pi},$ we can act as if all $\omega_{k,k'}$'s were drawn from the slab distribution.
In that setting, $\sum_{k>k'}|\omega_{k,k'}|$ is Gamma distributed with shape parameter $Q$ and rate parameter $\xi_1$. 
By using an appropriate tail probability for the Gamma distribution (\citet{boucheron2013concentration}, pp.29) and the fact $1+x-\sqrt{1+2x}\ge (x-1)/2,$ we compute
\begin{align*}
    \exp(2\xi_1 Q-\log(R))\tilde{\Pi}(\sum_{k>k'}|\omega_{k,k'}| > 8C_3q/4)&\le \exp\left[-Q\left(1-\sqrt{1+2\frac{8C_3q}{4Q\xi_1}}+\frac{8C_3q}{4Q\xi_1}\right)+2Q-\log(R)\right]\\
    &\le \exp\left[-\frac{8C_3q}{8\xi_1}+\left(\frac{5}{2}Q-\log(R)\right)\right]
\end{align*}

Since we have assumed $\xi_{1} \asymp 1/\max\{n,Q\},$ for sufficiently large $n,$ we have $ n\epsilon_n^2\ge q\log(q).$
Consequently, $qn\epsilon_n^2\ge Q\log(q),$ $Q=o(qn\epsilon_n^2)$, and we see that
\begin{align*}
    \frac{8C_3q}{8\xi_1}-\left(\frac{5}{2}Q-\log(R)\right)&\asymp C_3(\max\{n,Q\}q)-\left(\frac{5}{2}Q-\log(R)\right)\\
    &\ge C_3(qn\epsilon_n^2)-\left(\frac{5}{2}Q-\log(R)\right)\\
    &=C_3(qn\epsilon_n^2)-o(qn\epsilon_n^2)\\
    &\ge C_3n\epsilon_n^2
\end{align*}

The first order term of $Q$ on the left hand side can be ignored when $n$ large as the left hand side is dominated by the $Q\log(q)$ term. 
Note that we used the assumption that $\epsilon_n\to 0$. 
We further have
\begin{align*}
    \exp(2 \xi_1 Q-\log(R))\tilde{\Pi}((\sum_{k>k'}|\omega_{k,k'}| > 8C_3q/4))\le \exp(-C_3n\epsilon_n^2)
\end{align*}

For the diagonal, the sum follows a gamma distribution with shape $q$ and rate $\xi_1$. We obtain a similar bound
\begin{align*}
    \exp(2 \xi_1 Q-\log(R))\tilde{\Pi}(\sum_{k}|\omega_{k,k}| > 8C_3q/2)&\le \exp(2 Q-\log(R))\exp\left[-q\left(1-\sqrt{1+2\frac{8C_3q}{2q\xi_1}}+\frac{8C_3q}{2q\xi_1}\right)\right]\\
    &\le \exp\left[-\frac{8C_3q}{4\xi_1}+Q\left(2+\frac{q}{2Q}\right)-\log(R)\right]
\end{align*}

Using the same argument as before and the fact that $\xi_1 \asymp 1/\max\{Q,n\}$, we have
\begin{align*}
\frac{8C_3q}{4\xi_1}-Q\left(2+\frac{q}{2Q}\right)+\log(R)&\asymp 2C_3(\max\{Q,n\}q))-Q\left(2+\frac{q}{2Q}\right)+\log(R)\\
&\ge  C_3qn\epsilon_n^2-o(qn\epsilon_n^2)\\
&\ge C_3n\epsilon_n^2
\end{align*}

The first order term of $Q$ on the left hand side can be ignored when $n$ large as the left hand side is dominated by the $Q\log(q)$ term and $q/Q\to 0$. 

By combining the above results, we have: 

\begin{equation}
\label{eqn:sieve_Omega_bound}
\begin{aligned}
	\Pi((||\Omega||_1 > 8C_3q)\cap \{\Omega\succ \tau I\})\le& \exp(2 Q-\log(R))\tilde{\Pi}(||\Omega||_1 > 8C_3q)\\
	\le& \exp(2 Q-\log(R))\tilde{\Pi}(\sum_{k>k'}|\omega_{k,k'}|>\frac{8C_3q}{4})\\
    &+\exp(2 Q-\log(R))\tilde{\Pi}(\sum_{k}|\omega_{k,k}|>\frac{8C_3q}{2})\\
	\le& 2\exp(-C_3n\epsilon_n^2)
\end{aligned}
\end{equation}

The probability $||\Psi||_1>2C_3p$ can be bounded by tail probability of Gamma distribution with shape parameter $pq$ and rate parameter $\lambda_1$:
\begin{align*}
	\Pi(||\Psi||_1> 2C_3p)&\le \exp\left[-pq\left(1-\sqrt{1+2\frac{2C_3p}{pq\lambda_1}}+\frac{2C_3p}{pq\lambda_1}\right)\right]\\
	&\le  \exp\left[-pq\left(\frac{2C_3p}{2pq\lambda_1}-\frac{1}{2}\right)\right]\\
	&\le \exp\left(-\frac{2C_3p}{2\lambda_1}+\frac{pq}{2}\right)\\
\end{align*}

Using the same argument, we have $pn\ge pn\epsilon_n^2\ge pq\log(q)$ and thus, $pq=o(pn\epsilon_n^2)$ for large $n$. 
Consequently,
\begin{align*}
	\exp\left(-\frac{2C_3p}{2\lambda_1}+\frac{pq}{2}\right)\le \exp\left(-C_3pn\epsilon_n^2+o(pn\epsilon_n^2)\right)\le \exp(-C_3n\epsilon_n^2)
\end{align*}
and 
\begin{equation}
\label{eqn:sieve_B_bound}
\begin{aligned}
	\Pi(||\Psi||_1> 2C_3p)&\le\exp(-C_3n\epsilon_n^2)
\end{aligned}
\end{equation}

By combining the result from Equations~\eqref{eqn:sieve_Omega_bound} and ~\eqref{eqn:sieve_B_bound}, we conclude 
\begin{align*}
	\Pi(\mathcal{F}_n^c)\le 3\exp(-C_3n\epsilon_n^2)=\exp(-C_3n\epsilon_n^2+\log(3)).
\end{align*}
With our choice of $C_3$, the above probability is asymptotically bounded from above by $\exp(-C_2n\epsilon_n^2)$ with some $C_2\ge C_1+2$. 
\end{proof}

\subsubsection{Test condition}

In the next lemma we establish the so-called ``test condition'', i.e. the model family is not too large to be identified. 

The strategy is to construct a sequence of Neyman-Pearson tests with alternatives in a ball around representative points and show number of such balls is properly bounded. To do so we first establish a lemma for packing a ``effectively low dimensional'' sets. 

The packing number of a set usually depends exponentially on the sets dimensions.
Because \citet{Ning2020} studied posteriors which place positive probability on exactly sparse parameters, they were able to directly bound the packing number of suitable low-dimensional sets.
In our case, which uses an absolutely continuous prior, we need to instead control the packing number of ``effectively low dimensional'' spaces.

Lemma~\ref{lemma:packing_number_lp} provides a sufficient condition for bounding the complexity (evaluated by packing number) of an set of ``effectively sparse'' vectors can be bounded by the complexity of a set of actually sparse vectors. 

\begin{myLemma}[packing a shallow cylinder in Lp]\footnote{In an early version of this paper this lemma is numbered S4}
    \label{lemma:packing_number_lp}
    For a set of form $E=A \times [-\delta,\delta]^{Q-s}\subset \mathbb{R}^Q$ where $A\subset \mathbb{R}^s$, (with $s>0$ and $Q\ge s+1$ are integers) for $1\le p<\infty$ and a given $T>1$, if $\delta<\frac{\epsilon}{2[T(Q-s)]^{1/p}}$, we have the packing number:

    \begin{align*}
        D(\epsilon, A, ||\cdot||_p) \le D(\epsilon,E, ||\cdot ||_p)\le D((1-T^{-1})^{1/p}\epsilon, A, ||\cdot||_p)
    \end{align*}
\end{myLemma}

\begin{proof}
    The lower bound is trivial by observing $A\times\{0\}^{Q-s}\subset E$ and the packing number of $A\times\{0\}^{Q-s}$ is exactly the packing number of $A$. 
    For the upper bound, we show that for each packing on $E,$ we can slice that packing with the 0-plane to form a packing on $A$ with the same number of balls but smaller radius (see Figure~\ref{fig:packing} for an illustration). 

    We first show any Lp $\epsilon/2-$ball $B_\theta(\epsilon/2)$ centered in the set $E$ intersects the plane $\mathbb{R}^{s}\times \{0\}^{Q-s}$. 
    Assume the center is $\theta=(x_1,\dots, x_Q)$. 
    It suffices to show the center's distance to the plane is less than the radius of the ball. 
    Since the center is in $E$, we have $|x_i|\le \delta$ for the last $Q-s$ coordinates. 
    Denote the projection of the center on the plane as $\theta_A=(x_1,\dots, x_s,0)\in A\times \{0\}^{Q-s}$.
    Then the Lp distance from the center to the plane is
    \begin{align*}
        ||\theta_A-\theta||_p^p=\sum_{i=s+1}^Q |x_i|^p\le (Q-s)\delta^p<T^{-1}(\epsilon/2)^{p}
    \end{align*}
    
    Next we show the slice $B_\theta(\epsilon)\cap \mathbb{R}^{s}\times \{0\}^{Q-s}$ is also a ball centered at $\theta_A$ in the lower dimensional plane. 
    It suffice to show the boundary is a sphere. 
    Suppose we take a point $a$ from the intersection of boundary of $B_\theta(\epsilon)\cap \mathbb{R}^{s}\times \{0\}^{Q-s}$, the vector from center to the point can be decomposed to sum of two orthogonal component, namely the vector from $\theta_A$ to $a$ and from $\theta_A$ to $\theta$, we have in this case
    \begin{align*}
        ||a-\theta_A||_p^p+||\theta_A-\theta||_p^p=||a-\theta||_p^p=\epsilon^p/2^p
    \end{align*}
    because $a-\theta_A$ has all 0 entries on the last $Q-s$ axis and $\theta_A-\theta$ has all 0 entries on the first $s$ entries. 
    Thhus any such point has a fixed distance to $\theta_A$, the projection of the center $\theta$ on the plane of $A$.
Notice that
    \begin{align*}
        ||a-\theta_A||_p^p=\epsilon^p/2^p-||\theta_A-\theta||_p^p,
    \end{align*}
    which is fixed. Thus the collection of $a$ forms a sphere on $A$'s plane. 
    
    From here, we can also lower bound the radius of slice by $(1-T^{-1})^{1/p}\epsilon/2$ since $||\theta_A-\theta||_p^p<T^{-1}(\epsilon/2)^p$, thus we have the radius $||a-\theta_A||_p> (1-T^{-1})^{1/p}\epsilon/2$. Thus, we have the smaller ball must lie within the slice, i.e.
    \begin{equation}
        \label{eqn:lower_dim_balls}
        \begin{aligned}
            B_{\theta_A}((1-T^{-1})^{1/p}\epsilon/2)\times\{0\}^{Q-s}\subset \left( B_\theta(\epsilon/2)\cap (\mathbb{R}^s\times\{0\}^{Q-s})\right) \subset B_\theta(\epsilon/2)
        \end{aligned}
    \end{equation}
  
    That is, any $\epsilon/2-$ball centered in $E$ has a corresponding $(1-T^{-1})^{1/p}\epsilon/2$ lower dimension ball centered in $A$. 
    With the above observations in hand, we can now prove the inequality by contradiction.

    Suppose we have a packing on $E$ $\{\theta_1,\dots,\theta_D\}$, where $D$ is larger than the packing number of $A$ in the main result. 
    By Equation~\eqref{eqn:lower_dim_balls}, the lower dimension balls $B_{\theta_{iA}}((1-T^{-1})^{1/p}\epsilon/2)$ must also be disjoint. 
    Since the centers of the balls $\theta_{iA}\in A$, these balls form a packing of $A$ with radius $\epsilon'=(1-T^{-1})^{1/p}\epsilon$. 
    That is, we can find a packing with more balls than the packing number, yielding the desired contradiction. Thus we must have
    \begin{align*}
        D\le D((1-T^{-1})^{1/p}\epsilon, A, ||\cdot||_p)
    \end{align*}
\end{proof}

\begin{figure}[htp]
    \centering
    \includegraphics[width = 0.5\linewidth]{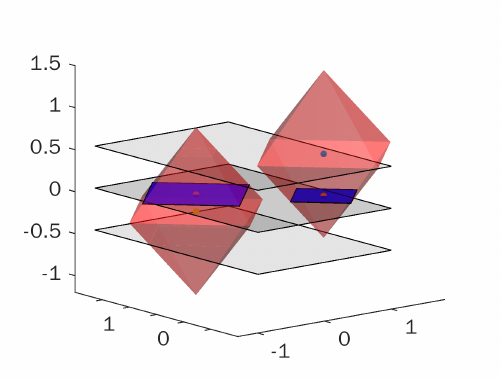}
    \caption{A schematic of argument used in the packing number lemma proof. We showed two disjoint unit L1 balls (red) centered in $(0.8,0,0.5)$ and $(-.3,1,-.2)$, all with in $A \times [-0.5,0.5]$ (with $A=[-1,1]\times [-1,1]$ shown in the middle plane), their slice in the $z=0$ plane (in blue) also forms L1 balls in $\mathbb{R}^2$ whose radius are lower bounded and centered within $A$, thus induced a packing of the lower dimensional set.} 
    \label{fig:packing}
\end{figure}

\begin{myLemma}
    \label{lemma:test_condition}
    There exists test $\varphi_n$, being maximum of single alternative tests, such that for some $M_2>C_1+1$:
    \begin{equation}
        \label{eqn:test_condition}
        \begin{aligned}
            &\mathbb{E}_{f_0}\varphi_n\lesssim e^{-M_2n\epsilon^2/2}\\
            \sup_{f\in \mathcal{F}_n:\rho(f_0,f)>M_2n\epsilon_n^2}&\mathbb{E}_{f}(1-\varphi_n)\lesssim e^{-M_2n\epsilon_n^2}
        \end{aligned}
    \end{equation}
    where $f=\prod_{i=1}^n \mathcal{N}(X_i\Psi\Omega^{-1},\Omega^{-1})$ while $f_0=\prod_{i=1}^n \mathcal{N}(X_i\Psi_0\Omega_0^{-1},\Omega_0^{-1})$
\end{myLemma}

\begin{proof}
    Instead of directly constructing the $\varphi_n$ on the whole sieve we construct tests versus a representative point and show that these tests works well in the neighborhood of the representative points.
We then take the supremum of these tests and show that the number of pieces needed to cover the entire sieve can be appropriately bounded.

For a representative point $f_1$, consider the Neyman-Pearson test for a single point alternative $H_0: f=f_0, H_1: f=f_1$, $\phi_n=I\{ f_1/f_0\ge 1 \}$. 
If the average half order R\'enyi divergence $-n^{-1}\log(\int\sqrt{f_0f_1}d\mu)\ge \epsilon^2$, we will have:
\begin{align*}
    \mathbb{E}_{f_0}(\phi_n)&\le\int_{f_1>f_0}\sqrt{f_1/f_0} f_0 d\mu\le \int \sqrt{f_1f_0}d\mu\le e^{-n\epsilon^2}\\
    \mathbb{E}_{f_1}(1-\phi_n)&\le\int_{f_0>f_1}\sqrt{f_0/f_1} f_1 d\mu\le \int \sqrt{f_0f_1}d\mu\le e^{-n\epsilon^2}\\
\end{align*}

By Cauchy-Schwarz, for any alternative $f$ we can control the Type II error rate:
\begin{align*}
    \E_f(1-\phi_n)\le \{\E_{f_1}(1-\phi_n)\}^{1/2}\{\E_{f_1}(f/f_1)^2\}^{1/2}
\end{align*}

So long as the second factor grows at most like  $e^{cn\epsilon^2}$ for some properly chosen small $c$, the full expression can be controlled.  
Thus we can consider the neighborhood around the representative point small enough so that the second factor can be actually bounded.

Consider every density with parameters satisfying
\begin{equation}
    \label{eqn:test_sets}
    \begin{aligned}
    &|||\Omega|||_2\le ||\Omega||_1\le 8C_3q,\\ 
    &||\Psi_1-\Psi||_2\le ||\Psi_1-\Psi||_1\le \frac{1}{\sqrt{2C_3np}},\\
    &|||\Omega_1-\Omega|||_2\le ||\Omega_1-\Omega||_1\le\frac{1}{8C_3n\max\{p,q\}^{3/2}} \le \frac{1}{8C_3nq^{3/2}}
\end{aligned}
\end{equation}
We show that $\E_{f_1}(f/f_1)^2$ is bounded on the above set when parameters are from the sieve $\mathcal{F}_n$. 

Denote $\Sigma_1=\Omega_1^{-1}$, $\Sigma=\Omega^{-1}$ as well as $\Sigma_1^\star=\Omega^{1/2}\Sigma_1\Omega^{1/2}$, and $\Delta_\Psi=\Psi-\Psi_1$ while $\Delta_\Omega=\Omega-\Omega_1$. 
Using the observation $\Psi\Omega^{-1}-\Psi_1\Omega_1^{-1}=(\Delta_\Psi-\Psi_1\Omega^{-1}\Delta_\Omega)\Omega^{-1},$ we have
\begin{equation}
    \label{eqn:test_f_over_f1}
    \begin{aligned}
        \E_{f_1}(f/f_1)^2=&|\Sigma_1^\star|^{n/2}|2I-\Sigma_1^{\star-1}|^{-n/2}\\
        &\times \exp\left(\sum_{i=1}^nX_i(\Psi\Omega^{-1}-\Psi_1\Omega_1^{-1})\Omega^{1/2}(2\Sigma_1^\star-I)^{-1}\Omega^{1/2}(\Psi\Omega^{-1}-\Psi_1\Omega_1^{-1})^{\top}X_i^{\top}\right)\\
        =&|\Sigma_1^\star|^{n/2}|2I-\Sigma_1^{\star-1}|^{-n/2}\\
        &\times \exp\left(\sum_{i=1}^nX_i(\Delta_\Psi-\Psi_0\Omega^{-1}\Delta_\Omega)\Omega^{-1/2}(2\Sigma_1^\star-I)^{-1}\Omega^{-1/2}(\Delta_\Psi-\Psi_0\Omega^{-1}\Delta_\Omega)^{\top}X_i^{\top}\right)
    \end{aligned}
\end{equation}

Since $\Omega\in\mathcal{F}_n$, we have $|||\Omega^{-1}|||_2\le 1/\tau $.
The fact $|||\Omega_1-\Omega|||_2\le \delta_n'= 1/8C_3nq^{3/2}$ implies
\begin{align*}
    |||\Sigma_1^\star-I|||_2\le |||\Omega^{-1}|||_2|||\Omega_1-\Omega|||_2\le \delta_n'/\tau
\end{align*}
and thus we can bound the spectrum of $\Sigma_1^\star$, i.e. $1-\delta_n'/\tau\le \eig_1(\Sigma_1^\star)\le \eig_q(\Sigma_1^\star)\le  1+\delta_n'/\tau $.

Thus
\begin{align*}
    \left(\frac{|\Sigma_1^\star|}{|2I-\Sigma_1^{\star-1}|}\right)^{n/2}&=\exp\left(\frac{n}{2}\sum_{i=1}^q\log(\eig_i(\Sigma_1^\star))-\frac{n}{2}\sum_{i=1}^q \log\left(2-\frac{1}{\eig_i(\Sigma_1^\star)}\right)\right)\\
    &\le \exp\left(\frac{nq}{2}\log(1+\delta_n'/\tau)-\frac{nq}{2}\log\left(1-\frac{\delta_n'/\tau}{1-\delta_n'/\tau}\right)\right)\\
    &\le \exp\left(\frac{nq^2}{2}\delta_n'+\frac{nq}{2}\left(\frac{\delta_n'/\tau}{1-2\delta_n'/\tau}\right)\right)\\
    &\le \exp(nq\delta_n'/\tau)\\
    &\le e
\end{align*}
The third inequality is due to the fact $1-x^{-1}\le \log(x)\le x-1$.

We can bound the log of the second factor of Equation~\eqref{eqn:test_f_over_f1}.
\begin{align*}
    |||\Omega^{-1}|||_2|||(2\Sigma_1^\star-I)^{-1}|||_2\sum_{i=1}^n||X_i(\Delta_\Psi-\Psi_1\Omega^{-1}\Delta_\Omega)||_2^2\le 2/\tau \sum_{i=1}^n||X_i(\Delta_\Psi-\Psi_1\Omega^{-1}\Delta_\Omega)||_2^2
\end{align*}

We can further bound the sum on the sieve.
\begin{align*}
    \sum_{i=1}^n||X_i(\Delta_\Psi-\Psi_1\Omega^{-1}\Delta_\Omega)||_2^2&\le 2\sum_{i=1}^n ||X_i\Delta_\Psi||_2^2+2\sum_{i=1}^n ||X_i\Psi_1\Omega^{-1}\Delta_\Omega||_2^2\\
    &\le 2np|||\Delta_\Psi|||_2^2+2np|||\Psi_1|||_2^2|||\Omega^{-1}|||_2^2||\Delta_\Omega||_F^2\\
    &\le 2np\frac{1}{2C_3np}+2np\left(2C_3p+\frac{1}{\sqrt{2C_3np}}\right)^2\frac{1}{\tau^2}\frac{1}{(8C_3n\max\{p,q\}^{3/2})^2}\\
    &\le 2np\frac{1}{2C_3np}+2np16C_3^2p^2\frac{1}{\tau^2}\frac{1}{(8C_3n\max\{p,q\}^{3/2})^2}\\
    &\lesssim 1
\end{align*}

We bound the norm of $\Psi_1$ using triangle inequality, $|||\Psi_1|||\le |||\Psi|||+|||\Psi_1-\Psi|||\le 2C_3p+1/\sqrt{2C_3np}$. 
The first term is $O(1)$ and second term is $O(1/q)$, by combining the result we conclude the second factor of Equation~\eqref{eqn:test_f_over_f1} is bounded. 

Thus, 
the desired test $\varphi_n$ in Equation~\eqref{eqn:test_condition} can be obtained as the maximum of all tests $\phi_n$ described above.

To finish the proof we show number such tests needed to maximum over is bounded using a covering argument. Formally, we show the number of sets described in Equation~\eqref{eqn:test_sets} needed to cover sieve $\mathcal{F}_n$, denoted by $N_*$, can be bounded by $\exp(Cn\epsilon_n^2)$ for some suitable constant $C$. 

we can bound the logarithm of the covering number $\log(N_*)$.
\begin{align*}
    \log(N_*)\le & \log\left[N\left(\frac{1}{\sqrt{2C_3np}},\{\Psi\in \mathcal{B}^\Psi_n:||\Psi||_1\le2C_3p\},||\cdot||_1\right)\right]\\
    &+\log\left[N\left(\frac{1}{8C_3n\max\{p,q\}^{3/2}},\{\Omega \in \mathcal{B}^\Omega_n,||\Omega||_1\le 8C_3q\},||\cdot||_1\right)\right]
\end{align*}

The two terms above can be treated in a similar way. Denote $\max\{p,q,s_0^B,s_0^\Omega\}=s^\star$. There are multiple ways to allocate the effective 0's, which introduces the binomial coefficient below: \begin{align*}
    &N\left(\frac{1}{8C_3n\max\{p,q\}^{3/2}},\{\Omega \in \mathcal{B}^\Omega_n,||\Omega||_1\le 8C_3q\},||\cdot||_1\right)\\
    \le & \binom{Q}{C_3's^\star} N\left(\frac{1}{8C_3n\max\{p,q\}^{3/2}},\{V\in \mathbb{R}^{Q+q}:|v_i|<\delta_\omega\text { for $1\le i\le Q+q-C_3's^\star$},||V||_1\le 8C_3q\},||\cdot||_1\right)\\
    &N\left(\frac{1}{\sqrt{2C_3np}},\{\Psi\in \mathcal{B}^\Psi_n:||\Psi||_1\le2C_3p\},||\cdot||_1\right)\\
    \le & \binom{pq}{C_3' s^\star}  N\left(\frac{1}{\sqrt{2C_3np}},\{V\in \mathbb{R}^{pq}:|v_i|<\delta_\psi\text { for $1\le i\le pq-C_3's^\star$},||V||_1\le 2C_3p\},||\cdot||_1\right)\\
\end{align*}

Note that $\Omega$ has $Q + q < 2Q$ free parameters.
We have first
\begin{align*}
    \log \binom{Q}{C_3's^\star}&\lesssim s^\star \log(Q)\lesssim n\epsilon_n^2 \\
    \log \binom{pq}{C_3's^\star}&\lesssim s^\star \log(pq)\lesssim n\epsilon_n^2 \\
\end{align*}

We further bound the covering number using the result in Lemma \ref{lemma:packing_number_lp}. Observe that $||V||_1\cap \{|v_i|<\delta_\omega \text{ for } 1\le i\le Q+q-C_3's^\star\}\subset \{||V'||\le 8C_3q\}\times[-\delta_\omega,\delta_\Omega]^{Q+q-C_3's^\star}$ where $V'\in \mathbb{R}^{C_3's^\star}$ we have
\begin{align*}
    &N\left(\frac{1}{8C_3n\max\{p,q\}^{3/2}},\{V:|v_i|<\delta_\omega\text { for $1\le i\le Q+q-C_3's^\star$},||V||_1\le 8C_3q\},||\cdot||_1\right)\\
    \le & N\left(\frac{1}{8C_3n\max\{p,q\}^{3/2}},\{V\in \mathbb{R}^{C_3's^\star}:||V'||_1\le 8C_3q\times [-\delta_\omega,\delta_\omega]^{Q+q-C_3's^\star}\},||\cdot||_1\right)\\
\end{align*}

We check the condition of Lemma~\ref{lemma:packing_number_lp} (with $p=1$ and $T=2$), by our assumption on $\xi_0$, we have:
\begin{align*}
    (Q+q-C_3's^\star)\delta_\omega&\le 2Q\delta_\omega = 2Q\frac{1}{\xi_0-\xi_1}\log\left[\frac{1-\eta}{\eta}\frac{\xi_0}{\xi_1}\right]\lesssim  \frac{Q\log(\max\{p,q,n\})}{\max\{Q,pq,n\}^{4+b/2+b/2}}\\
    &\le \frac{1}{\max\{Q,pq,n\}^{3+b/2}}
\end{align*}

The denominator dominates $ C_3n\max\{p,q\}^{3/2}$ thus for large enough $n$, we have $(Q+q-C_3's^\star)\delta_\omega\le \frac{1}{32C_3n\max\{p,q\}^{3/2}}$ thus by Lemma \ref{lemma:packing_number_lp}, we can control the covering number by the packing number:
\begin{align*}
    &
    \log N\left(\frac{1}{8C_3n\max\{p,q\}^{3/2}},\{V:|v_i|<\delta_\omega\text { for $1\le i\le Q+q-C_3's^\star$},||V||_1\le 8C_3q\},||\cdot||_1\right)\\
    \le &\log D\left(\frac{1}{16C_3n\max\{p,q\}^{3/2}},\{V'\in \mathbb{R}^{C_3's^\star},||V'||_1\le 8C_3q\},||\cdot||_1\right)\\
    \lesssim & s^\star \log(128C_3^2qn\max\{p,q\}^{3/2})\\
    \lesssim& n\epsilon_n^2
\end{align*}

Similarly for $\Psi$,
\begin{align*}
    &N\left(\frac{1}{\sqrt{2C_3np}},\{V:|v_i|<\delta_\psi\text { for $1\le i\le pq-C_3's^\star$},||V||_1\le 2C_3p\},||\cdot||_1\right)\\
    \le& N\left(\frac{1}{\sqrt{2C_3np}},\{V'\in \mathbb{R}^{C_3's^\star}:||V'||_1\le 2C_3p\times [-\delta_\psi,\delta_\psi]^{pq-C_3's^\star}\},||\cdot||_1\right)\\
\end{align*}

We again check the condition of Lemma~\ref{lemma:packing_number_lp} (again with $p=1$ and $T=2$):
\begin{align*}
    (pq-C_3's^\star )\delta_\psi&\le pq\delta_\psi=\frac{pq}{\lambda_0-\lambda_1}\log\left[\frac{1-\theta}{\theta}\frac{\lambda_0}{\lambda_1}\right]\lesssim \frac{pq\log(\max\{p,q,n\})}{\max\{pq,n\}^{5/2+b/2+b/2}}\\
    &\le \frac{1}{\max\{pq,n\}^{3/2+b/2}}
\end{align*}

The denominator dominates $\sqrt{2C_3np}$, Thus for enough large $n$, we have $(pq-C_3's^\star )\delta_\psi\le 1/4\sqrt{2C_3np}$. Thus similar to $\Omega$, we have:
\begin{align*}
    &
    \log N\left(\frac{1}{\sqrt{2C_3np}},\{V:|v_i|<\delta_\omega\text { for $1\le i\le pq-C_3's^\star$},||V||_1\le 2C_3p\},||\cdot||_1\right)\\
    \le &\log D\left(\frac{1}{2\sqrt{2C_3np}},\{V'\in \mathbb{R}^{C_3's^\star},||V'||_1\le 2C_3p\},||\cdot||_1\right)\\
    \lesssim & s^\star \log(4C_3p\sqrt{2C_3np})\\
    \lesssim& n\epsilon_n^2
\end{align*}

\end{proof}

Contraction of log-affinity follows Lemma~\ref{lemma:test_condition} and Lemma~\ref{lemma:KL_cg}. 

\subsection{From log-affinity to $\Omega$ and $X\Psi\Omega^{-1}$}
\label{sec:log_affinity_to_sieve}

In this section we prove Theorem 1 of the main text using the contraction under log-affinity.

\begin{proof}
    Denoting $\Psi-\Psi_0=\Delta_\Psi$ and $\Omega-\Omega_0=\Delta_\Omega$ we have the log-affinity $\frac{1}{n}\sum \rho(f_i,f_{0i})$ is
\begin{align*}
    \frac{1}{n}\sum \rho(f_i,f_{0i})=&-\log\left(\frac{|\Omega^{-1}|^{1/4}|\Omega_0^{-1}|^{1/4}}{|(\Omega^{-1}+\Omega_0^{-1})/2|^{1/2}}\right)\\
    &+\frac{1}{8n}\sum X_i(\Psi\Omega^{-1}-\Psi_0\Omega_0^{-1})\left(\frac{\Omega^{-1}+\Omega_0^{-1}}{2}\right)^{-1}(\Psi\Omega^{-1}-\Psi_0\Omega_0^{-1})^{\top}X_i^{\top}
\end{align*}

Thus $\sum \rho(f_i-f_{0i})\lesssim n\epsilon_n^2$ implies
\begin{equation}
    \label{eqn:log_affinity_implies_some_bound}
    \begin{aligned}
        -\log\left(\frac{|\Omega^{-1}|^{1/4}|\Omega_0^{-1}|^{1/4}}{|(\Omega^{-1}+\Omega_0^{-1})/2|^{1/2}}\right)&\lesssim \epsilon_n^2\\
        \frac{1}{8n}\sum X_i(\Psi\Omega^{-1}-\Psi_0\Omega_0^{-1})\left(\frac{\Omega^{-1}+\Omega_0^{-1}}{2}\right)^{-1}(\Psi\Omega^{-1}-\Psi_0\Omega_0^{-1})^{\top}X_i^{\top}&\lesssim \epsilon_n^2\\
    \end{aligned}
\end{equation}

We can directly apply the result from \citet{Ning2020}'s Equation 5.11 as it is the same as the first equation in Equation~\eqref{eqn:log_affinity_implies_some_bound}. 
Because $\Psi_{0}$ and $\Omega^{-1}$ have bounded operator norms and because $\Delta_{\Omega}$ can be controlled, the cross-term is also controlled by $\epsilon_{n}.$
The first part of Equation~\eqref{eqn:log_affinity_implies_some_bound} implies $||\Omega^{-1}-\Omega_0^{-1}||_F^2\lesssim \epsilon_n^2$. 

Meanwhile $||\Omega^{-1}-\Omega_0^{-1}||_F^2\lesssim \epsilon_n^2$ implies for large enough $n$ $\Omega$'s L2 operator norm is bounded (since we assume bounds on $\Omega_0^{-1}$'s operator norm and the difference cannot have very large eigenvalues which make the sum has 0 eigenvalue), using the result $||AB||_F\le |||A|||_2||B||_F$, while also observe $\Omega_0-\Omega=\Omega(\Omega^{-1}-\Omega_0^{-1})\Omega_0$, and by assumption that $\Omega_0$ has bounded L2 operator norm, we conclude~\eqref{eqn:log_affinity_implies_some_bound} implies $||\Omega-\Omega_0||_F\lesssim \epsilon_n$.

Since $|||\Omega^{-1}|||_2$ is bounded for large enough $n$, we can 
conclude ~\eqref{eqn:log_affinity_implies_some_bound}'s second part implies:
\begin{align*}
    \epsilon_n^2&\gtrsim \frac{1}{8n}\sum ||X_i(\Psi\Omega^{-1}-\Psi_0\Omega^{-1}_0)||_2^2|||\frac{\Omega^{-1}+\Omega_0^{-1}}{2}|||^{-1}_2\\
    &\gtrsim \frac{1}{8n}\sum ||X_i(\Psi\Omega^{-1}-\Psi_0\Omega^{-1}_0)||_2^2|||\frac{\Omega^{-1}+\Omega_0^{-1}}{2}|||^{-1}_2\\
    &\gtrsim\frac{1}{n}\sum ||X_i(\Psi\Omega^{-1}-\Psi_0\Omega^{-1}_0)||_2^2/\sqrt{\epsilon_n^2+1}
\end{align*}
Combining all of these we have the contraction of $\Omega$ and $X\Psi $
\end{proof}

\subsection{Contraction of $\Psi$}

We give the proof of Corollary 1.

Contraction of $\Psi$ requires more assumptions on the design matrix $X$.
Similar to \citet{RockovaGeorge2018_ssl} and \citet{Ning2020}, we introduce the restricted eigenvalue
\begin{align*}
    \phi^2(\tilde{s})=\inf \left\{\frac{||XA||^2_F}{n||A||_F^2}:0\le |\nu(A)|\le \tilde{s}\right\}
\end{align*}

With this definition, 

\begin{proof}
    \begin{align*}
        ||X(\Psi\Omega^{-1})-\Psi_0\Omega^{-1}_0)||_F^2&\lesssim n\epsilon_n^2\\
        ||\Omega-\Omega_0||_F^2\lesssim \epsilon_n^2
    \end{align*}
    implies the result in Equation (4.5) of the main text.
    Namely,
    \begin{align*}
        ||\Psi\Omega^{-1}-\Psi_0\Omega^{-1}_0||_F^2=||(\Delta_\Psi-\Psi_0\Omega^{-1}\Delta_\Omega)\Omega^{-1}||_F^2\lesssim \epsilon_n^2/\phi^2(s_0^\Psi+C_3's^\star)
    \end{align*}
    
    Since both $\Omega$ and $\Omega^{-1}$ have bounded operator norm when $||\Omega-\Omega_0||_F^2\lesssim \epsilon_n^2,$ for large enough $n$, we must have:
    \begin{align*}
        ||\Delta_\Psi||_F-||\Psi_0\Omega^{-1}\Delta_\Omega||_F\le ||\Delta_\Psi-\Psi_0\Omega^{-1}\Delta_\Omega||_F\lesssim \epsilon_n/\sqrt{\phi^2(s_0^\Psi+C_3's^\star)}
    \end{align*}
    
    Since $\Psi_0$ and $\Omega^{-1}$ have bounded operator norm, $||\Psi_0\Omega^{-1}\Delta_\Omega||_F\lesssim \epsilon_n$, and we must have:
    \begin{align*}
        ||\Delta_\Psi||_F\lesssim \epsilon_n/\sqrt{\min\{\phi^2(s_0^\Psi+C_3's^\star),1\}}
    \end{align*}
    
    Thus we can conclude
    \begin{align*}
        \sup_{\Psi\in \mathcal{T}_0,\Omega\in\mathcal{H}_0}\E_0 \Pi\left(||\Psi-\Psi_0||_F^2\ge \frac{M'\epsilon_n^2}{\min\{\phi^2(s_0^\Psi+C_3's^\star),1\}}\right)\to 0
    \end{align*}
    
\end{proof}

\section{Bayesian Bootstrap}
\label{appendix:BB}
\subsection{Bayesian bootstrap cgSSL}

For likelihood reweighting Bayesian bootstrap, we directly reweight $Y^*_i=\sqrt{w_i^t}Y_i$, $X_i^*=\sqrt{w_i^t}X_i$ and use the reweighted data in the cgSSL algorithm to generate BB samples. 

Similarly, we defined the jittered prior on $\Psi$ same as in \citet{Nie2020}.

\begin{align}
    \begin{split}
        \pi(\Psi\vert \theta,\mu )=\left(\prod_{j,k}\left[\frac{\theta\lambda_1}{2}e^{-\lambda_1|\psi_{j,k}-\mu_{j,k}|}+\frac{(1-\theta)\lambda_0}{2}e^{-\lambda_0|\psi_{j,k}-\mu_{j,k}|}\right]\right) 
    \end{split}
\end{align}

To fit the jittered version, we work with the slightly transformed parameter $\Psi^*,\Omega$ where $\Psi^*$ was defined as $\Psi^*=\Psi-\mu_t$, with reweighted data $Y^*=\sqrt{w_i^t}Y_i$ and $X_i^*=\sqrt{w_i^t}X_i$. We fit the model with the transformed data using the original $\Psi$ updating procedure in the original cgSSL with transformed data $Y^*\Omega-(x^*_i)^\top\mu_t$. We then update $\Omega$ using modified cgQUIC after the E step. We post process to get $\Psi = \Psi^*+\mu_t$. We also need to modify the QUIC algorithm to the jittered QUIC (Algorithm~\ref{algo:jitQUIC}) to account for the jittering.

\begin{algorithm}[htp]
	\caption{The jittered QUIC algorithm}
	\label{algo:jitQUIC}
	\KwData{$S=Y^{\top}Y/n$, the jittering matrix $J$, regularization parameter matrix $\Xi$, initial $\Omega_0$, inner stopping tolerance $\epsilon$, parameters $0 < \sigma < 0.5$, $0 < \beta < 1$}
	\KwResult{path of positive definite $\Omega_t$ that converge to $\argmin_\Omega f$ with $f(\Omega)=g(\Omega) +\sum_{k,k'}{\xi_{k,k'}\lvert \omega_{k,k'}-J_{k,k'}\rvert}$, where $g(\Omega)=-\log(\lvert\Omega\rvert) + \tr(S\Omega)$}
	
	Initialize $W_0=\Omega_0^{-1}$;
	
	\For{$t=1,2,\dots$}{
		$D = 0, U = 0$\;
		\While{not converged}{
			Partition the variables into fixed and free sets based on gradient\footnote{For compactness, we omit the subscript $t-1$ for $\Omega$ and $W$} \\
			$S_{fixed}:= \{(k,k'):|\nabla_{k,k'} g(\Omega)|<\xi_{k,k'}\text{ and } \omega_{k,k'}-J_{k,k'}=0\}$\;
			$S_{free}:= \{(k,k'):|\nabla_{k,k'} g(\Omega)|\ge\xi_{k,k'}\text{ or } \omega_{k,k'}-J_{k,k'}\ne0\}$\;
			\For{$(k,k')\in S_{free}$}{
				Calculate Newton direction:\\
				$b=S_{k,k'}-W_{k,k'}+w_k^{\top}Dw_{k'}$\;
				$c=\Omega_{k,k'}+D_{k,k'}-J_{k,k'}$\;
				\If{$i\ne j$}{
					$a=W_{k,k'}^2+W_{k,k}W_{k',k'}$\;
				}
				\Else{
					$a=W_{k,k}^2$\;
				}
				$\mu=-c+[\lvert c-b/a\rvert-\xi_{k,k'}/a]_+ \sign(c-b/a)$ \;
				$D_{k,k'}+=\mu$\;
				$\boldsymbol{u}_k+=\mu \boldsymbol{w}_{k'}$\;
				$\boldsymbol{u}_{k'}+=\mu \boldsymbol{w}_{k}$\;
				
			}
		}
		\For{$\alpha=1,\beta,\beta^2,\dots$}{
			Compute the Cholesky decomposition of $\Omega_{t-1}+\alpha D$\;
			\If{$\Omega_t$ is not positive definite}{
				continue\;
			}
			Compute $f(\Omega_{t-1}+\alpha D)$\;
			\If{$f(\Omega_{t-1}+\alpha D^*)\le f(\Omega_{t-1})+\alpha \sigma\delta, \delta = \tr[\nabla g(\Omega_{t-1})^{\top} D^*] + ||\Omega_{t-1}+D^*-J||_{1,\Xi}-||\Omega_{t-1}-J||_{1,\Xi}$}{
				break\;
			}
		}
	$\Omega_t = \Omega_{t-1}+\alpha D$\;
	$W_{t}=\Omega_t^{-1}$ using Cholesky decomposition result\;
	}
return $\{\Omega_t\}$\;
\end{algorithm}

\subsection{Results on simulation}

\begin{figure}[h]
	\centering
	\label{fig:BB_10_10_100}
	\includegraphics[width = \textwidth]{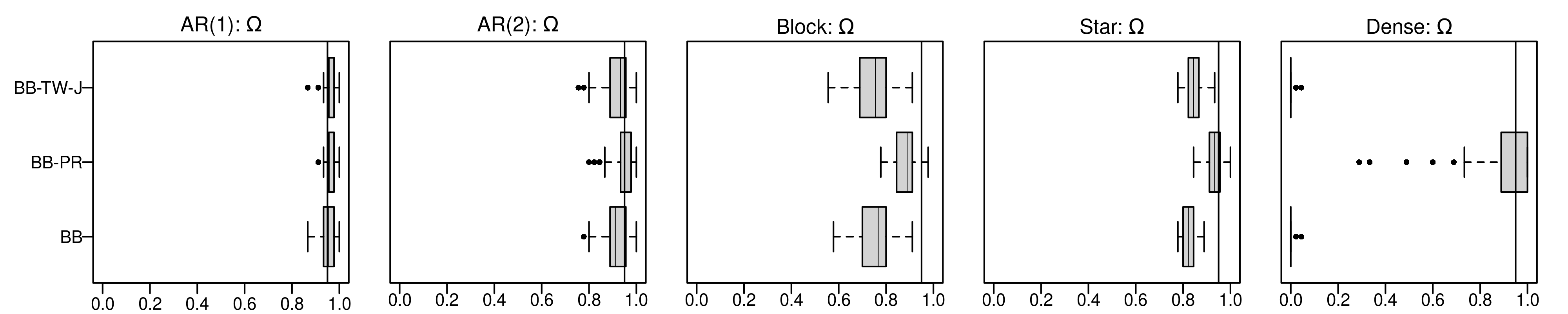}
	\includegraphics[width = \textwidth]{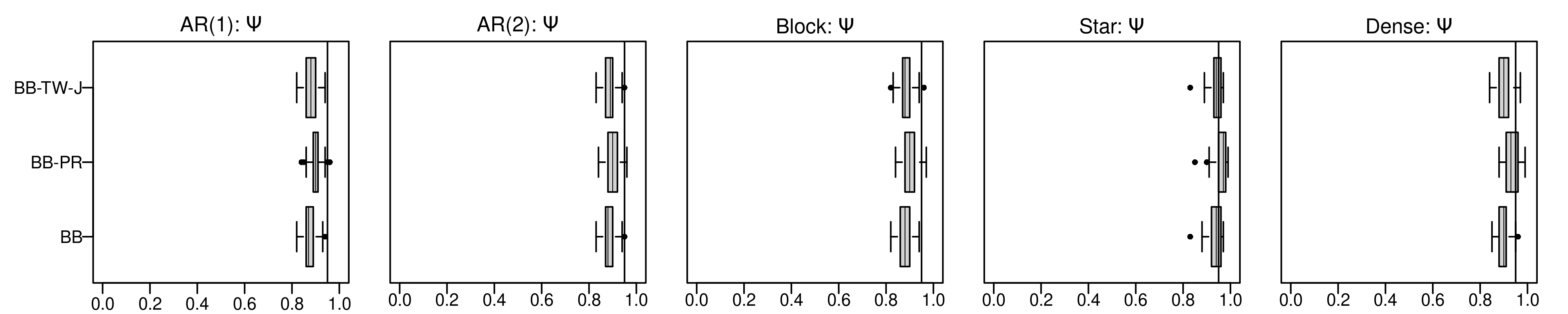}
	\caption{Coverage of three BB procedures models with $p=10,q=10,n=100$.}
\end{figure}

\begin{figure}[h]
	\centering
	\label{fig:BB_20_30_100}
	\includegraphics[width = \textwidth]{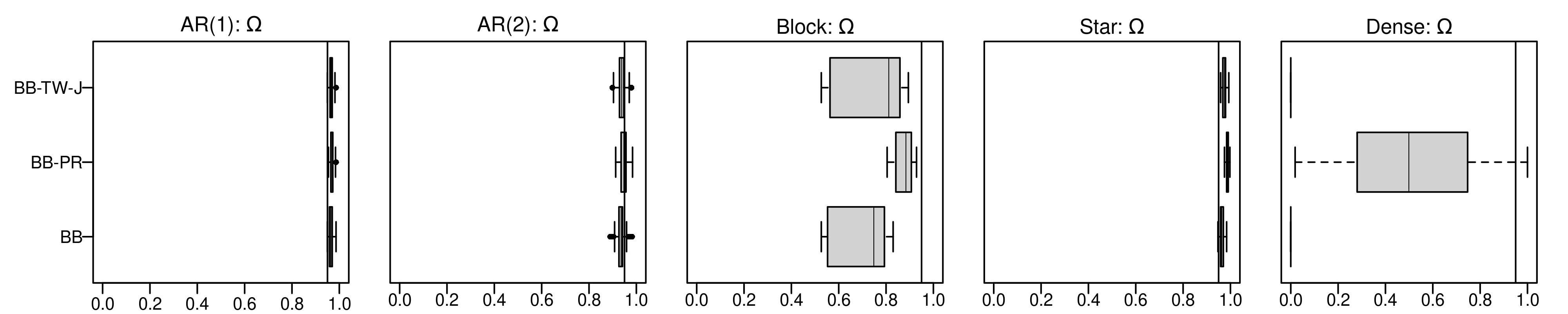}
	\includegraphics[width = \textwidth]{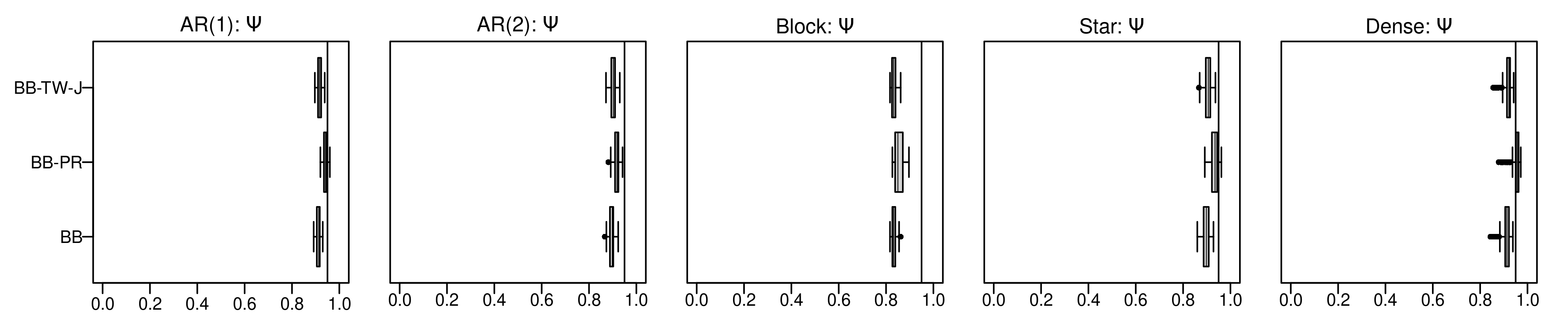}
	\caption{Coverage of three BB procedures models with $p=20,q=30,n=100$.}
\end{figure}

\begin{figure}[h]
	\centering
	\label{fig:BB_100_30_400}
	\includegraphics[width = \textwidth]{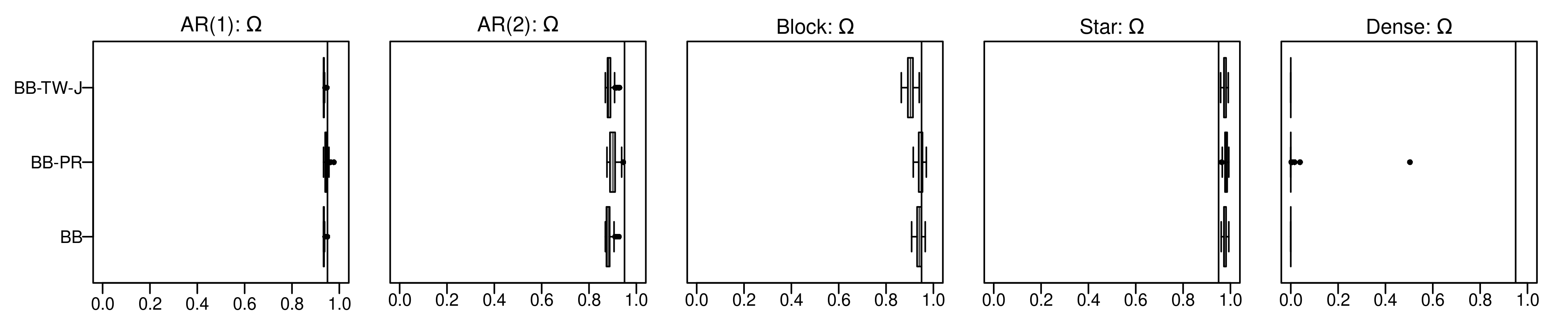}
	\includegraphics[width = \textwidth]{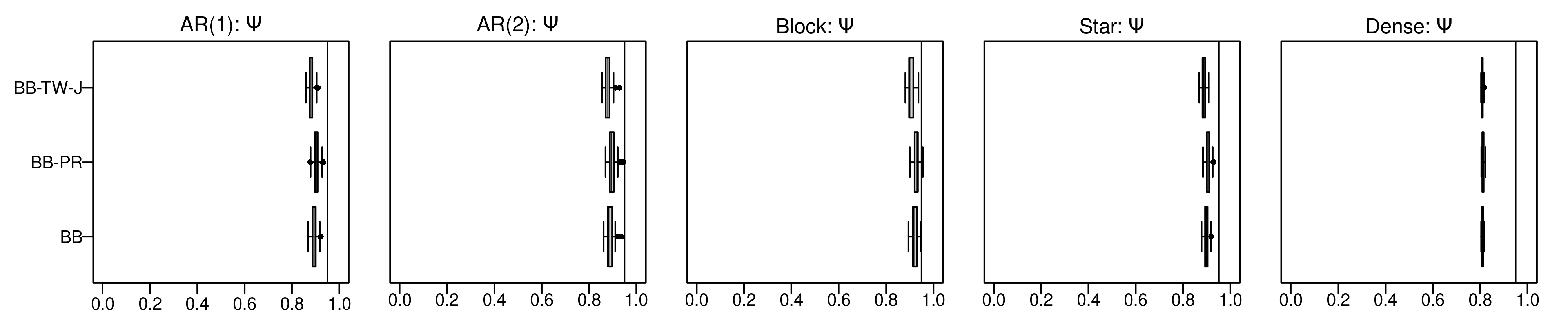}
	\caption{Coverage of three BB procedures models with $p=100,q=30,n=400$.}
\end{figure}

\begin{figure}[h]
	\centering
	\label{fig:BB}
	\includegraphics[width = \textwidth]{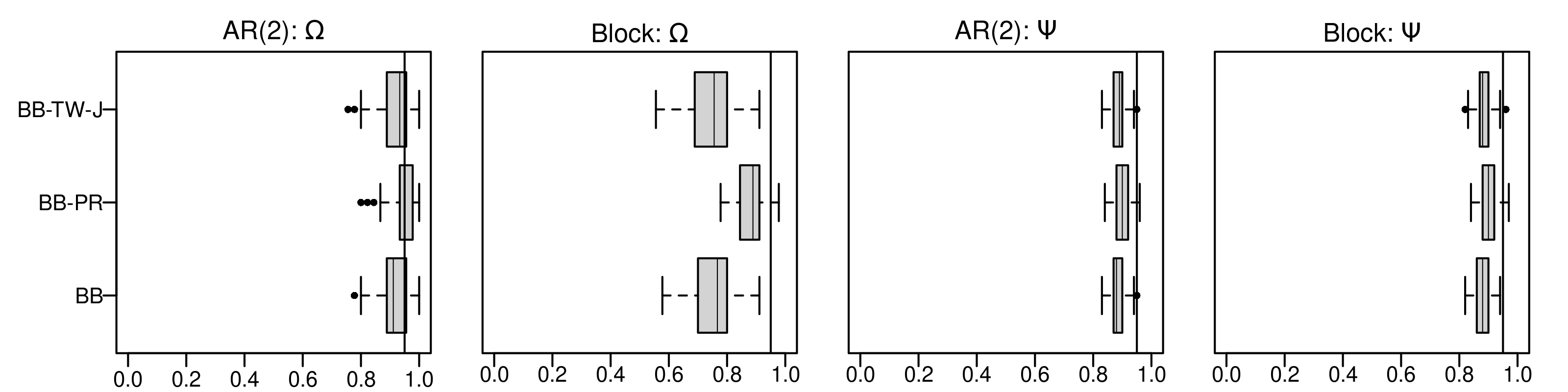}
	\caption{Coverage of three BB procedures in AR(2) and Block models with $p=10,q=10,n=100$.}
\end{figure}

\end{document}